\documentclass[a4paper,twocolumn,11pt,accepted=2024-01-03]{quantumarticle}
\pdfoutput=1
\usepackage[utf8]{inputenc}
\usepackage[english]{babel}
\usepackage[T1]{fontenc}
\usepackage{amsmath}
\usepackage{hyperref}

\usepackage{graphicx, color, graphpap}
\usepackage{float}
\usepackage{enumitem}
\usepackage{amssymb}
\usepackage{amsthm}
\usepackage{bbold}
\usepackage{multirow}
\usepackage[dvipsnames]{xcolor}
\usepackage{bbm}
\usepackage{thmtools,thm-restate}
\usepackage{verbatim}
\usepackage{mathtools}
\usepackage{titlesec}
\usepackage{dsfont}
\usepackage[caption = false]{subfig}
\usepackage{subfig}
\usepackage[normalem]{ulem}
\usepackage{cite}

\usepackage{tikz}
\usepackage{lipsum}

\newcommand{\BC}{\mathcal{B}}

\newcommand{\HC}{\mathcal{H}}

\newcommand{\NC}{\mathcal{N}}
\newcommand{\OC}{\mathcal{O}}

\newcommand{\SC}{\mathcal{S}}

\newcommand{\UC}{\mathcal{U}}

\newcommand{\WC}{\mathcal{W}}

\renewcommand{\vec}[1]{\boldsymbol{#1}}
\newcommand{\thv}{\vec{\theta}}

\DeclareMathOperator*{\argmin}{arg\,min}
\newcommand{\Tr}{{\rm Tr}}

\newcommand{\Var}{{\rm Var}}
\newcommand{\Cov}{{\rm Cov}}
\newcommand{\poly}{\operatorname{poly}}
\newcommand{\id}{\mathbb{1}}
\DeclareRobustCommand{\rchi}{{\mathpalette\irchi\relax}}
\newcommand{\irchi}[2]{\raisebox{0.6\depth}{$#1\chi$}}
\makeatletter
\newcommand{\dbloverline}[1]{\overline{\overline{#1}}}
\makeatother

\newcommand{\ket}[1]{|#1\rangle}               
\newcommand{\bra}[1]{\langle #1|}              

\newtheorem{theorem}{Theorem}
\newtheorem{lemma}{Lemma}
\newtheorem{remark}{Remark}
\newtheorem{supproposition}{Supplemental Proposition}
\newtheorem{corollary}{Corollary}
\newtheorem{proposition}{Proposition}

\newtheorem{definition}{Definition}

\begin{document}

\title{Can Error Mitigation Improve Trainability of Noisy Variational Quantum Algorithms?}

\author{Samson Wang}
\email{samsonwang@outlook.com}
\affiliation{Theoretical Division, Los Alamos National Laboratory, Los Alamos, NM 87545, USA}
\affiliation{Department of Physics, Imperial College London, London, SW7 2AZ, UK}
\orcid{0000-0003-2344-0634}

\author{Piotr Czarnik}
\affiliation{Theoretical Division, Los Alamos National Laboratory, Los Alamos, NM 87545, USA}
\affiliation{Faculty of Physics, Astronomy, and Applied Computer Science, Jagiellonian University, Krak\'ow, Poland}
\affiliation{Mark Kac Center for Complex Systems Research, Jagiellonian University, Krak\'ow, Poland}
\orcid{0000-0002-0477-1158}

\author{Andrew Arrasmith}
\affiliation{Theoretical Division, Los Alamos National Laboratory, Los Alamos, NM 87545, USA}
\affiliation{Quantum Science Center, Oak Ridge, TN 37931, USA}
\orcid{0000-0003-2674-9370}

\author{M. Cerezo}
\affiliation{Theoretical Division, Los Alamos National Laboratory, Los Alamos, NM 87545, USA}
\affiliation{Quantum Science Center, Oak Ridge, TN 37931, USA}
\affiliation{Center for Nonlinear Studies, Los Alamos National Laboratory, Los Alamos, NM 87545, USA
}
\orcid{0000-0002-2757-3170}

\author{Lukasz Cincio}
\affiliation{Theoretical Division, Los Alamos National Laboratory, Los Alamos, NM 87545, USA}
\affiliation{Quantum Science Center, Oak Ridge, TN 37931, USA}
\orcid{0000-0002-6758-4376}

\author{Patrick J. Coles}
\affiliation{Theoretical Division, Los Alamos National Laboratory, Los Alamos, NM 87545, USA}
\affiliation{Quantum Science Center, Oak Ridge, TN 37931, USA}
\orcid{}

\maketitle

\begin{abstract}
  Variational Quantum Algorithms (VQAs) are often viewed as the best hope for near-term quantum advantage. However, recent studies have shown that noise can severely limit the trainability of VQAs, e.g., by exponentially flattening the cost landscape and suppressing the magnitudes of cost gradients. Error Mitigation (EM) shows promise in reducing the impact of noise on near-term devices. Thus, it is natural to ask whether EM can improve the trainability of VQAs. In this work, we first show that, for a broad class of EM strategies, exponential cost concentration cannot be resolved without committing exponential resources elsewhere. This class of strategies includes as special cases Zero-Noise Extrapolation, Virtual Distillation, Probabilistic Error Cancellation, and Clifford Data Regression. Second, we perform analytical and numerical analysis of these EM protocols, and we find that some of them (e.g., Virtual Distillation) can make it harder to resolve cost function values compared to running no EM at all. As a positive result, we do find numerical evidence that Clifford Data Regression (CDR) can aid the training process in certain settings where cost concentration is not too severe. Our results show that care should be taken in applying EM protocols as they can either worsen or not improve trainability. On the other hand, our positive results for CDR highlight the possibility of engineering error mitigation methods to improve trainability.
\end{abstract}

\section{Introduction}

The prospect of obtaining quantum computational advantage for practical problems, such as simulating systems in chemistry and materials science, has generated much excitement. The past few years have witnessed tremendous progress towards this end, with significant focus on algorithm development for Noisy Intermediate-Scale Quantum (NISQ) computers. In particular, Variational Quantum Algorithms (VQAs) are a leading algorithmic approach because they adapt to the constraints of NISQ devices. Specifically, VQAs minimize a cost function by training a parameterized quantum circuit via a classical-quantum feedback loop \cite{mcclean2016theory, cerezo2020variationalreview}. The cost is computed efficiently on a quantum computer whilst the parameter optimization is carried out classically. Different implementations of this versatile framework have been proposed for a broad spectrum of problems from dynamical quantum simulation \cite{mcardle2019variational, grimsley2019adaptive, cirstoiu2020variational, commeau2020variational, gibbs2021long, yao2020adaptive, endo2020variational, li2017efficient, lau2021quantum,heya2019subspace,yuan2019theory} to machine learning \cite{schuld2020circuit, verdon2017quantum, romero2021variational, farhi2018classification, beer2020training, cong2019quantum, grant2018hierarchical} and beyond \cite{peruzzo2014variational,bauer2016hybrid,jones2019variational, farhi2014quantum,wang2018quantum,crooks2018performance,hadfield2019quantum, bravo2020variational, xu2019variational, koczor2020variational, meyer2020variational, anschuetz2019variational, khatri2019quantum,sharma2019noise,jones2018quantum, arrasmith2019variational,cerezo2020variational,larose2019variational,verdon2019quantum,johnson2017qvector}.

A central challenge in the NISQ regime is to combat the effects of noise as full error correction is not possible \cite{preskill2018quantum}. Decoherence, gate errors, and measurement noise all conspire to limit the complexity of quantum circuits that can be implemented on NISQ devices. While VQAs themselves offer some strategy to mitigate the impact of noise \cite{mcclean2016theory}, it is widely viewed that VQAs alone will not be enough, and additional strategies will be needed to obtain quantum advantage in the face of noise. This has spawned the field of error mitigation (EM), and many researchers believe that VQAs combined with EM techniques will be the path forward. Indeed, EM methods like Zero-Noise Extrapolation \cite{temme2017error,li2017efficient, endo2018practical,kandala2018error}, Clifford Data Regression \cite{czarnik2020error}, Virtual Distillation \cite{huggins2020virtual, koczor2020exponential}, Probabilistic Error Cancellation \cite{temme2017error, endo2018practical} and others \cite{mcclean2017hybrid,o2021error,mcardle2019error,bonet2018low,huggins2021efficient, barron2020measurement, smith2021qubit, su2021error} have been demonstrated to reduce errors of observable expectation values, sometimes by orders of magnitude. Hence, there has been hope that one can simply train the VQA in the presence of noise, and then after training, one can apply an EM method to extract the correct cost value (e.g., the ground state energy in the case of the variational quantum eigensolver~\cite{peruzzo2014variational}).

However, new challenges have recently been discovered for this approach~\cite{wang2020noise, franca2020limitations}. It is now recognized that noise impacts the trainability of VQAs, that is, the ability of the classical optimizer to find the global cost minimum. For ansatzes (i.e., parameterized quantum circuits) with depth linear or superlinear in the number of qubits and local Pauli noise, the cost function landscape exponentially flattens, leading to an exponentially vanishing cost gradient, a phenomenon known as Noise-Induced Barren Plateaus (NIBPs)~\cite{wang2020noise}. Thus, noise impedes the training process of VQAs, as in such a setting one requires an exponential number of shots per optimization step to resolve the cost landscape against finite sampling noise. As with other barren plateau effects~\cite{mcclean2018barren,cerezo2021cost},  this exponential scaling does not only arise for gradient-based optimizers but also impacts gradient-free methods~\cite{arrasmith2020effect} and optimizers that use higher-order derivatives~\cite{cerezo2020impact}. NIBPs represent a serious issue for VQA scalability, and could ultimately be a roadblock for near-term quantum advantage. It is therefore crucial to investigate potential methods to mitigate them.

Given the great success of EM methods in suppressing error in observable expectation values, it is natural to ask whether EM methods could address NIBPs. More generally, one could simply ask: does it help to use error mitigation during the training process for VQAs? This question is precisely the topic of our article. We remark that error mitigation has been successfully implemented during the VQA training process for a small-scale problem~\cite{kandala2018error}. However, it is an open question as to whether or not EM can resolve large-scale trainability issues associated with cost concentration. This is due to the fact that even 
\begin{figure}[tp!]
    \includegraphics[width=\columnwidth,trim={0cm 6.6cm 0cm 0cm},clip]{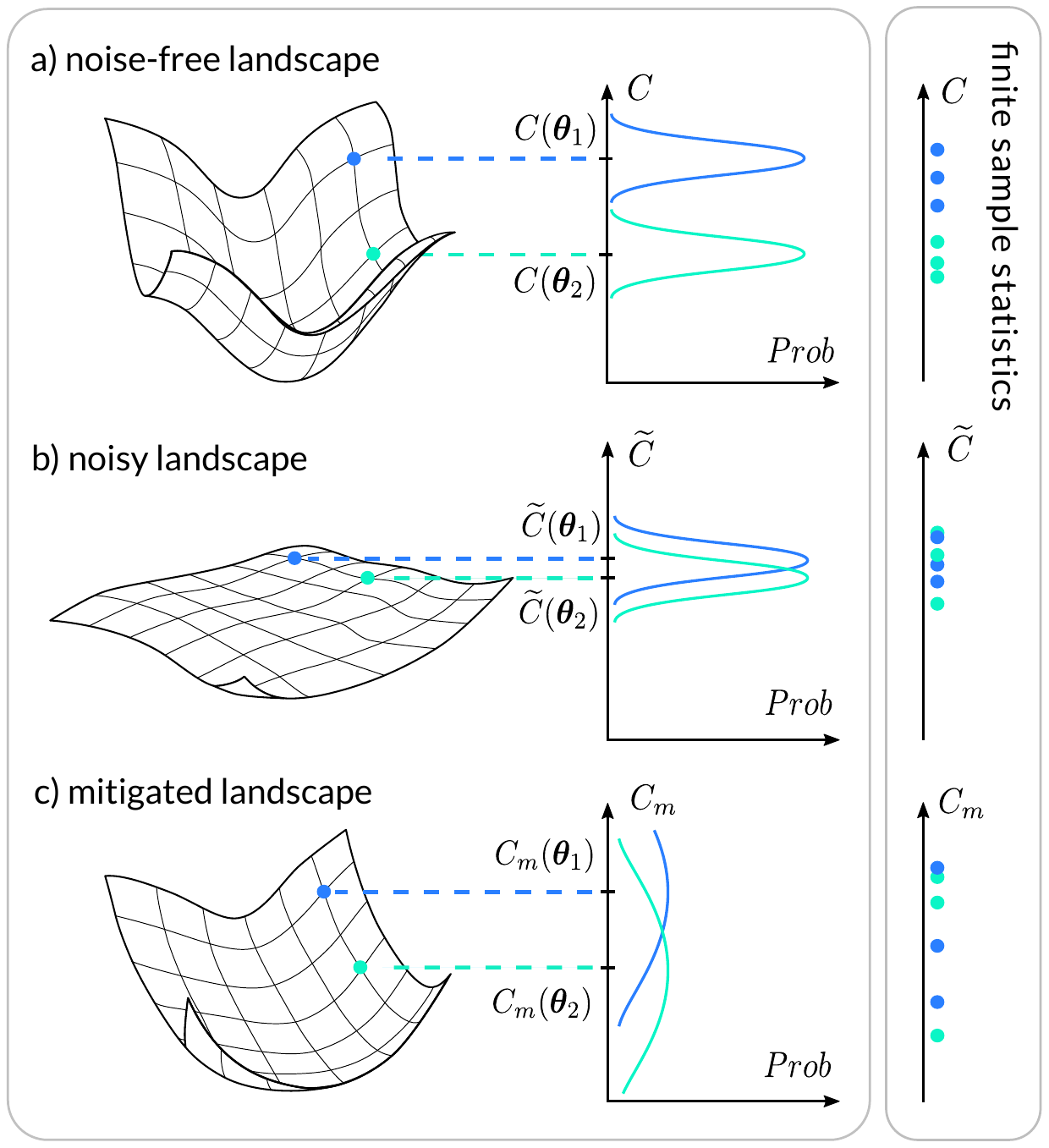}
    \caption{\textbf{Error mitigation can impair the resolvability of cost function landscapes.} (a): A central primitive in training VQAs is the task of comparing two cost function values ($C(\thv_1 )$ and $C(\thv_2 )$) on the cost landscape in parameter space. Ideally (with infinite sampling), these cost values correspond to the mean values of some probability distributions (left panel). However, in an experimental setup, one only has a finite shot budget and by collecting measurement statistics one obtains an estimate of the mean values by sampling from these distributions (right panel). (b): The effect of certain types of noise models is to concentrate cost function values. This impedes trainability as any two cost function values ($\widetilde{C}(\thv_1 )$ and $\widetilde{C}(\thv_2 )$)  have small separation and require many shots to accurately distinguish. (c): Error mitigation can mitigate many effects of noise and potentially recover key features of the noise-free cost function. In an ideal scenario, the separation of the mitigated cost values ($C_m(\thv_1 )$ and $C_m(\thv_2 )$) closely resembles that of the noise-free landscape. However, the caveat is that the variance of statistical outcomes can increase greatly. The effect of this is that the two cost function points can often require even more shots to resolve accurately, compared to the unmitigated case.}
    \label{fig:paper-summary}
\end{figure}
though EM can reverse the concentration of cost values, it also increases the statistical uncertainty in the mitigated quantities, as summarized in Figure~\ref{fig:paper-summary}. { If the statistical uncertainty increases too quickly, then error mitigation may make it harder to find cost minimizing directions, or reliably compare relative magnitudes of cost values, which are the key tasks in order to reliably train on a cost landscape. Centrally, it is a non-trivial question as to whether or not EM improves the \textit{resolvability} of cost function values.  To be clear, we are not solely quantifying the effectiveness of error mitigation in reconstructing the cost landscape given unbounded samples, which is a central question widely studied in the literature. Rather, viewing NIBPs as an exponential sample issue, in this work we compare the sample efficiency of extracting information needed for optimization over the cost landscape using error mitigation versus not using error mitigation at all.}

In this work, we investigate the effects of error mitigation on the resolvability of the cost function landscape. First, we consider a broad class of error mitigation protocols and show that, under the class of local Pauli noise that is known to cause NIBPs, in order to reverse exponential cost concentration any such protocol needs to spend resources (e.g., shot resources or number of state copies) scaling at least exponentially in the number of qubits. This suggests that NIBPs are a serious scaling issue that cannot be simply resolved with error mitigation.

Second, we study four specific error mitigation protocols in further detail: Zero-Noise Extrapolation, Virtual Distillation, Probabilistic Error Cancellation, and strategies that implement a linear ansatz which includes Clifford Data Regression. We find that Virtual Distillation can actually decrease the resolvability of the noisy cost landscape, and impede trainability. Under more restrictive assumptions on the cost landscape, we find a similar result for Zero-Noise Extrapolation. We also show that any improvement in the resolvability after applying Probabilistic Error Cancellation under local depolarizing noise exponentially degrades with increasing number of qubits. Finally, for strategies that use a linear ansatz such as Clifford Data Regression, we show that there is no change to the resolvability of any pair of cost values if the same ansatz is used. However, we do observe numerically that Clifford Data Regression increases trainability in some settings. This last observation provides some hope that a careful choice of error mitigation method can be useful. It also suggests that researchers could design and engineer error mitigation methods to enhance VQA trainability.

The rest of the manuscript is structured as follows. Section \ref{sec:framework} introduces the framework and notation for our work. We present our theoretical results in Section \ref{sec:theoreticalresults} and our numerical results in Section  \ref{sec:numerics}. Finally, our concluding discussions are presented in Section \ref{sec:discussions}. The proofs for our main results are presented in the Appendix.

\section{Framework}\label{sec:framework}

\subsection{Variational Quantum Algorithms}

The main goal of Variational Quantum Algorithms (VQAs) is to solve an optimization problem by minimizing a cost function that can be efficiently estimated on a quantum computer. In this work we consider settings where the cost function takes the form
\begin{equation}\label{eq:costfn}
    C(\thv) = \Tr\left[U(\thv)\rho_{in} U\dag(\thv)O\right]\,.
\end{equation}
In the above, given some Hilbert space $\HC$, we define the set of density operators $\SC(\HC)$ and set of bounded linear operators $\BC(\HC)$. We then denote $\rho_{in} \in \SC(\HC)$ as the input state, $U(\thv)\in \BC(\HC)$ as a unitary that corresponds to a parametrized quantum circuit with trainable parameters $\thv$, and $O\in \BC(\HC)$ is a Hermitian operator. The Variational Quantum Eigensolver~\cite{peruzzo2014variational}, variational quantum compiling~\cite{khatri2019quantum,sharma2019noise,jones2018quantum,heya2018variational}, quantum autoencoders~\cite{romero2017quantum}, and several other VQAs fit under the framework of Eq.~\eqref{eq:costfn}.

A quantum computer is employed to evaluate the cost function, or gradients thereof, and part of the computational complexity of the algorithm is designated to a classical computer that leverages the power of classical optimizers to solve the problem
\begin{equation}\label{eq:optimization}
\argmin_{\thv}    C(\thv)\,.
\end{equation}
The optimization task defined in Eq.~\eqref{eq:optimization} has been shown to be NP-hard~\cite{bittel2021training}. Moreover, on top of the typical difficulties associated with solving classical non-convex optimization problems, there are challenges that arise when training the parameters of a VQA due to the quantum nature of the problem itself. 

As quantum mechanics is intrinsically  a probabilistic theory, one has to deal with shot noise arising from finite sampling when estimating the cost function (or its gradient). This has led to the development of several quantum-aware optimizers that are frugal in the number of shots \cite{kubler2020adaptive,arrasmith2020operator, gu2021adaptive}.  Additionally,  it has been recently shown that certain properties of the cost function can induce so-called barren plateaus, originating due to highly expressive ansatzes~\cite{mcclean2018barren,holmes2021connecting,holmes2021barren}, global cost functions~\cite{cerezo2021cost}, high levels of entanglement~\cite{marrero2020entanglement,patti2020entanglement}, or the controllability of $U(\thv)$~\cite{larocca2021diagnosing}. When a cost function exhibits a barren plateau, with high probability the cost function partial derivatives are exponentially suppressed across the landscape. This means that an exponentially large number of shots are needed to navigate the flat landscape and determine a cost-minimizing direction~\cite{cerezo2020impact,arrasmith2020effect}. 

In this work we investigate the effect of noise and error mitigation techniques in solving the optimization task of Eq.~\eqref{eq:optimization}. For this purpose we investigate the task of resolving two points on the cost function landscape, as presented in Fig.~\ref{fig:paper-summary}. This is a central primitive in the training process that is utilized at each optimization step, regardless of whether one is using gradient-based or gradient-free methods. In gradient-based methods, a common strategy is to use the parameter shift rule, which constructs partial derivatives from two cost function values \cite{mitarai2018quantum, schuld2019evaluating}. Gradient-free methods such as simplex-based methods also compare two or more cost function values at each optimization step \cite{nelder1965simplex,powell1994direct}. Thus, this task is a key step for both gradient-based and gradient-free optimizers, and it reflects the ability of the optimizer to find a cost-minimizing direction at each step of the optimization.  As discussed below, under a finite shot budget this task becomes harder under cost concentration, leading to trainability issues.

\subsection{Effect of noise on the training landscape}\label{sec:effectsofnoise}

Hardware noise can impact the cost function landscape in a variety of ways such as changing the optimal cost function value, shifting the position of minima, and demoting a global minimum to a local minimum. All of the above present further challenges in the training of VQAs. In this section we briefly review some of the literature on the effect of noise on VQAs cost function landscapes. We summarize some of these effects in Fig.~\ref{fig:noise-effects}.

\subsubsection{Noise resilience}

Certain cost functions have been demonstrated to show optimal parameter resilience under particular noise models \cite{sharma2019noise}. This is a phenomenon where the position of the global cost minimum of the cost landscape is invariant under the action of noise. This has important consequences for trainability. There are many VQAs where the goal is to obtain optimal parameters, rather than the optimal cost value, such as when solving combinatorial optimization problems with the Quantum Approximate Optimization Algorithm (QAOA) \cite{farhi2014quantum}. If such cost landscapes display optimal parameter resilience, this leaves open the possibility of noisy training even if the cost value of the global minimum is altered by the noise. In fact, it has recently been shown that a small amount of dephasing errors can recover layerwise training of the QAOA \cite{campos2021training}. However, noise can also severely affect the trainability of the landscape in a number of ways, which we summarize below. 

\subsubsection{Noise-induced cost concentration and noise-induced barren plateaus}\label{sec:costconcentration}

Here we summarize the phenomenon of noise-induced cost concentration and noise-induced barren plateaus (NIBPs), as well as introduce some notation that we will use throughout the rest of this manuscript. 
This was formulated in Ref.~\cite{wang2020noise} for a general class of VQAs and a class of Pauli noise that includes as a special case local depolarizing noise. (See also Refs.~\cite{franca2020limitations,xue2021effects,marshall2020characterizing} for other discussions of the impact of noise.)
Consider a model of noise acting through a depth $L$ circuit with $n$-qubit input state $\rho_{in}$ as 
\begin{equation} \label{eq:noisystate}
     \widetilde{\rho} = (\mathcal{N}\circ \mathcal{U}_L\circ \cdots \circ \mathcal{N}\circ \mathcal{U}_1\circ \mathcal{N}\big) (\rho_{in}) 
\end{equation}
where $\{\UC_k\}^L_{k=1}$ denote unitary channels that describe collections of gates that act together in a layer, and $\NC=\bigotimes_{i=1}^n \NC_i$ is an instance of local Pauli channels. {In general we can consider different Pauli noise channels in each layer and our theoretical results can be simply extended to such settings, but we do not consider it here for simplicity of presentation.}
The action of $\mathcal{N}_j$ on a local Pauli operator $\sigma\in\{X,Y,Z\}$  can be expressed as 
\begin{equation}\label{eq:noisemodel}
    \mathcal{N}_j(\sigma)=q^{(j)}_{\sigma}\sigma\,,
\end{equation}
where we assume $-1< q^{(j)}_X,q^{(j)}_Y,q^{(j)}_Z<1$ for all qubit labels $j$. Here, we characterize the noise strength with a single parameter { $q=\max_{j,\sigma}\big\{\big|q^{(j)}_\sigma \big|\big\}<1$}. We denote a noisy cost function as
\begin{equation}
    \widetilde{C} = \Tr\left[O\widetilde{\rho}\right]\,,
\end{equation}
where $O$ is some Hermitian measurement operator (throughout the article we will use a tilde to denote noisy quantities). In Ref.~\cite{wang2020noise} it was shown that 
\begin{equation}\label{eq:costconcentration}
    \Big\vert \widetilde{C} - \frac{1}{2^n}\Tr[O]\Big\vert\, \leq\, D(q,n) \,,
\end{equation}
where $D(q,n)\in \mathcal{O}(q^{\alpha n})$ for some positive constant $\alpha$ if $L \in \Omega(n)$. {More generally the quantity on the left-hand side of Eq.~\eqref{eq:costconcentration} vanishes exponentially with increasing circuit depth $L$ for any fixed $n$.}  Thus, in the presence of the class of noise models considered, the noisy cost function exponentially concentrates on a fixed value if the depth scales linearly or superlinearly in the number of qubits.

The gradients across the cost function landscape show similar scaling \cite{wang2020noise}, {in that they also vanish exponentially in the number of qubits for linear depth circuits}, demonstrating a phenomenon known as NIBPs. This implies that the task of accurately determining gradients or cost function differences during the training process requires an exponential number of shots due to the need to resolve quantities to an exponentially small precision.

\begin{figure}[t]
    \includegraphics[width=\columnwidth,clip]{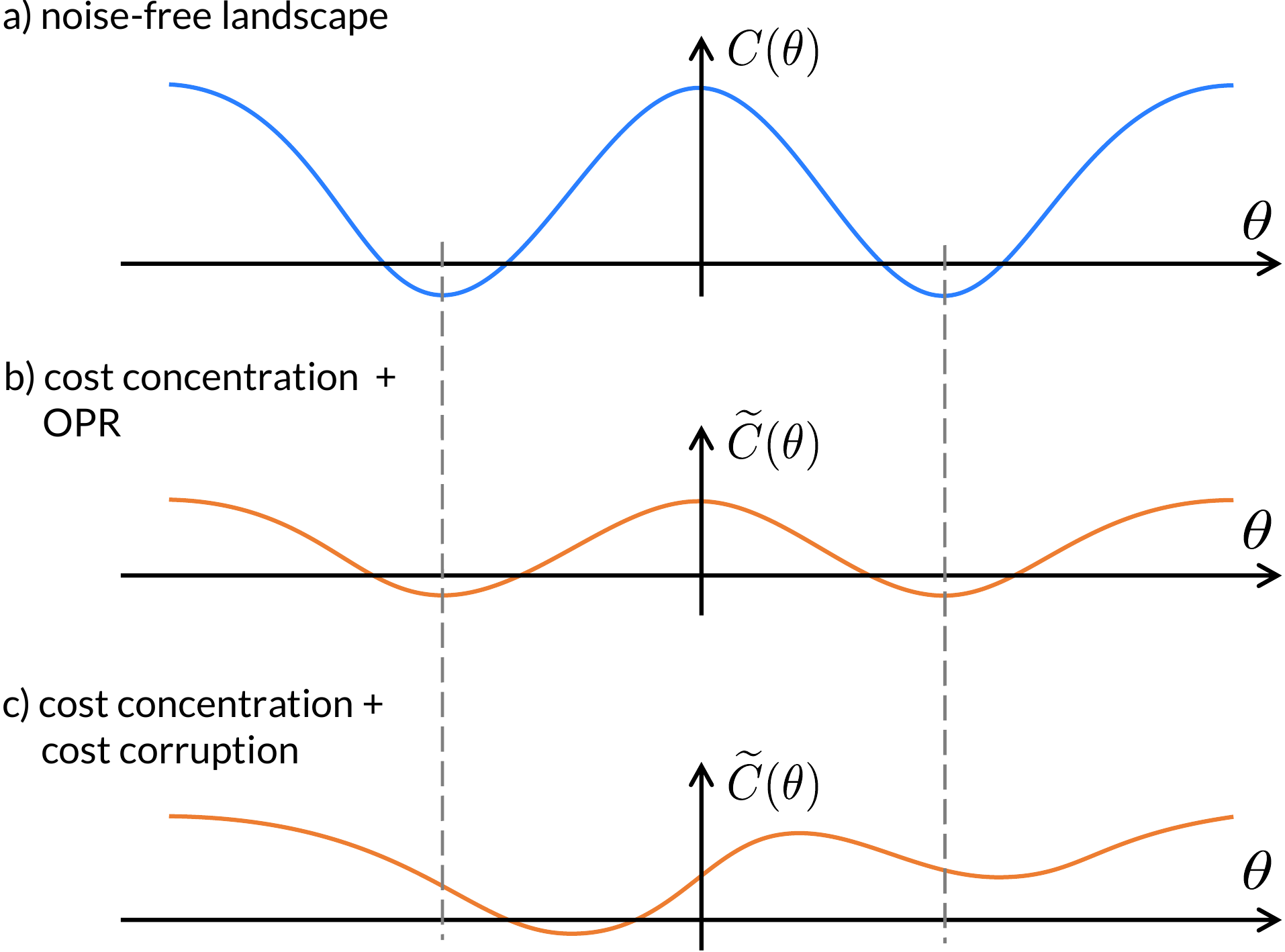}
    \caption{\textbf{Schematic of different effects due to noise on cost landscapes.} We present a 1-dimensional slice of a simplified cost landscape corresponding to a single parameter $\theta$. a) Depending on the parameterization strategy, some ansatzes can have degenerate minima. b) Certain types of local Pauli noise can cause the cost landscape to exponentially concentrate on a fixed value. Some can problems display optimal parameter resilience (OPR), where the location of the optimal parameters are invariant under action of the certain noise models. c) Aside from cost concentration, noise can also corrupt the cost landscape by breaking the degeneracy of optimal parameters, and shifting the location of minima.}\label{fig:noise-effects}
\end{figure}

\subsubsection{Cost corruption}\label{sec:costcorruption}

In general, a noise model that exhibits cost concentration and NIBPs would not simply uniformly flatten the cost landscape. Instead, we expect noise to additionally alter the cost landscape in many non-trivial ways. We refer to any additional adverse effects on the landscape as cost corruption. 
For example, it was shown in Ref.~\cite{fontana2020optimizing} that non-unital noise can break the degeneracy of exponentially-occurring global minima, thus proliferating local minima and impacting trainability. In addition, cost functions that do not exhibit optimal parameter resilience~\cite{sharma2019noise} limit the quality of noisy optimization, as the optimal parameters of $\widetilde{C}(\thv)$ do not correspond to the optimal parameters of $C(\thv)$.

\subsection{Error Mitigation Techniques}\label{sec:errormitigation}

We finish the discussion of our framework with a summary of the key features of the error mitigation techniques that we study in this article.
For a more detailed review, readers can refer to Refs.~\cite{endo2021hybrid,cerezo2020variationalreview}.

Consider the effects of noise on the cost function in Eq.~\eqref{eq:costfn}. We suppose the noise can be characterized by a single {(scalar)} parameter $\varepsilon$ and we denote the corresponding noisy state and cost function as $\widetilde{\rho}(\thv,\varepsilon)$ and $\widetilde{C}(\thv,\varepsilon) = \Tr[\widetilde{\rho}(\thv,\varepsilon)O]$ respectively. The goal of error mitigation is to construct an experimental protocol which obtains a mitigated cost function estimator $C_m(\thv)$ that approximates the noise-free value $C(\thv)$. The protocol to obtain $C_m(\thv)$ generally consists of running circuits that modify the original circuit of interest by inserting additional gates, preparing multiple copies of a state, changing the measurement operator, and classical post-processing of the expectation values of these circuits. These different utilizations of resources are summarized in a schematic in Fig.~\ref{fig:em_framework}.

Error mitigation protocols often lead to a larger variance in the statistical outcomes of each experiment, and thus more shots are required to estimate the error-mitigated cost value $C_m(\thv,\varepsilon)$ to a desired precision compared to the unmitigated noisy value $\widetilde{C}(\thv,\varepsilon)$. This is often quantified by the error mitigation cost, which is defined below. 

\begin{definition}[Error mitigation cost]\label{def:em_cost}
We define the error mitigation cost as
\begin{equation}
    \gamma (\thv,\varepsilon) = \frac{\Var[{C}_m(\thv,\varepsilon)]}{\Var[\widetilde{C} (\thv,\varepsilon)]}\,,
\end{equation}
where $\widetilde{C} (\thv,\varepsilon)$ denotes the noisy cost function value corresponding to vector of parameters $\thv$ at noise level $\varepsilon$, and ${C}_m(\thv,\varepsilon)$ denotes the corresponding error-mitigated quantity. 
\end{definition}

In certain settings we encounter in our theoretical analyses, $\gamma(\thv,\varepsilon)$ is independent of $\thv$. {In other cases, we assume that it is parameter independent, or we seek parameter-independent bounds. Thus, from hereon in this manuscript we will generally drop the parameter dependence of $\gamma(\thv, \varepsilon)$.}

We now summarize the error mitigation techniques that we study in this article. We note that recently, unified error mitigation techniques have also been proposed that combine two or more of the protocols that we discuss in this section \cite{lowe2020unified,mari2021extending,bultrini2021unifying}. Our results are also applicable to such strategies, however, we will only review the root strategies here.

\begin{figure*}[t]
    \includegraphics[width=\textwidth,clip]{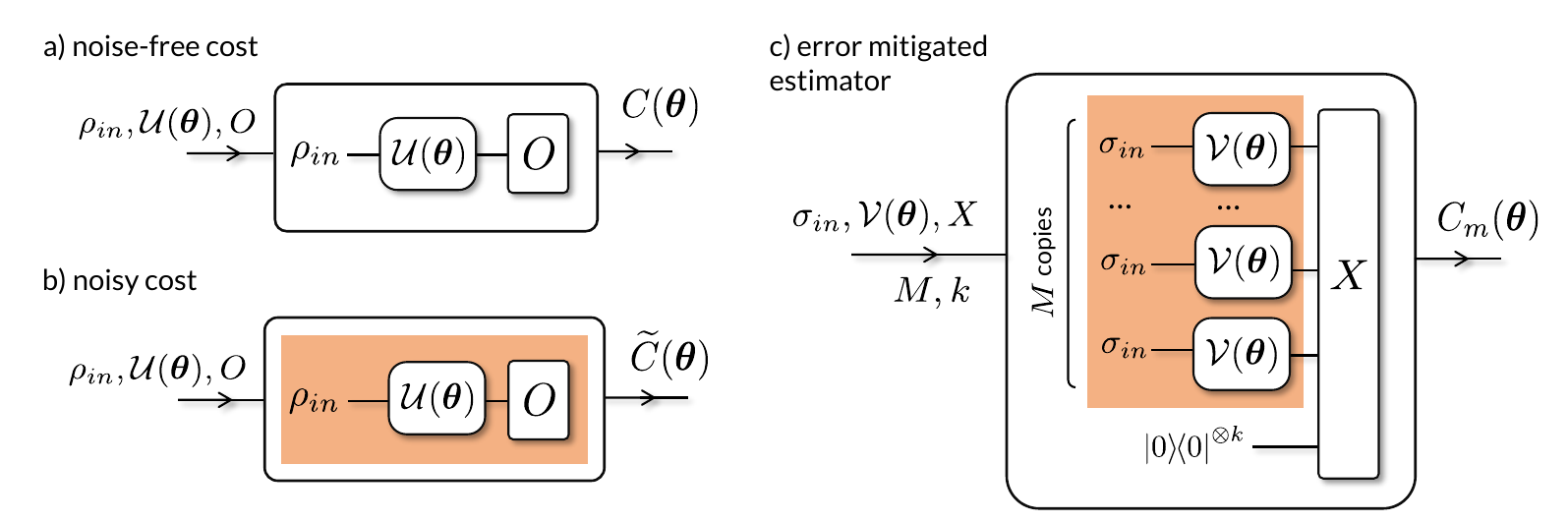}
    \caption{\textbf{Schematic of resource use in error mitigation.} Noise is indicated by the shaded orange region. (a) Cost function values are obtained by taking input state $\rho_{in}$, applying parameterized gates which we denote as a unitary channel $\mathcal{U}(\thv)$, and measuring the resulting state $\mathcal{U}(\thv)(\rho_{in})$ with observable $O$. (b) Noise can corrupt the gates in the circuit, as well as the state preparation and measurement processes. (c) Error mitigation aims to obtain a good approximation to the noise-free cost $C(\thv)$ by employing a number of strategies such as: modifying the gates implemented $\mathcal{U}(\thv) \rightarrow \mathcal{V}(\thv)$ or the input state $\rho_{in} \rightarrow \sigma_{in}$, utilizing multiple copies of the quantum circuit, modifying the measurement operator $O \rightarrow X$, and utilizing clean ancillary qubits at the end of the circuit. Many such circuits with different hyperparameters can be run, with their expectation values combined in a post-processing step, in order to construct the final error mitigated cost value $C_m(\thv)$.  {Note that here we have only indicated noise occuring in the initial part of the circuit--this reflects the assumptions of analyses in prior works} \cite{koczor2020exponential, huggins2020virtual}. {As we investigate the limitations of such error mitigation schemes, we keep these assumptions as a “best case” analysis.} One feature that distinguishes the approaches to error mitigation studied here from error correction is that error correction allows global access to the larger Hilbert space from the start of the circuit, whereas the error mitigation techniques studied here only allow the possibility for global operations at the end of the circuit. \label{fig:em_framework} } 
\end{figure*}

\subsubsection{Zero-Noise Extrapolation}

The goal of Zero-Noise Extrapolation is to run a given circuit of interest at $m+1$ increasing noise levels $\varepsilon < a_1\varepsilon < ... < a_m\varepsilon$, and to use information from the resulting expectation values to obtain an estimate of the zero-noise result. Here we summarize the key features of a protocol using Richardson extrapolation \cite{temme2017error,li2017efficient}, and exponential extrapolation \cite{endo2018practical}. 

\emph{Richardson Extrapolation.} Suppose that $\widetilde{C}(\thv_i,\varepsilon)$ admits a Taylor expansion in small noise parameter $\varepsilon$ as
\begin{equation} \label{eq:ZNE_taylor}
    \widetilde{C}(\thv_i,\varepsilon)=\widetilde{C}(\thv_i,0)+\sum_{k=1}^{m} p_{k}(\thv_i) \varepsilon^{k}+\mathcal{O}(\varepsilon^{m+1})\,,
\end{equation}
where $p_k$ are unknown parameters and $\widetilde{C}(\thv_i,0)=C(\thv)$ is the zero-noise cost function. By considering the equivalent expansion of $\widetilde{C}(\thv_i,a_1 \varepsilon)$ and combining the two equations one obtains 
\begin{align}\label{eq:ZNE_TaylorCm}
    C_m(\thv_i) &= \frac{a_1 \widetilde{C}(\thv_i,\varepsilon) - \widetilde{C}(\thv_i,a_1 \varepsilon) }{ a_1-1} \\
    &= \widetilde{C}(\thv_i,0) +\mathcal{O}(\varepsilon^2)\,,
\end{align}
which is a higher-order approximation of $\widetilde{C}(\thv_i,0)$ compared to simply using $\widetilde{C}(\thv_i,\varepsilon)$. This process can be repeated iteratively $m$ times to obtain an estimator which is accurate up to $\mathcal{O}(\varepsilon^{m+1})$ error.

\emph{Exponential extrapolation.} In some cases the noisy behavior may not be well-depicted by a Taylor expansion. As an alternative one can consider an exponential model
\begin{equation}\label{eq:ZNE_exp}
    \widetilde{C}(\thv_i,\varepsilon)= r(\thv_i,\varepsilon)^{-t(\thv_i,\varepsilon)} \Big( \sum_{k=0}^{m} p_{k}(\thv_i) \varepsilon^{k}+\mathcal{O}(\varepsilon^{m+1}) \Big),
\end{equation}
for some $r$ and $t$ which in general can be functions of $\varepsilon$. {For instance, in Ref.~\cite{endo2018practical} it is chosen that $r(\thv_i,\varepsilon)^{-t(\thv_i,\varepsilon)} = e^{-N_g \varepsilon}$ where $N_g$ is the number of gates. In this case, the noise-free cost value is $p_0(\thv_i)$.} We can also construct an extrapolation strategy that is tailored towards noisy cost function values that are dominated by NIBP scaling as in Eq.~\eqref{eq:costconcentration}, where we model the effects of noise as
\begin{align}\label{eq:ZNE_NIBP}
    \widetilde{C}(\thv_i,q)&= A + q^{L} \Big( B(\thv_i) +\sum_{k=1}^{m} p_{k} (1-q)^{k} \nonumber \\
    & \qquad\qquad\qquad\quad +\mathcal{O}\big((1-q)^{m+1}\big) \Big)\,,
\end{align}
{where $q< 1$ is the Pauli noise parameter defined in Eq.~\eqref{eq:noisemodel} which equals zero for maximal noise. Here, $A$ is the fixed point of the noise (corresponding to the maximally mixed state) and $A+ B(\thv_i)$} is the noise-free cost value.
For these two strategies we can similarly construct $C_m(\thv_i)$ as linear combinations of $\{\widetilde{C}(\thv_i,a_i\varepsilon) \}_{i=0}^{m}$ {or $\{\widetilde{C}(\thv_i,q/a_i) \}_{i=0}^{m}$} to achieve $\mathcal{O}(\varepsilon^{m+1})$ approximations of the {zero-noise cost value}. We detail these constructions in Section \ref{sec:appdx_review_ZNE} of the Appendix.

\subsubsection{Virtual Distillation}

Virtual Distillation, also known as Error
Suppression by Derangement, was proposed concurrently in Refs.~\cite{koczor2020exponential} and \cite{huggins2020virtual}. In this article we consider the two error mitigation protocols in Ref.~\cite{koczor2020exponential} (denoted ``A'' and ``B'') to respectively prepare
\begin{equation}\label{eq:VD_A}
    C_m^{(A)}(\thv_i)=\Tr[\widetilde{\rho}_i^M O]/\Tr[\widetilde{\rho}_i^M]\,,
\end{equation} and
\begin{equation}\label{eq:VD_B}
    C_m^{(B)}(\thv_i) = \Tr[\widetilde{\rho}_i^M O]/\lambda_i^M\,,
\end{equation}
where $\lambda_i$ is the dominant eigenvalue of $\widetilde{\rho}_i \equiv \widetilde{\rho}(\thv_i)$. The operator $\widetilde{\rho}_i^M$ can be obtained by preparing $M$ copies of $\widetilde{\rho}_i$ in a tensor product state $\widetilde{\rho}_i^{\otimes M}$ and applying a cyclic shift operator. We note that protocol B presumes access to the dominant eigenvalue beforehand, which could potentially be computed via the techniques of Ref.~\cite{cerezo2020variational}.

\subsubsection{Probabilistic Error Cancellation}

Probabilisitic Error Cancellation utilizes many modified circuit runs in order to construct a quasiprobability representation of the noise-free cost function \cite{temme2017error, endo2018practical}. We assume that the effect of the noise can be described by a quantum channel $\mathcal{N}$ that occurs after a gate that we denote with unitary channel $\mathcal{U}$. Here we make the simplifying assumption that this is the only gate in the circuit, and we treat the general case in Section \ref{sec:appdx_review_QP} of the Appendix, as well as provide a more detailed exposition. The goal of this protocol is to simulate the inverse map $\mathcal{N}^{-1}$. Note that, in general, this will not always correspond to a CPTP map. Despite this fact, if one has a basis of (noisy) quantum channels $\{\mathcal{B}_\alpha \}_{\alpha}$, corresponding to experimentally available channels, one can expand the inverse map in this basis as $\mathcal{N}^{-1} = \sum_{\alpha} q_{\alpha}  \mathcal{B}_{\alpha}$,  for some set of $q_{\alpha} \in \mathbb{R}$. By defining a probability distribution $p_{\alpha} = |q_{\alpha}|/G_\mathcal{N}$ where $G_\mathcal{N} = \sum_{\alpha} |q_{\alpha}|$, the noise-free expectation value can then be written as a quasiprobability distribution
\begin{align}
    {C}_{\mathcal{U}(\rho)} &= G_\mathcal{N} \sum_{\alpha} \mathrm{sgn}(q_{\alpha})\, p_{\alpha}\, \Tr \big[ \BC_{\alpha}\NC\UC(\rho_{in})O\big] \,,
\end{align}
where $\rho_{in}$ is the input state, $O$ is the measurement operator, and $\mathrm{sgn}(q_{\alpha})$ denotes the sign of $q_{\alpha}$. The idea is that if one has access to the set of CPTP maps $\{ \mathcal{B}_{\alpha}\}_\alpha$ in the noisy native hardware gate set, then one can obtain an estimate of the noise-free cost ${C}_{\mathcal{U}(\rho)}$ as follows: (1) With probability $p_\alpha$, prepare the circuit of interest with additional gate $\mathcal{B}_\alpha$ in order to obtain the expectation value $\Tr \left[ \BC_{\alpha}\NC\UC(\rho_{in})O\right]$. (2) Multiply the result by $\mathrm{sgn}(q_{\alpha})G_\mathcal{N}$. (3) Repeat process many times and sum results.

\subsubsection{Clifford Data Regression (CDR) and linear ansatz methods}

The main idea of linear ansatz methods is to assume that we can approximately reverse the effects of noise with an affine map, and thus we construct a linear ansatz of the form
\begin{equation}\label{eq:CDR_ansatz}
    C_m(\thv,\boldsymbol{a}) = a_1(\thv)\widetilde{C}(\thv) + a_2(\thv)\,,
\end{equation}
where $\boldsymbol{a}(\thv) = (a_1(\thv),a_2(\thv))$ is a vector of parameters to be determined.  In general we expect $\boldsymbol{a}$ to be highly dependent on $\thv$. In Ref.~\cite{czarnik2020error}, the authors use data regression to learn the optimal parameters $\vec{a}^*(\thv)$ with training data comprising of pairs of noise-free and corresponding noisy cost function values $\mathcal{T}_{\thv} = \{(C_j, \widetilde{C}_j )\}_j$, where the circuits are predominantly constructed from Clifford gates. The noise-free cost values can be simulated efficiently on a classical computer whilst the noisy cost values can be evaluated directly on the quantum computer. This strategy is known as Clifford Data Regression.

Other methods have been proposed to learn the optimal parameters $\vec{a}^*(\thv)$. In Ref.~\cite{montanaro2021error} the authors further develop the idea of training-based error mitigation by considering alternative training data comprising of fermionic linear optics circuits. One can also model the noise as global depolarizing noise. 
Under this assumption, $\vec{a}^*(\thv)$ has an exact solution in terms of a single noise parameter.
Subsequently, various techniques can be used to estimate the noise parameter \cite{vovrosh2021simple, rosenberg2021experimental, he2020zero, shaw2021classical, google2020observation}. {Finally, we note that an alternative learning-based method has been proposed in which Clifford data is used to learn the optimal quasi-probability distribution \cite{strikis2020learning}. In this case, our results on probabilistic error cancellation are directly applicable.}

\subsubsection{Previous results on sampling overhead}
{It is well known that error mitigation techniques require a larger shot budget than estimating unmitigated expectation values, due to the amplification of statistical variance. Indeed, this has been discussed as part of the original proposal of many error mitigation schemes} \cite{temme2017error, huggins2020virtual, koczor2020exponential, strikis2020learning, czarnik2021qubit}. {For probabilistic error cancellation, in Ref.~\cite{xiong2020sampling} Xiong et al.~investigate the sampling overhead of probabilistic error cancellation in further detail for various noise channels. 
We stress that all of the aforementioned analysis quantifies the sampling overhead in recovering individual expectation values to constant precision. We note that, to the best of our knowledge, prior to our work the effects of error mitigation in resolving trainability issues due to noise have not been studied.
}

\section{Theoretical Results}\label{sec:theoreticalresults}

We present two sets of theoretical results.
First, in Section \ref{sec:estimator_conc} we show that a broad class of error mitigation techniques cannot undo the exponential resource requirement that exponential cost concentration presents. This has implications for both the trainability of noisy VQAs, as well as the accurate estimation of noise-free cost function values in general. Second, in Section \ref{sec:protocol_specific}, we work predominantly in the non-asymptotic  regime (in terms of scaling in $n$) and investigate to what extent different error mitigation strategies can improve the resolvability of the noisy cost landscape, assuming that some cost concentration has occurred. For these purposes we introduce a class of quantities which quantify the improvement of the resolvability of the cost function landscape after error mitigation, which we call the relative resolvability (see Defs. \ref{def:resolvability}-\ref{def:av_resolvabilityII}). Using these quantities we study Zero-Noise Extrapolation (Sec.~\ref{sec:ZNE}), Virtual Distillation (Sec.~\ref{sec:VD}), Probabilistic Error Cancellation (Sec.~\ref{sec:QP}) and linear ansatz methods which include Clifford Data Regression (Sec.~\ref{sec:CDR}).  In the settings that we consider, we find that in many cases error mitigation impedes the optimizer's ability to find good optimization steps, and is worse than performing no error mitigation.   

\subsection{Asymptotic scaling results (exponential estimator concentration)}\label{sec:estimator_conc}

In this section we show that full mitigation of exponential cost concentration is not possible for a general class of error mitigation strategies. Specifically, we show that one cannot remove the exponential scaling that local Pauli noise incurs without investing exponential resources elsewhere in the mitigation protocol.

We start by remarking that, as summarized in Fig.~\ref{fig:em_framework}, all of the strategies presented in Sec.~\ref{sec:errormitigation} consist of preparing linear combinations of expectation values of the form
\begin{align}\label{eq:thm1_estimator}
    E_{\sigma,X,M,k} = \Tr\left[X \left(\sigma^{\otimes M}\otimes \ket{0}\bra{0}^{\otimes k}\right)\right]\,,
\end{align}
for some  $n$-qubit quantum state $\sigma \in S(\mathcal{H})$ that in general can be prepared by a different circuit to that of the state of interest, for $\ket{0}\bra{0} \in S(\mathcal{H}')$ and for some $X\in B(\mathcal{H}^{\otimes M}\otimes \mathcal{H}'^{\otimes k})$. That is, one can prepare multiple copies of a state, prepare different quantum circuits, and apply general measurement operators. In order to generalize the setting further, we also allow the possibility to utilize multiple clean ancillary qubits at the end of the circuit. By considering linear combinations of such quantities, we also account for the ability to post-processing measurement results classically with a linear map, such as is the case with Probabilistic Error Cancellation. In the following theorem we show how quantities of the form \eqref{eq:thm1_estimator} concentrate under local Pauli noise of the form \eqref{eq:noisystate}.

\begin{theorem}\label{thm:fullmitigation}
Consider an error mitigation strategy that, as a step in its protocol, estimates $E_{\sigma,X,M,k}$ as defined in Eq.~\eqref{eq:thm1_estimator}.
Suppose that $\sigma$ is prepared with a depth $L_{\sigma}$ circuit and experiences local Pauli noise according to Eq.~\eqref{eq:noisystate}. Under these conditions, $E_{\sigma,X,M,k}$ exponentially concentrates with increasing circuit depth on a state-independent fixed point as
\begin{align}
    &\left| E_{\sigma,X,M,k} - \Tr\left[X \left( \frac{\id^{\otimes M}}{2^{Mn}} \otimes \ket{0}\bra{0}^{\otimes k}\right)\right]\right| \\ 
    &\hspace{11.5em} \leq\, G_{\sigma,X,M}(n) \,, \label{eq:prop1_conc} 
\end{align}
where $\id \in B(\mathcal{H})$ is the $n$-qubit identity operator and 
\begin{equation}
    G_{\sigma,X,M}(n) =  \sqrt{\ln 4}\,\big\|X\big\|_\infty M n^{1/2} q^{{c}(L_{\sigma}+1)}\,,
\end{equation}
with noise parameter $q\in[0,1)$ and constant $c = 1/(2\ln 2) \approx 0.72$. 
\end{theorem}

Theorem \ref{thm:fullmitigation} shows that quantities of the form \eqref{eq:thm1_estimator} exponentially concentrate in the depth of the circuit. As we summarize in the schematic in Fig.~\ref{fig:em_framework}, such quantities generalize expectation values that are prepared by many different error mitigation protocols. {We provide a proof of the theorem in Appendix \ref{sec:appdx_asymptotics}.} We now explicitly demonstrate how Theorem \ref{thm:fullmitigation} affects the mitigated cost values that these protocols output.

\begin{corollary}[Exponential estimator concentration]\label{cor1}
Consider an error mitigation protocol that approximates the noise-free cost value $C(\thv)$ by estimating the quantity 
\begin{equation}\label{eq:cor1_estimator}
    C_m(\thv) = \sum_{(\sigma(\thv),X,M,k)\in T} a_{X,M,k}\,E_{\sigma(\thv),X,M,k}\,,
\end{equation}
where each $E_{\sigma,X,M,k}$ takes the form \eqref{eq:thm1_estimator} {and each $a_{X,M,k} \in \mathbb{C}$}. We denote $M_{max}$ and $a_{max}$ as the maximum values of $M$ and $a_{X,M,k}$ respectively accessible from a set $T$ defined by the given protocol. Assuming $\|X\|_\infty \in \OC(\poly(n))$, there exists a fixed point $F$ independent of $\thv$ such that
\begin{equation}
    \big|C_m(\thv)-F\big|\in\OC(2^{-\beta n}a_{max}|T|M_{max})\,,
\end{equation}
for some constant $\beta\geq 1$ if the circuit depths satisfy 
\begin{equation}
    L_{\sigma(\thv)} \in \Omega(n)\,,
\end{equation}
for all $\sigma(\thv)$ in the construction \eqref{eq:cor1_estimator}. That is, if the depths of the circuits scale linearly or superlinearly in $n$ then one requires at least exponential resources to distinguish $C_m$ from its fixed point, for instance by requiring an exponential number of shots, or by requiring an exponential number of state copies $M_{max}$.
\end{corollary}

We note that the assumption $\|X\|_\infty \in \OC(\poly(n))$ is satisfied in most settings, and in particular is satisfied for all error mitigation protocols discussed in Sec.~\ref{sec:errormitigation}. For instance, in the case of Virtual Distillation, $X$ corresponds to a cyclic shift operator followed by a Pauli observable, and thus $\|X\|_\infty \in \OC(1)$. {Corollary \ref{cor1} implies that under conditions that generate a NIBP, in order to distinguish any two cost values with constant probability, one requires resource consumption (in shots or number of state copies) that scales exponentially in the number of qubits. Thus, if one views NIBPs as an exponential resource (shots) issue, Corollary \ref{cor1} shows that the class of error mitigation schemes considered cannot circumvent this issue.}  In Appendix \ref{sec:appdx_asymptotics} we present a more detailed statement that explains how such resources may be consumed.

Whilst the use of clean ancillary qubits as part of an error mitigation protocol, utilized as in Equation \eqref{eq:thm1_estimator} and Fig.~\ref{fig:em_framework}, has not been widely studied, Corollary \ref{cor1} rules out the possibility that such resources used at the end of the circuit would offer advantage in countering the exponential scaling effects due to cost concentration. Indeed, upon inspecting \eqref{eq:prop1_conc}, the ancilla appear explicitly in the form of the fixed point. This highlights a key difference between many error mitigation strategies and error correction, as error correction utilizes resources (such as a larger Hilbert space) in the middle of the computation, whilst the error mitigation protocols considered here are based on processing states obtained at the end of a noisy computation. Our result leaves open the possibility that novel error mitigation protocols that move beyond the framework of \eqref{eq:cor1_estimator} and Fig.~\ref{fig:em_framework} can have hope of countering the exponential scaling of exponential cost concentration and NIBPs.

\subsection{Non-asymptotic protocol-specific results}\label{sec:protocol_specific}

In this section we present predominantly non-asymptotic results for Zero-Noise Extrapolation (Sec.~\ref{sec:ZNE}), Virtual Distillation (Sec.~\ref{sec:VD}), Probabilistic Error Cancellation (Sec.~\ref{sec:QP}), and methods which use a linear ansatz such as Clifford Data Regression (Sec.~\ref{sec:CDR}). For each protocol, we
investigate the effect of error mitigation on the resolvability of the cost landscape, for different classes of noisy states. To {this} end, we first define a class of resolvability measures which quantify {how many shots it takes to resolve the cost landscape after applying error mitigation,} compared to no mitigation at all. {We provide proofs of these theoretical statements, as well as some extensions thereof, in Appendix \ref{appendix:protocol-specific}.}

\subsubsection{Definitions}

\begin{definition}[Relative resolvability for two points] \label{def:resolvability}
Consider two locations in parameter space $\thv_{1}$, $\thv_{2}$ and their corresponding points on the cost landscape. Denote the number of shots to resolve these two points from each other {to precision proportional to their cost difference} with and without error mitigation as $N_\text{EM}$ and $N_\text{noisy}$ respectively. We define the relative resolvability for $\thv_{1}$ and $\thv_{2}$ at error level $\varepsilon$ as 
\begin{align} 
    \rchi({\thv_{1,2}},\varepsilon) &= \frac{N_\text{noisy}({\thv_{1,2}},\varepsilon)}{N_\text{EM}({\thv_{1,2}},\varepsilon)} \label{eq:chi-pt1} \\
    &= \frac{1}{\gamma(\varepsilon)}\, \left(\frac{\Delta C_m({\thv_{1,2}},\varepsilon)}{\Delta \widetilde{C}({\thv_{1,2}},\varepsilon)}\right)^2 \label{eq:chi-pt2}\,,
\end{align}
where we have used the shorthand notation {for functional dependence on ${\thv_{1}},{\thv_{2}}$ as $f({\thv_{1,2}}) = f({\thv_{1}},{\thv_{2}})$}, $\gamma$ is the error mitigation cost as defined in Definition $\ref{def:em_cost}$, and where we denote
\begin{align}
    \Delta \widetilde{C}({\thv_{1,2}},\varepsilon) &= \widetilde{C}({\thv_{1}},\varepsilon) - \widetilde{C}({\thv_{2}},\varepsilon)\,, \\
    \Delta {C}_m({\thv_{1,2}},\varepsilon) &= {C}_m({\thv_{1}},\varepsilon) - {C}_m({\thv_{2}},\varepsilon)\,. 
\end{align}
\end{definition}

{Our definition of relative resolvability is centered around quantifying the sample overhead of the operational task of distinguishing two states (points in parameter space) via their corresponding cost values. The relative resolvability compares this task for the noisy setting with the error-mitigated setting. In order to obtain Eq.~\eqref{eq:chi-pt2} from Eq.~\eqref{eq:chi-pt1}, we consider the usual formula for the standard error of the sample mean, that is, $\sqrt{\operatorname{Var}[{C}_m(\thv, \varepsilon)]/ N_{EM} }$ for mitigated cost values and $\sqrt{\operatorname{Var}[\widetilde{C}(\thv, \varepsilon)]/ N_{noisy} }$ for non-mitigated cost values. Specifically, we suppose that successful resolution of the two points corresponds to achieving a small enough sample mean error proportional to the difference in exact cost value corresponding to those two points. Said differently, we ask an error mitigated optimizer and a non-mitigated optimizer to resolve a length scale proportional to the separation between two cost function values on the mitigated cost landscape and non-mitigated cost landscape, respectively. The proportionality constant can be thought of as being chosen arbitrarily, as in Eq.~\eqref{eq:chi-pt1} we consider a ratio of shots (thus whatever proportionality constant is chosen cancels out).} 

We see that if $\rchi({\thv_{1,2}},\varepsilon) > 1$, {then ${N_\text{EM}({\thv_{1,2}},\varepsilon)} < {N_\text{noisy}({\thv_{1,2}},\varepsilon)}$. Thus,} error mitigation has successfully increased the resolvability of the cost values corresponding to the cost values at $\thv_1$ and $\thv_2$. Note that this criterion is a necessary but not sufficient condition for error mitigation to reverse the effects of cost concentration on the cost landscape. Namely, it does not require the mitigated landscape to accurately reflect the noise-free landscape, and it does not account for other trainability issues such as proliferation of minima. If $\rchi({\thv_{1,2}},\varepsilon) < 1$ then this implies that error mitigation has exacerbated the resolvability issues associated with cost concentration and NIBPs, and it has been counterproductive in fixing these trainability issues.

For a general cost {function}, the relative resolvability of cost function points after mitigation may vary significantly across the landscape, or be different for different choices of ansatzes and noise models. {Specifically, we seek a more representative length scale to resolve rather than the cost difference between two arbitrary cost values.}  This motivates us to seek averaged measures of resolvability. We consider two types of averaging: first, an {(``operationally-motivated'')} average over cost function points across a given cost landscape; second, a {(``physically motivated'')} average over a set of noisy states.

\begin{definition}[{Average relative resolvability across cost landscape}]\label{def:av_resolvability}
Denote {$\thv_{*} = \textrm{argmin}_{\thv}\,\widetilde{C}(\thv, \varepsilon)$} the vector of parameters that corresponds to the global {noisy} cost minimum at noise parameter $\varepsilon$. We then define the averaged relative resolvability as
\begin{equation} \label{eq:av_resolvability}
    \overline{\rchi}(\varepsilon) = \frac{1}{\gamma(\varepsilon)}\, \left( \frac{\langle\Delta C_m({\thv_{i,*}},\varepsilon)\rangle_i}{\langle\Delta \widetilde{C}({\thv_{i,*}},\varepsilon)\rangle_i} \right)^2\,,
\end{equation}
where $\langle\cdot\rangle_i$ denotes the mean over all parameter vectors $\thv_{i}$ accessible with the given ansatz of consideration, and where we denote
\begin{align}
    \Delta {C}_m({\thv_{i,*}},\varepsilon) &= {C}_m({\thv_{i}},\varepsilon) - {C}_m({\thv_{*}},\varepsilon)\,,\\
    \Delta \widetilde{C}({\thv_{i,*}},\varepsilon) &= \widetilde{C}({\thv_{i}},\varepsilon) - \widetilde{C}({\thv_{*}},\varepsilon)\,.
\end{align}
\end{definition}

{Definition \ref{def:av_resolvability} provides a quantity which evaluates the average performance of a given error mitigation protocol across some given cost landscape.} Averaging across a given cost landscape gives a result that is particular to the choice of ansatz, measurement operator and noise model.   {Alternatively, we can view cost concentration as ultimately physically originating from a concentration of states to the maximally mixed state. In order to evaluate the performance of error mitigation in aiding the resolution of states, we consider a physically-motivated average over the basis of a noisy state. Specifically, we} consider an average over noisy states that have the same spectrum. This choice of class of noisy states is also motivated by the fact that a central mechanism of cost concentration and NIBPs under unital Pauli noise is the loss of purity. When we consider the {basis-averaged} relative resolvability for Virtual Distillation in Section \ref{sec:VD}, it will turn out to be bounded by a function of the purity for such states.

\begin{definition}[{Basis-averaged relative resolvability}]\label{def:av_resolvabilityII}
Consider a  normalized spectrum $\boldsymbol{\lambda} \in \mathbb{R}_+^{2^n}$ which corresponds to the eigenspectrum of some noisy state. We define the 2-design-averaged relative resolvability as
\begin{equation}
    \dbloverline{\rchi}_{\boldsymbol{\lambda}} = \frac{1}{\gamma(\boldsymbol{\lambda})} \frac{\big\langle\big(C_m(\rho_{\boldsymbol{\lambda}},U_i) - \Tr[O]/2^n \big)^2\big\rangle_{U_i}}{\big\langle\big(\widetilde{C}(\rho_{\boldsymbol{\lambda}},U_i) - \Tr[O]/2^n \big)^2\big\rangle_{U_i} }\,, 
\end{equation}
where $\langle\cdot\rangle_{U_i}$ denotes an average over unitaries $U_i$ drawn from a unitary 2-design, $\rho_{\boldsymbol{\lambda}}$ is an arbitrarily chosen reference state with spectrum $\boldsymbol{\lambda}$, and where we denote 
\begin{align}
    \widetilde{C}(\rho_{\boldsymbol{\lambda}},U_i) &= \Tr[U_i\rho_{\boldsymbol{\lambda}}U^\dag_i O]\,, \\
    {C}_m(\rho_{\boldsymbol{\lambda}},U_i) &= \Tr[\mathcal{M}(U_i\rho_{\boldsymbol{\lambda}}U^\dag_i) O]\,,
\end{align}
where $\mathcal{M}: S(\mathcal{H})\mapsto B(\mathcal{H})$ is a map that describes the action of the error mitigation protocol.
\end{definition}
In Fig.~\ref{fig:averagegame} we present {a} schematic of our {basis-averaged} relative resolvability. 

\begin{figure}[t]
    \includegraphics[width=\columnwidth,clip]{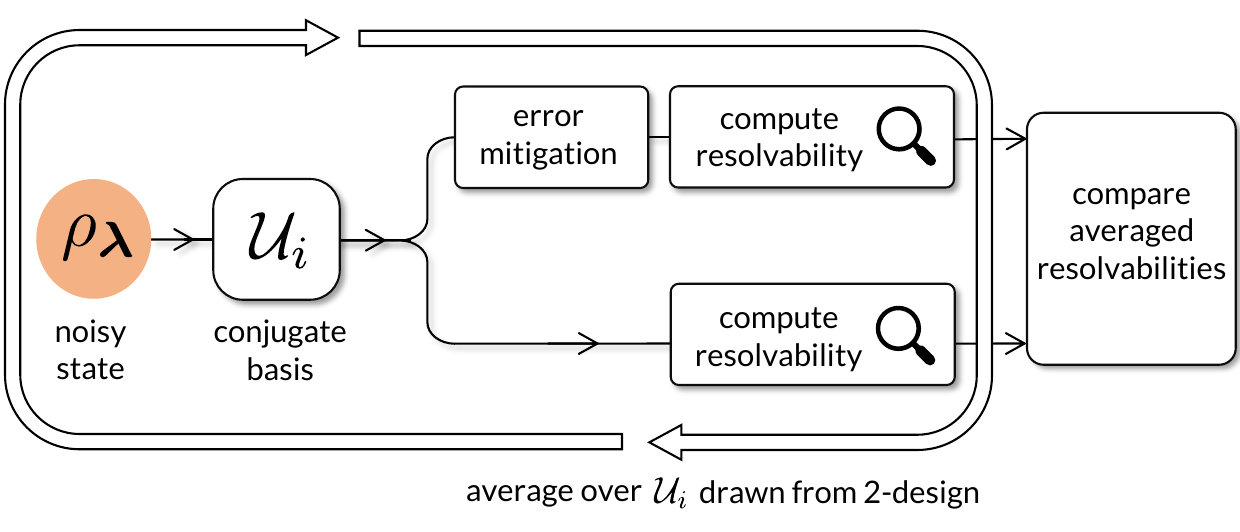}
    \caption{\textbf{{Schematic for basis-averaged relative resolvability.}} In Definition \ref{def:av_resolvabilityII} we consider a broader averaged resolvability measure where the average is taken over a class of noisy states, rather than a cost landscape generated by a particular ansatz. This is constructed from the following game: (1) Prepare a reference noisy state $\rho_{\boldsymbol{\lambda}}$ with spectrum $\boldsymbol{\lambda}$ and conjugate by unitary $U_i$ drawn from a 2-design. (2) Pass the resulting state through the considered error mitigation protocol, and evaluate the resolvability from the fixed point of the noise. (3) Do the same, without error mitigation. (4) Average over the 2-design and compare the averaged resolvabilities.  }
    \label{fig:averagegame}
\end{figure}

In the results that follow we do not constrain ourselves to investigating specific models of Pauli noise. Instead, we simply suppose that there exists a noise model that causes some concentration of the cost function onto some state-independent fixed point. Any further assumptions on the effects of the noise model are explicitly specified in the statement of the particular result in question. As a precursor to looking at more generic models of noisy states, we will also find it useful to investigate the performance of error mitigation strategies in the presence of global depolarizing noise of the form
\begin{equation}\label{eq:gdepol_noise}
    \rho \xrightarrow{\mathcal{D}} \widetilde{\rho} = (1-p) \rho + p \frac{\id}{2^n}\,,
\end{equation}
where $\mathcal{D}$ is the global depolarizing channel and $p$ is the depolarizing probability. Our justification for studying this noise model is twofold. First, global depolarizing noise provides a clean model of cost concentration with no other cost corrupting effects of the noise. Therefore, if a given error mitigation strategy is to mitigate the effects of cost concentration and NIBPs, we expect it to be able to perform well on this noise model. Second, the structure of many error mitigation strategies is directly motivated by the model of global depolarizing noise \cite{czarnik2020error, montanaro2021error, vovrosh2021simple, rosenberg2021experimental, he2020zero, shaw2021classical, google2020observation}. Indeed, many such strategies have been shown to achieve good or perfect performance with this noise model in mitigating noisy cost function values \cite{koczor2020exponential,huggins2020virtual, czarnik2020error, vovrosh2021simple}. However, we stress that trainability may simultaneously get worse, which is what we will now investigate.

\subsubsection{Zero-Noise Extrapolation}\label{sec:ZNE}

In this section we present our results on Zero-Noise Extrapolation. Throughout this section, in order to estimate the sample cost of error mitigation we will make the simplifying assumption that
\begin{equation}\label{eq:ZNE_assumption}
    \Var[\widetilde{C}(\thv, a\varepsilon)] \geq \Var[\widetilde{C}(\thv,\varepsilon)]\,,
\end{equation}
for all $\thv$ and $a\geq 1$ that is, the statistical fluctuations in measurement outcomes at the boosted noise level are no smaller than that at the base noise level. We note that similar assumptions are made in the literature for Zero-Noise Extrapolation and Quasi-Probability Methods \cite{endo2021hybrid, endo2018practical}. In Appendix \ref{sec:appdx_ZNE} we provide intuition as to why we expect this assumption to be true for large $a$ for noise models whose fixed point is the maximally mixed state.

First, we consider the simple model of global depolarizing noise.

\begin{proposition}[Relative resolvability of Zero-Noise Extrapolation with global depolarizing noise, 2 noise levels]\label{prop:ZNE_depol}
Consider a circuit with $L$ instances of global depolarizing noise of the form Eq.~\eqref{eq:gdepol_noise}.
Consider a Richardson extrapolation strategy based on Eq.~\eqref{eq:ZNE_taylor}, an exponential extrapolation strategy based on Eq.~\eqref{eq:ZNE_exp} and a NIBP extrapolation strategy based on Eq.~\eqref{eq:ZNE_NIBP}. We presume access to an augmented noisy circuit where the error probability is exactly increased by factor $a_1>1$ as $p\rightarrow a_1p$. Then, we have
\begin{align}
    &{\rchi_{depol}}\leq \frac{\left(c-\frac{(1-a_1p)^L}{(1-p)^L}\right)^2}{c^2+1}\,,
\end{align}
where $\rchi_{depol}$ is the relative resolvability (see Definition \ref{def:resolvability}) for global depolarizing noise, and where 
\begin{equation}\label{eq:c-cases}
    c = \begin{cases}
    \quad a_1\quad &\text{for Richardson extrapolation,} \\
    \frac{a_1r(\varepsilon)^{t(\varepsilon)}}{r(a_1\varepsilon)^{t(a_1\varepsilon)}} \quad &\text{for exponential extrapolation,} \\
    \; a_1^{-(L+1)} \quad &\text{for NIBP extrapolation}\,.
    \end{cases}
\end{equation}
Thus, $\rchi_{depol}\leq 1$ for all of the above extrapolation strategies with access to 2 noise levels.
\end{proposition}

We see that for all the above techniques, Zero-Noise Extrapolation with access to 2 noise levels decreases the resolvability of the cost function under global depolarizing noise. Further, if one attempts to directly reverse the exponential scaling of NIBPs that global depolarizing noise incurs, one obtains an exponentially worse relative resolvability. We now consider how resolvability behaves under Zero-Noise Extrapolation on average across the cost landscape, given a generic noise model.

\begin{proposition}[Average relative resolvability of Zero-Noise Extrapolation, 2 noise levels]\label{prop:ZNE}
Consider a Richardson extrapolation strategy based on Eq.~\eqref{eq:ZNE_taylor}, an exponential extrapolation strategy based on Eq.~\eqref{eq:ZNE_exp} and a NIBP extrapolation strategy based on Eq.~\eqref{eq:ZNE_NIBP}. We presume perfect access to an augmented noisy circuit where the noise rate is increased by factor $a_1>1$. We denote $\thv_{\varepsilon*}$ as the parameter corresponding to the global cost minimum at base noise parameter $\varepsilon$. 
Further denote $ \frac{\langle\Delta\widetilde{C}(\thv_{i,\varepsilon*}, a_1\varepsilon)\rangle_{i}}{\langle\Delta\widetilde{C}(\thv_{i,\varepsilon*},\varepsilon)\rangle_{i}} = z$. 
Any such noise model has an average relative resolvability
\begin{equation}
    \overline{\rchi} \leq  \frac{(z-c)^2}{c^2+1}\,,
\end{equation}
{ where $c$ is defined in Eq.~\eqref{eq:c-cases}.}
Thus, under the assumption that $z\leq1$ and $\langle\Delta \widetilde{C}(\thv_{i,\varepsilon*},a_1\varepsilon)\rangle_i\geq0$, $\overline{\rchi}\leq 1$ for all of the above extrapolation strategies with access to 2 noise levels.
\end{proposition}

Proposition \ref{prop:ZNE} shows that under mild assumptions of the effect of the noise on the cost landscape, Zero-Noise Extrapolation with access to 2 noise levels impairs the resolvability of the cost landscape. These assumptions have physical meaning: $z\leq1$ implies that on average the cost concentrates when the noise parameter is boosted, whilst $\langle\Delta \widetilde{C}(\thv_{i,\varepsilon*},a_1\varepsilon)\rangle\geq0$ implies that the landscape is not heavily corrupted after boosting the noise parameter so that the global minimum at the base noise level remains below the average cost value. We also see that in the presence of exponential cost concentration and NIBPs, the relative resolvability is exponentially small if one attempts to directly reverse the exponential scaling of NIBPs.

In the Appendix, we study a modification of the averaged resolvability in Definition \ref{def:av_resolvabilityII} and find that this is bounded by a function of the purity of the noisy states, such that the resolvability decreases if purity decreases with increasing noise level.
This result, along with the proofs of the above propositions, can be found in Appendix \ref{sec:appdx_ZNE}. 
Finally, we remark that in the above results we consider a scenario where the Richardson, exponential or NIBP extrapolation strategies utilize expectation values from only two noise levels. In Appendix \ref{sec:appdx_ZNE3} we show that similar results may be obtained for Richardson extrapolation with access to 3 distinct noise levels.

\subsubsection{Virtual Distillation}\label{sec:VD}

Here we present our results on Virtual Distillation. Similar to our results for Zero-Noise Extrapolation, throughout this section we make the assumption that 
the statistical uncertainty of the measurement outcomes of the circuit that prepares $\Tr[\widetilde{\rho}^M_iO]$ are no smaller than that of $\Tr[\widetilde{\rho}_iO]$ for any choice of parameters $\thv_i$. In Appendix \ref{sec:appdx_VD} we provide some intuition for this assumption.

In the following proposition we start again with the simple model of global depolarizing noise.

\begin{proposition}[Relative resolvability of Virtual Distillation with global depolarizing noise]\label{prop:VD_depol}
Consider global depolarizing noise of the form in Eq.~\eqref{eq:gdepol_noise}
acting on some $n$-qubit pure state $\rho$ with error probability $p$. We consider the two error mitigation protocols of Ref.~\cite{koczor2020exponential} (denoted "A" and "B") to respectively prepare \eqref{eq:VD_A} and \eqref{eq:VD_B}, using $M$ copies of a quantum state. The relative resolvabilities to resolve any two arbitrary cost function points satisfy
\begin{align}
    \rchi_{depol}^{(A)} \leq \rchi_{depol}^{(B)} = \Gamma(n,M,p)\,,
\end{align}
for all $n\geq1$, $M\geq2$, $p \in [0,1]$, and where 
\begin{equation}
    \Gamma(n,M,p)\leq 1 \,,
\end{equation}
is a monotonically decreasing function in $M$ (with asymptotically exponential decay) for $n\geq 1$, $M\geq 2$. Within this region the bound is saturated as $\Gamma(1,2,p)= 1$ for all $p$.
\end{proposition}

Proposition \ref{prop:VD_depol} shows that Virtual Distillation decreases the resolvability of cost landscapes suffering from global depolarizing noise. Moreover, as the number of state copies $M$ increases, the effect worsens. We find similar results in the following proposition when considering averaged resolvabilities over a class of noisy states.

\begin{proposition}[Average relative resolvability of Virtual Distillation]\label{prop:VD_av}
Consider an error mitigation protocol that prepares estimator $C_m(\thv_i) = \Tr[\widetilde{\rho}_i^M O]/\Tr[\widetilde{\rho}_i^M]$ from some noisy parameterized quantum state $\widetilde{\rho}_i \equiv \widetilde{\rho}(\thv_i)$. Consider the average relative resolvability $\dbloverline{\rchi}_{\boldsymbol{\lambda}}$ for noisy states of some spectrum $\boldsymbol{\lambda}$ with purity $P$ as defined in Definition \ref{def:av_resolvability}. We have
\begin{equation}
    \dbloverline{\rchi}_{\boldsymbol{\lambda}}\leq G(n,M,P)\,\leq\,1\,,
\end{equation}
where $G(n,M,P)$ is a monotonically decreasing function in $M$ (with asymptotically exponential decay) for all $n\geq 1$, $M\geq 2$.
Within this region the bound is saturated as $G(1,2,P)=1$ for all $P$, and as $G(n,M,1)=1$ for all $n\geq 1, M\geq 2$.
\end{proposition}

We present the explicit forms of $\Gamma(n,M,p)$ and $G(n,M,P)$ as well as a proof of the above propositions in Appendix \ref{sec:appdx_VD}. In Appendix \ref{sec:appdx_VD_av} we also show that outside of the highly mixed regime, the bound in Proposition \ref{prop:VD_av} decreases with decreasing noisy state purity. This indicates that within such settings, the greater the loss of purity due to noise, the worse the impact on resolvability is after error mitigation with Virtual Distillation.

\subsubsection{Probabilistic Error Cancellation }\label{sec:QP}

Here we present our results for Probabilistic Error Cancellation. We utilize the optimal quasiprobability decompositions studied in Ref.~\cite{takagi2020optimal}, and the proofs can be found in Section \ref{sec:appdx_QP} of the Appendix.

\begin{proposition}[Relative resolvability of Probabilistic Error Cancellation under global depolarizing noise]\label{prop:QP_gdepol}
Consider a quasi-probability method that corrects global depolarizing noise of the form \eqref{eq:gdepol_noise}. For any pair of states corresponding to points on the cost function landscape, the optimal quasiprobability scheme gives
\begin{equation}
    \rchi_{depol} = \frac{2^{2n}}{2^{2n}-p(2-p)} \geq 1\,,
\end{equation}
for all $n\,\geq1,\,  p\in[0,1]$, which is achieved with access to noisy Pauli gates.           
\end{proposition}

Proposition \ref{prop:QP_gdepol} shows that for the special case of global depolarizing noise, Probabilistic Error Cancellation actually improves the resolvability of the noisy cost landscape. However, this improvement is generally small and is decreasing quickly with the number of qubits $n$. For instance, for $n=1$, $\rchi_{depol}$ has maximum value $4/3$ (achieved in the limit of maximum depolarization probability). {For} $n=2$, $\rchi_{depol}$ has maximum value $\approx 1.07$. In the limit of large $n$, $\rchi_{depol}$ tends to 1.

We extend this study to local depolarizing noise in Appendix \ref{sec:appdx_QP}. 
We find that for a single instance of local depolarizing noise, Probabilistic Error Correction can either improve resolvability or worsen it, depending on the strength of concentration of the cost. 
In addition, we show in the following proposition that if one wishes to mitigate all the noisy gates in the circuit and one has NIBP scaling, the improvement due to Probabilistic Error Cancellation degrades exponentially, and ultimately for large problem sizes this impairs resolvability.

\begin{proposition}[Scaling of Probabilistic Error Cancellation with local depolarizing noise]\label{prop:PEC}
Consider local depolarizing noise with depolarizing probability $p$ acting in $L$ layers through a depth $L$ circuit as in Eq.~\eqref{eq:noisystate}. Suppose that the effect of this noise is to cause cost concentration 
\begin{equation}\label{eq:QP_localscaling-assumption}
    \langle\Delta \widetilde{C}(\thv_{i,*})\rangle_i = Aq^L \langle\Delta C(\thv_{i,*})\rangle_i\,,
\end{equation}
for some constant $A$ and noise parameter $q\in [0,1)$. The optimal quasiprobability method to mitigate the depolarizing noise in the circuit yields 
\begin{equation}
    \overline{\rchi} = \frac{1}{A^2q^{2L}}\left(Q(p) \right)^{nL} \,,
\end{equation}
where $Q(p)=1 - \frac{3p(2-p)}{4-p(2-p)} \in [0,1)$ for $p\in (0,1]$.
\end{proposition}

{We note that it is known that noisy cost differences under local depolarizing noise are known to be at best as large as $\langle\Delta \widetilde{C}(\thv_{i,*})\rangle_i \propto (1-p)^L \langle\Delta C(\thv_{i,*})\rangle_i$ \cite{wang2020noise}. Thus,  } Eq.~\eqref{eq:QP_localscaling-assumption} gives the best possible scaling of noisy cost differences allowed by \eqref{eq:costconcentration}. {Proposition \ref{prop:PEC} shows that if this exponential scaling is no worse than Eq.~\eqref{eq:QP_localscaling-assumption}, then under local depolarizing noise the relative resolvability has unfavourable scaling with respect to system size, for any depth circuit.}

\subsubsection{Linear ansatz methods}\label{sec:CDR}

In Proposition \ref{prop:CDR} we consider a scenario where the same linear ansatz \eqref{eq:CDR_ansatz} is applied to two points on the noisy cost landscape. For Clifford Data Regression this is a reasonable assumption in scenarios where one is comparing two points that are close in parameter space, for instance, when a simplex-based optimizer is exploring a small local region. However, we remark this is not always true in general settings.

\begin{proposition}[Linear ansatz methods]\label{prop:CDR}
Consider any error mitigation strategy that mitigates noisy cost function value $\widetilde{C}(\thv)$ by constructing an estimator $C_m(\thv)$ of the form
\eqref{eq:CDR_ansatz}. For any two noisy cost function points to which the same ansatz is applied, we have
\begin{equation}
    \rchi = 1\,,
\end{equation}
for any noise process.
\end{proposition}

\begin{corollary}[Linear ansatz methods under global depolarizing noise]\label{cor:CDR}
Under global depolarizing noise, the optimal linear ansatz gives $\rchi = 1$ for any pair of cost function points.
\end{corollary}

Corollary \ref{cor:CDR} comes simply by noting that the optimal choice of linear ansatz under global depolarizing noise corrects the noise exactly and is state independent \cite{czarnik2020error}. The above results imply that in some settings CDR has a neutral effect on the resolvability of the cost function landscape. This opens up the possibility that in practical settings CDR can improve the trainability of cost landscapes, if it can remedy other cost corrupting effects due to noise outside of cost concentration. This motivates our numerical studies of CDR, which we present in the following section. 

\section{Numerical Results}\label{sec:numerics}

As discussed in Sec.~\ref{sec:protocol_specific}, in many settings, current state-of-the-art error mitigation methods do not mitigate the effects of cost concentration. Nevertheless, as discussed in Sec.~\ref{sec:costcorruption}, trainability of VQAs is also affected by other cost-corrupting effects. We expect that error mitigation can reverse some of the effects due to cost corruption that affect the trainability of VQAs when the effects of cost concentration are not too severe.
In this section, we numerically investigate the effects of error mitigation on trainability in such a setting to provide possible evidence towards beneficial effects of error mitigation. To this end, we focus on CDR as in some settings it does not worsen the effects of cost concentration, as shown in Sec.~\ref{sec:CDR}. {While we use this result as a guiding heuristic for the choice of the error mitigation method, we note that a direct comparison of the performance of various methods would be necessary to establish the optimal method for a particular optimization task.}

\begin{figure}[tb!]
    \includegraphics[width=\columnwidth]{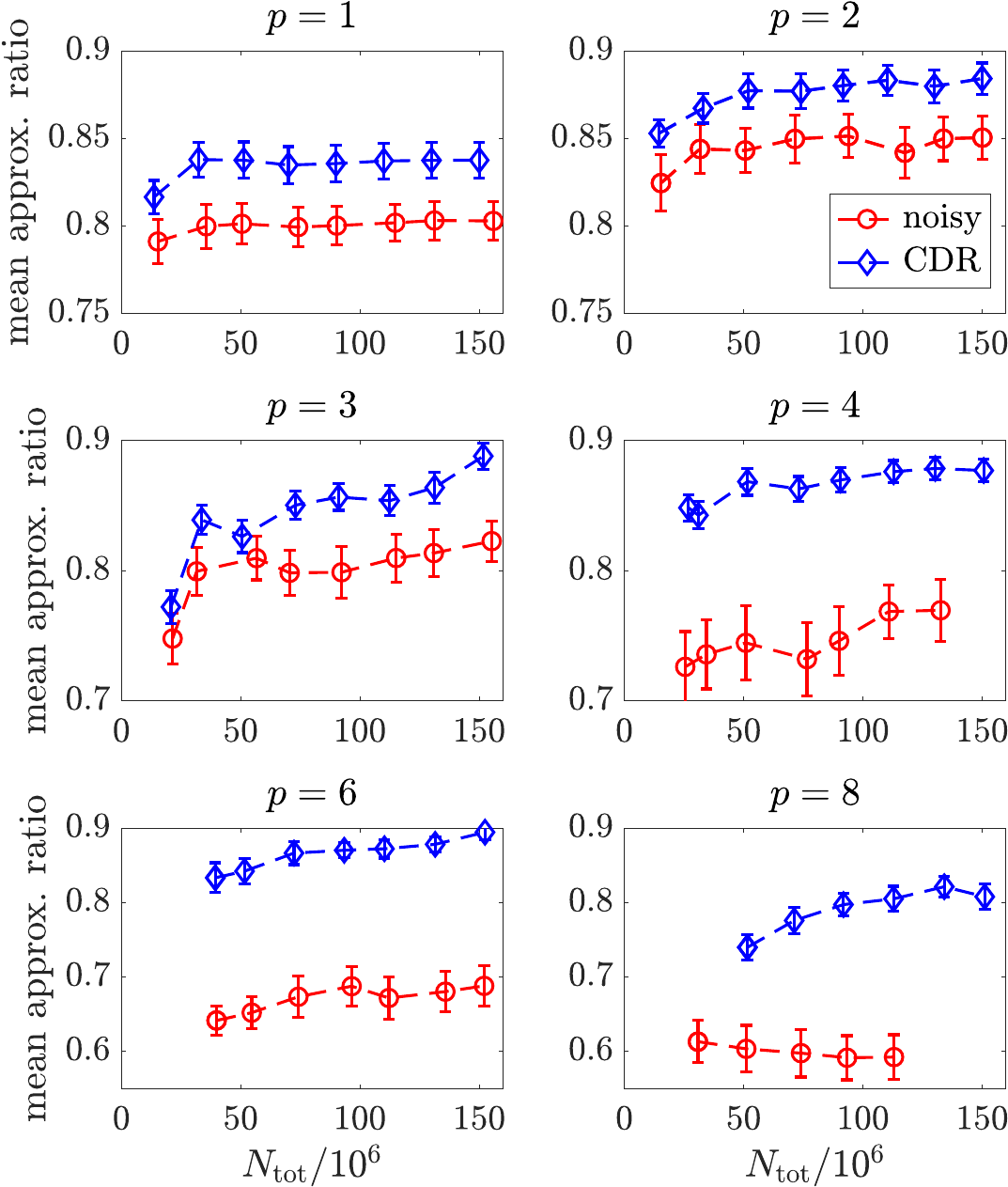}
    \caption{\textbf{Comparing CDR-mitigated and noisy optimization for  $\boldsymbol{5}$-qubit Max-Cut QAOA.} We plot the approximation ratio of solutions for noisy (red circles)  and CDR-mitigated  (blue diamonds) optimization of Max-Cut QAOA for $5$ qubits. Different panels show results for different numbers of QAOA rounds $p$ plotted versus total number of shots $N_{\rm tot}$ spent on the optimization of a MaxCut problem. Here, we compute the approximation ratios using exact $H_{\rm MaxCut}$ energies  to benchmark quality of the noisy and CDR-mitigated optimization. {The approximation ratio is defined as the ratio of a given solution's energy to the true ground state energy. A higher approximation ratio indicates better solution quality.}  For each $p$ we average  the approximation ratio over $36$ MaxCut graphs chosen randomly   from   Erd\"os-R\'enyi ensemble. {The error bars show a standard deviation of the mean computed as a standard deviation of the ratio for a graph sample divided by a square root of the number of graphs.} For all  $p$ and $N_{\rm tot}$ values we see an advantage of the CDR-mitigated optimization over noisy optimization.      }
    \label{fig:CDR-numerics}
\end{figure}

We perform our numerical experiments by simulating the Quantum Approximate Optimization Algorithm (QAOA)~\cite{farhi2014quantum} for 5-qubit {and 8-qubit} MaxCut problems.
{For $n=5$ we use a realistic noise model of an IBM quantum computer~\cite{cincio2021machine}, which has been obtained by gate set tomography of IBM's Ourense quantum device. In the case of $n=8$ we modify the noise model taking a convex combination of the IBM's Ourense and noiseless process matrices to  reduce  the  noise strength by a factor of $5$ with respect to the real device. This noise reduction was necessary to ensure trainability of the ansatz as for our problem larger $n$ implies   more layers of native gates in the optimized circuits and consequently larger cost concentration.} Furthermore, we assume here linear connectivity of the simulated quantum computer.

A MaxCut problem is defined for a graph $G=(V,E)$ of nodes $V$ and edges $E$. The problem is to find a bipartition of the nodes into two sets which maximizes the number of edges connecting the sets. This problem can be reformulated as finding the ground state of a  Hamiltonian
\begin{equation} 
H_{\rm MaxCut}= -\frac{1}{2} \sum_{ij\in E}   (\id - Z_i Z_j), 
\end{equation} 
where $Z_i,Z_j$ are Pauli $Z$ matrices.
Here we consider graphs with $n=5$ {($n=8$)} vertices, and with $36$ {($30$)} randomly generated instances, respectively. The instances are obtained according to the Erd\"os-R\'enyi model~\cite{erdos1959random}, where for each pair of vertices in the graph there is a connecting edge with probability $0.5$.

To  approximate the ground state of $H_{\rm MaxCut}$ we simulate the QAOA for number of rounds ranging from $p=1$ to $8$. The QAOA ansatz applied to the input state is given as
\begin{equation}
    \prod_{j=p,p-1\dots,1} e^{i \beta_j H_M} e^{i \gamma_j H_{\rm MaxCut}}  (|+\rangle)^{\otimes n}, 
\end{equation}
where $H_M =   \sum_{j} X_j$, $X_j$ are Pauli $X$ matrices, we denote  $\ket{+}=(\ket{0}+\ket{1})/\sqrt{2}  $ , and $\beta_j$, $\gamma_j$ are variational parameters. We minimize the cost function $\langle H_{\rm MaxCut} \rangle$ using the Nelder-Mead algorithm \cite{nelder1965simplex} {and choose the initial values of $\beta_j, \gamma_j$ randomly}. We perform the optimization with shot budgets ranging from $N_{\rm tot}=10^7$ to $1.5\times 10^8$ {for $n=5$  and from $N_{\rm tot}= 10^7$ to $7\times 10^7$ for $n=8$.}   We define $N_{\rm tot}$ as total number of shots spent on the optimization. We detail the optimization procedure in Appendix~\ref{appendix:numerics}.  In our numerics, the values of $N_{\rm tot}$ are chosen to enable  {an}  {optimization runtime} {which is feasible} with current quantum computers.\footnote{{To approximately estimate the time required to run
the optimization we assume a delay time of $250$~${\rm \mu s}$
in between shots which is a default setting of current IBM quantum
computers~\cite{wack2021quality}.  Furthermore, we assume that the circuit 
compilation and execution time is negligible. This assumption is 
justified for sufficiently shallow circuits and sufficiently
large numbers of shots per circuit~\cite{wack2021quality}. 
Under such assumptions for shot budgets used in our  5-qubit MaxCut QAOA 
numerics we obtain times ranging from $1$ to $14$~hours for $n=5$ and $14$ to $140$~hours for $n=8$.}}  To quantify the quality of the noisy or CDR-mitigated optimization we compute approximation ratios of the solutions using the exact expectation value of $\langle H_{\rm MaxCut} \rangle$. The approximation ratio is defined here as the ratio of a given solution's energy to the true ground state energy.

We gather our numerical results for CDR  {at $n=5$} in Fig.~\ref{fig:CDR-numerics}.  In the figure we plot the approximation ratio averaged over $36$ randomly chosen graphs versus $N_{\rm tot}$.  We compare the quality of the solutions of noisy (unmitigated) and CDR-mitigated optimization and find that  CDR-mitigated optimization outperforms noisy optimization for all considered $p$ and  $N_{\rm tot}$ values. We observe that the solutions for  $p=2$  outperform those for $p=1$ for both CDR-mitigated and noisy optimization. The quality of $p>2$ solutions decline with increasing $p$ for noisy optimization, while it remains approximately the same for $p=2$ to $6$ for CDR-mitigated optimization. With CDR-mitigated optimization we see a decrease in quality of solution for the largest considered $p=8$.

{In Fig.~\ref{fig:CDR-numerics-n8} we gather CDR results for $30$ instances of Erd\"os-Renyi   graphs at $n=8$ and $p=1-4$ plotting again the  approximation ratio averaged over instances versus $N_{\rm tot}$. Similar to the case of $n=5$ we typically see an advantage of the CDR mitigated optimization  over noisy optimization, although the improvement in approximation ratio is smaller than observed for $n=5$. This 
result underscores the need for more detailed investigation of the properties of an optimization problem which make it favorable for error-mitigated optimization. We leave such an investigation for future work. 
For the deepest $p=3,4$ ansatze we see that performance of both the noisy and CDR-mitigated optimization is degraded in comparison to the shallower $p=2$ case.}

The numerical results presented here are obtained for circuits shallow enough to be trainable while using the CDR-mitigated cost function. Therefore they are outside of the NIBP scaling regime.  As discussed in Section~\ref{sec:effectsofnoise}, even outside the NIBP regime noise may adversely impact trainability by corrupting the cost function landscape, which error mitigation has a chance to remedy. 
Our results give hope that  CDR-mitigated optimization  may overall offer a trainability advantage for problems with such cost function landscape corruption.

As  discussed in Section~\ref{sec:protocol_specific}  optimizing  an error mitigated cost function is not guaranteed to  outperform its noisy optimization even outside the NIBP regime. { Indeed, in Appendix~\ref{appendix:VD_numerics} we find  that for the  $p=2,4$  MaxCut graphs and moderate $N_{\rm tot}$ used here, VD-mitigated optimization does not outperform noisy optimization. In Appendix \ref{appendix:ZNE_numerics} we reach a similar conclusion for optimization mitigated with Zero-Noise Extrapolation when increasing noise strength digitally by widely-used CNOT identity insertions~\cite{giurgica2020digital}.}    

{We note that this conclusion may be problem-dependent. In particular, for a sufficiently large $N_{\rm tot}$, an idealized error mitigation method that perfectly corrects the expectation values in the limit of an infinite shot number should improve optimization quality as the cost corruption effects will become dominant. For a realistic case, error mitigation has a bias that may depend on a problem choice, noise, and even error mitigation method's implementation details. For example, it has been found that Zero-Noise Extrapolation's bias depends on the method of increasing the noise strength~\cite{kim2021scalable}. Therefore, we expect all those factors to be relevant when the shot number is large enough that the cost corruption limits VQA's performance. Consequently, caution is necessary when judging the power of a particular error mitigation approach compared to others in removing the optimization landscape corruption based on a few test cases.
}

\begin{figure}[t]
    \includegraphics[width=\columnwidth]{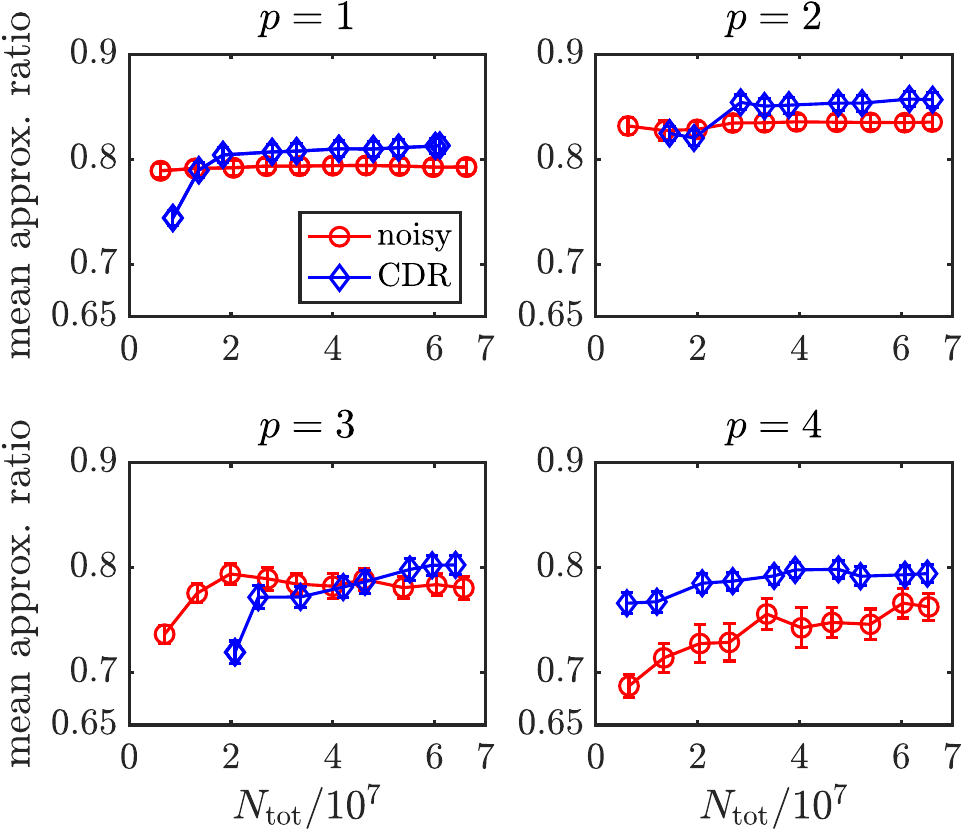}
    \caption{\textbf{ {Comparing CDR-mitigated and noisy optimization for  $\boldsymbol{8}$-qubit Max-Cut QAOA.}} {Similar to Fig.~\ref{fig:CDR-numerics},  here we plot the approximation ratio of solutions for noisy  and CDR-mitigated  optimization averaged over $30$ randomly chosen  graphs  from   Erd\"os-R\'enyi ensemble. The error bars are computed as in Fig.~\ref{fig:CDR-numerics}.} }    
    \label{fig:CDR-numerics-n8}
\end{figure}

\section{Discussion}\label{sec:discussions}

Noise can exponentially degrade the trainability of linear (or superlinear) depth Variational Quantum Algorithms (VQAs) by flattening the cost landscape, thus requiring an exponential precision in system size {(and therefore exponential shot budget)} to resolve its features \cite{wang2020noise, franca2020limitations}. This limits the scope for achieving possible quantum advantage with VQAs. At present there are no known strategies to avoid this exponential scaling completely aside from pursuing algorithms with sublinear circuit depth, and current strategies to mitigate this effect consist only of reducing hardware noise rates. Thus, {it is} a pressing challenge to search for possible solutions to this problem. Error mitigation strategies emerge as a natural candidate to tackle this problem under near-term constraints.

In this work we investigate the effects of error mitigation on the trainability of noisy cost function landscapes in two regimes.  {We note that despite the fact it is known that error mitigation increases shot budgets, it has been a priori unclear whether error mitigation strategies in general can have a positive or negative contribution towards this problem, as it relies on a careful balance of how effectively errors are mitigated, compared to how quickly statistical uncertainty is amplified.} First, we work in the asymptotic regime (in terms of scaling with system size) and find that if a VQA is suffering from exponential cost concentration, requiring an exponential number of shots to accurately resolve cost values, then a broad class of error mitigation strategies (including as special cases Zero-Noise Extrapolation, Virtual Distillation, Probabilistic Error Cancellation, Clifford Data Regression) cannot remove this exponential scaling. Within the considered paradigm, this exponential scaling implies that at least an exponential number of resources needs to be spent in order to extract accurate information from the cost landscape in order to find a cost-minimizing optimization direction. In Corollary \ref{cor1} we identify circuit samples (or shots) as well as number of copies of a quantum state as two such resources.

Second, we move out of the asymptotic regime and investigate whether or not particular error mitigation protocols can improve the resolvability of noisy cost landscapes. Should such a landscape be burdened with exponential cost concentration, this would correspond to an improvement in the coefficient in the exponential scaling. Our results indicate that some error mitigation protocols can worsen the resolvability, and ultimately the trainability, of cost landscapes in certain settings. In particular, in Propositions \ref{prop:VD_depol} and \ref{prop:VD_av} we show analytically that Virtual Distillation impairs resolvability with worsening resolvability as the number of state copies increases. We obtain similar results for Zero-Noise Extrapolation in Propositions \ref{prop:ZNE_depol} and \ref{prop:ZNE} under some assumptions of the cost landscape. Numerical analysis of a particular MaxCut problem {for a moderate shot number} indicates that trainability is impaired for Virtual Distillation {and} {for some Zero-Noise Extrapolation strategies}.

{For the considered problem,} Clifford Data Regression (CDR) distinguishes itself from the other error mitigation techniques considered in this article, as in contrast to the other protocols it does not necessarily increase the statistical uncertainty of cost values more than it reverses their concentration {(Probabilistic Error Cancellation also improves the resolvability under a global depolarizing noise assumption, but the scaling quickly deteriorates under a local depolarizing noise assumption)}. This is reflected in the fact that as it only uses a linear ansatz, CDR has neutral impact on resolvability (Proposition \ref{prop:CDR}). However, it is also known that CDR can remedy the effects of more complex noise models. This {suggests} that CDR could resolve trainability issues arising due to corruptions of the cost function outside of cost concentration, whilst having a neutral effect on cost concentration itself, and thus overall improve trainability. In the numerical example studied, presented in Fig.~\ref{fig:CDR-numerics}, we observe this to be the case. This points to deeper future work studying the mechanisms that allow error mitigation to improve the trainability of noisy cost landscapes. {Such work should consider a wide range of problems and state-of-the-art hardware implementations as the bias of error mitigation methods (which may limit reversing cost function corruption) has been shown to depend on the problem choice and details of the  {methods' implementation.}} \cite{kim2021scalable, cirstoiu2022volumetric}. {We note that here we have also disregarded the burden of training data for CDR, which is an important consideration.}

Finally, we identify that the broad class of error mitigation protocols we study in our asymptotic analysis all only consist of post-processing expectation values of noisy circuits, as summarized in Fig.~\ref{fig:em_framework}. This gives intuition as to why they cannot escape the exponential scaling of noise-induced barren plateaus (NIBPs). However, the theory of error correction indicates that with sufficient resources NIBPs can indeed be avoided. This gives hope that there can exist novel error mitigation strategies that move beyond the framework of the protocols considered in this article and thereby avoid the exponential impairment to trainability that NIBPs present.

\section{Code avaialability}

{Further implementation details are available from the
authors upon request.}

\section{Acknowledgements}

This work was supported by the Quantum Science Center (QSC), a National Quantum Information Science Research Center of the U.S. Department of Energy (DOE). SW was supported by the U.S. DOE, Office of Science, Office of Advanced Scientific Computing Research, under the Quantum Computing Application Teams~program. SW was also partially supported by the Samsung GRP grant. Piotr C. and AA were supported by the Laboratory Directed Research and Development (LDRD) program of LANL under project numbers 20190659PRD4 (Piotr C.) and 20210116DR (AA and Piotr. C). MC acknowledge support from the Center for Nonlinear Studies at Los Alamos National Laboratory. MC and LC were also initially supported by the LDRD program of LANL under project number 20190065DR. PJC also acknowledges initial support from the LANL ASC Beyond Moore's Law project. This research used resources provided by the Los Alamos National Laboratory Institutional Computing Program, which is supported by the U.S. Department of Energy National Nuclear Security Administration under Contract No. 89233218CNA000001.

\noindent \textit{Note added.} Concurrently, a related work, Ref.~\cite{takagi2021fundamental} appeared in the literature. {They establish a similar result to Theorem \ref{thm:fullmitigation}, bounding the sample overhead to mitigate local depolarizing noise in an $L$-layered circuit. More broadly, the authors establish fundamental bounds on the sampling overhead of error mitigation schemes. In contrast, the rest of our results focuses on the ability of error mitigation to aid optimization of variational quantum algorithms, which depends on a trade-off between sample overheads and mitigation ability. }

\bibliography{quantum.bib}

\bibliographystyle{quantum}

\onecolumn
\appendix

\section{Road map of appendices} 
In Appendix
\ref{appendix:preliminaries} we present some notation and definitions that we need in order to prove our main results, as well as provide further details on the error mitigation protocols studied in this article. In Appendix \ref{appendix:lemmas} we derive some useful lemmas that are required for our proofs. In Appendix \ref{sec:appdx_asymptotics} we present the proof for our asymptotic results on the exponential concentration of estimators. In Appendix \ref{appendix:protocol-specific} we present our protocol-specific results on the change in resolvability of the cost landscape under error mitigation. Finally, in Appendix  \ref{appendix:numerics} we discuss details of our numerical implementations for Clifford Data Regression, Virtual Distillation, {and Zero-Noise Extrapolation}.

\section{Preliminaries}\label{appendix:preliminaries}

\subsection{Further details on error mitigation techniques}\label{sec:appdx_review_EM}

In this section we expand on our discussion in Section \ref{sec:errormitigation} and provide further details on the Zero-Noise Extrapolation and Probabilistic Error Cancellation protocols.

\subsubsection{Zero-Noise Extrapolation (ZNE)}\label{sec:appdx_review_ZNE}

For convenience in this section we recall the key points of Zero-Noise Extrapolation as summarized in Section \eqref{sec:ZNE}. We also detail the explicit forms of the estimators that can be constructed for exponential extrapolation and an extrapolation strategy tailored towards NIBP effects, which will be required in order to prove our results.

\emph{Richardson Extrapolation.} We suppose that $\widetilde{C}(\boldsymbol{\theta}_i,\varepsilon)$ admits a Taylor expansion in small noise parameter $\varepsilon$ as
\begin{equation} \label{eq:appdx_ZNE_taylor}
    \widetilde{C}(\boldsymbol{\theta}_i,\varepsilon)=\widetilde{C}(\boldsymbol{\theta}_i,0)+\sum_{k=1}^{m} p_{k}(\boldsymbol{\theta}_i) \varepsilon^{k}+\mathcal{O}(\varepsilon^{m+1})\,,
\end{equation}
where $p_k$ are unknown parameters and $\widetilde{C}(\boldsymbol{\theta}_i,0)$ is the zero-noise cost function. By considering the equivalent expansion of $\widetilde{C}(\boldsymbol{\theta}_i,a_1 \varepsilon)$ and combining the two equations we can obtain 
\begin{equation}\label{eq:appdx_ZNE_TaylorCm}
    C^{(2)}_m(\boldsymbol{\theta}_i) = \frac{a_1 \widetilde{C}(\boldsymbol{\theta}_i,\varepsilon) - \widetilde{C}(\boldsymbol{\theta}_i,a_1 \varepsilon) }{ a_1-1} = \widetilde{C}(\boldsymbol{\theta}_i,0) +\mathcal{O}(\varepsilon^2)\,,
\end{equation}
which is a higher-order approximation of $\widetilde{C}(\boldsymbol{\theta}_i,0)$ compared to simply using $\widetilde{C}(\boldsymbol{\theta}_i,\varepsilon)$. This process can be repeated iteratively $m$ times to obtain an estimator which is accurate up to $\mathcal{O}(\varepsilon^{m+1})$ error. It can be shown that the general form for the estimator that uses $k$ noise levels can be written as
\begin{equation}\label{eq:appdx_ZNE_general_estimator}
    C^{(k)}_m(\boldsymbol{\theta}_i) = \sum^k_{j=0} \beta_j \widetilde{C}(\thv_i,a_j\varepsilon)\,,
\end{equation}
where the coefficients $\beta_j$ satisfy the linear system of equations $\sum_{j=0}^k\beta_j=1$ and $\sum_{l=0}^k\beta_la_l^t=0$ for all $t\in \{1,...,k \}$ \cite{sidi2003practical}.
For 3 noise levels, \eqref{eq:appdx_ZNE_general_estimator} explicitly gives
\begin{equation}\label{eq:appdx_ZNE_TaylorCm3}
    C^{(3)}_m(\boldsymbol{\theta}_i) = \frac{a_1a_2(a_2-a_1) \widetilde{C}(\boldsymbol{\theta}_i,\varepsilon) - a_2(a_2-1)\widetilde{C}(\boldsymbol{\theta}_i, a_1\varepsilon) + a_1(a_1-1)\widetilde{C}(\boldsymbol{\theta}_i, a_2\varepsilon) }{(a_1-1)(a_2-1)(a_2-a_1)} = \widetilde{C}(\boldsymbol{\theta}_i,0) +\mathcal{O}(\varepsilon^3)\,.
\end{equation}

\emph{Exponential extrapolation.} We can also consider an exponential model
\begin{equation}\label{eq:appdx_ZNE_exp}
    \widetilde{C}(\boldsymbol{\theta}_i,\varepsilon)= r(\boldsymbol{\theta}_i,\varepsilon)^{-t(\boldsymbol{\theta}_i,\varepsilon)} \left( 
    \sum_{k=0}^{m} p_{k}(\boldsymbol{\theta}_i) \varepsilon^{k}+\mathcal{O}(\varepsilon^{m+1}) \right)\,,
\end{equation}
for some $r$ and $t$, which in general can be functions of $\varepsilon$. Following in similar steps to Richardson extrapolation we can consider the same expansion at an augmented noise level $a_1 \varepsilon$. This enables us to construct the estimator
\begin{align}\label{eq:appdx_ZNE_expCm}
    C_m(\boldsymbol{\theta}_i) = \frac{1}{a_1-1} \Big( &a_1 r(\boldsymbol{\theta}_i,\varepsilon)^{t(\boldsymbol{\theta}_i,\varepsilon)} \widetilde{C}(\boldsymbol{\theta}_i,\varepsilon) \nonumber \\
    &-r(\boldsymbol{\theta}_i,a_1\varepsilon)^{t(\boldsymbol{\theta}_i,a_1\varepsilon)} \widetilde{C}(\boldsymbol{\theta}_i,a_1 \varepsilon) \Big)\,,
\end{align}
which approximates $\widetilde{C}(\boldsymbol{\theta}_i,0)$ to a higher order in $\varepsilon$ compared to $\widetilde{C}(\boldsymbol{\theta}_i,\varepsilon)$.

\emph{NIBP extrapolation.} We can construct an alternative Zero-Noise Extrapolation strategy that is tailored towards noisy cost function values that are dominated by NIBP scaling as in Eq.~\eqref{eq:costconcentration}. We model the effects of noise as
\begin{equation}\label{eq:appdx_ZNE_NIBP}
    \widetilde{C}(\boldsymbol{\theta}_i,q)= A + q^{L} \left( B(\boldsymbol{\theta}_i) +\sum_{k=1}^{m} p_{k} (1-q)^{k}+\mathcal{O}\left((1-q)^{m+1}\right) \right)\,,
\end{equation}
for all noisy cost function points $\widetilde{C}(\boldsymbol{\theta}_i,q)$, where $A = \widetilde{C}(\boldsymbol{\theta}_i,q=0)$ is the fixed point of the noise (corresponding to the maximally mixed state) and $A + B(\boldsymbol{\theta}_i)$ is the noise-free cost value. (Note that for NIBPs we cannot consider lower-order polynomials of $q$, as else the NIBP condition would be broken for small $q$.) 
We construct estimators for any given parameter $\boldsymbol{\theta}_i$, as 
\begin{equation}\label{eq:appdx_ZNE_Cm_NIBP}
    {C}_m(\boldsymbol{\theta}_i) = \frac{ a^{L+1}q^{-L} (\widetilde{C}(\boldsymbol{\theta}_i,q/a)-\widetilde{C}(\boldsymbol{\theta}_i,q=0)) -  q^{-L}(\widetilde{C}(\boldsymbol{\theta}_i,q)-\widetilde{C}(\boldsymbol{\theta}_i,q=0))}{a-1} + K\,,
\end{equation}
where $K=A-\sum_k p_k$ is an additive constant. As we are only interested in cost function differences for our results, this will cancel out.

\subsubsection{Probabilistic Error Cancellation}\label{sec:appdx_review_QP}

\emph{General idea.} Probabilisitic Error Cancellation utilizes many modified circuit runs in order to construct a quasiprobability representation of the noise-free cost function \cite{temme2017error, endo2018practical}. We assume that the effect of the noise can be described by a quantum channel $\mathcal{N}$ that occurs after a gate that we denote with unitary channel $\mathcal{U}$. For now we assume this is the only instance of noise in the circuit, however, we will later generalize to many instances of noise. The goal of Probabilistic Error Cancellation is to simulate the inverse map $\mathcal{N}^{-1}$. Note in general this will not always correspond to a CPTP map. Despite this, if we have a basis of (noisy) quantum channels $\{\mathcal{B}_\alpha \}_{\alpha}$, corresponding to experimentally available channels, we can expand the inverse map in this basis as 
\begin{equation}
    \mathcal{N}^{-1} = \sum_{\alpha} q_{\alpha}  \mathcal{B}_{\alpha} \,,
\end{equation}
for some set of $q_{\alpha} \in \mathbb{R}$. Then, the channel that describes the noiseless gate can be written as 
\begin{align}
   \mathcal{U} &= \NC^{-1}\circ \NC\circ \UC \\
   &=\sum_{\alpha} q_{\alpha} \mathcal{K}_{\alpha} \,,\label{eq:appdx_QuasiNoiselessgate}
\end{align}
where we have defined $\mathcal{K}_{\alpha} = \mathcal{B}_{\alpha}\circ\mathcal{N}\circ\mathcal{U}$. Denote the input state to the gate as $\rho_{in}$ and a measurement operator as $O$. The expectation value can be written 
\begin{equation}\label{eq:appdx_QuasiEV}
    {C}_{\mathcal{U}(\rho)} = \Tr \left[ \mathcal{U}(\rho_{in})O\right] = \sum_{\alpha} q_{\alpha} \widetilde{C}_{\mathcal{K}_{\alpha}(\rho)} \,,
\end{equation}
where for simplicity we first assume that $\mathcal{U}$ is the only gate in the circuit, and $\widetilde{C}_{\mathcal{K}_{\alpha}(\rho)} \equiv \Tr \left[ \mathcal{K}_{\alpha}(\rho)O\right]$. Finally, we can explicitly define a probability distribution $p_{\alpha} = |q_{\alpha}|/G_\mathcal{N}$ where $G_\mathcal{N} = \sum_{\alpha} |q_{\alpha}|$. This gives us an alternative way to write \eqref{eq:appdx_QuasiNoiselessgate} and \eqref{eq:appdx_QuasiEV} as
\begin{align}
    \mathcal{U} &= G_\mathcal{N} \sum_{\alpha} \mathrm{sgn}(q_{\alpha})\, p_{\alpha}\, \mathcal{K}_{\alpha} \,,\label{eq:appdx_QuasiNoiselessgate2} \\
    {C}_{\mathcal{U}(\rho)} &= G_\mathcal{N} \sum_{\alpha} \mathrm{sgn}(q_{\alpha})\, p_{\alpha}\, \widetilde{C}_{\mathcal{K}_{\alpha}(\rho)} \,,
\end{align}
where $\mathrm{sgn}(q_{\alpha})$ denotes the sign of $q_{\alpha}$. We call this the quasiprobability representation of the gate $\mathcal{U}$. The idea is that if we have access to the set of CPTP maps $\{ \mathcal{B}_{\alpha}\}_\alpha$ in our noisy native hardware gate set, then we can obtain an estimate of the noiseless expectation value ${C}_{\mathcal{U}(\rho)}$ as follows: (1) With probability $p_\alpha$, prepare the circuit corresponding to $\mathcal{K}_{\alpha}(\rho)$ and obtain $\widetilde{C}_{\mathcal{K}_{\alpha}(\rho)}$. (2) Multiply the measurement result by $\mathrm{sgn}(q_{\alpha})G_\mathcal{N}$. (3) Repeat process many times and sum results.

\emph{Correcting many gates.} So far we have only considered a circuit with a single gate $\mathcal{U}$. We can generalize \eqref{eq:appdx_QuasiNoiselessgate2} to a general circuit $\prod_k^{N_g} \mathcal{U}_k$ with ${N_g}$ gates with the quasiprobability representation
\begin{equation}\label{eq:appdx_QuasiCircuit}
    \prod_k \mathcal{U}_k = G^{tot}_\mathcal{N} \sum_{\vec{i}} \mathrm{sgn}(q_{\vec{i}}) p_{\vec{i}} \mathcal{K}_{\vec{i}} \,,
\end{equation}
where $G^{tot}_\mathcal{N} = \prod_k G_k$,\; $\vec{i}=(i_1, ..., i_{N_g})$,\; $q_{\vec{i}} = \prod_k q_{i_k}$,\; $p_{\vec{i}} = \prod_k p_{i_k}$,\; $\mathcal{K}_{\vec{i}} = \prod_k \mathcal{K}_{i_k}$. Thus, a similar procedure can be carried out in order to mitigate the noise on each individual gate in the circuit.

\section{Useful Lemmas}\label{appendix:lemmas}

\subsection{Noise Induced Cost Concentration}

\begin{lemma}\label{lem:NIBP}
Consider a parameterized noisy cost function $\widetilde{C}(\boldsymbol{\theta})=\Tr\left[\widetilde{\rho}(\boldsymbol{\theta})O\right]$, where $\widetilde{\rho}(\boldsymbol{\theta})$ is an $n$-qubit noisy state given by Eq.~\eqref{eq:noisystate} and $\thv \in \Theta$ is drawn from some set of accessible parameters $\Theta$.  Suppose the cost is suffering from exponential cost concentration according to Ref.~\emph{\cite{wang2020noise}}, that is
\begin{equation}\label{eq:appdx_costconcentration}
    \left| \widetilde{C}(\boldsymbol{\theta}) - \frac{\Tr[O]}{2^n} \right| \leq q^L B(n)\left\|\rho_{in} - \frac{ \id}{2^n} \right\|_1 \,,
\end{equation}
for all $\thv \in \Theta$, where $\rho_{in}$ is the input state in Eq.~\eqref{eq:noisystate}, $0\leq q<1$ is some noise parameter, $B(n) \in \poly(n)$, {$L$ is the number of layers of gates}, and $\|\cdot\|_1$ denotes the Schatten 1-norm (trace norm).  Then, $\exists A(n) \in \poly(n)$ such that
\begin{equation}
    \Big|\widetilde{C}(\boldsymbol{\theta}_1,q) - \widetilde{C}(\boldsymbol{\theta}_2,q)\Big| \leq A(n) q^L\,,
\end{equation}
for any two sets of parameters $\thv_1, \thv_2 \in \Theta$.
\end{lemma}
\begin{proof}
Starting from Eq.~\eqref{eq:appdx_costconcentration} can simply write
\begin{equation}
      \Big|\widetilde{C}(\boldsymbol{\theta}_1,q) - \widetilde{C}(\boldsymbol{\theta}_2,q)\Big|\, \leq\, 4B(n) q^{L} \,,
\end{equation}
for any two sets of parameters $\boldsymbol{\theta}_1, \boldsymbol{\theta}_2$, where we have used the triangle inequality in $1D$ and the fact that the trace distance has a maximum value of 1.
\end{proof}

For the next lemma we will consider the $n$-qubit channel
    \begin{equation}\label{eq:appdx_noiseunitarylayers}
      \mathcal{W} = \mathcal{U}_k\circ \mathcal{N}\circ \cdots \circ \mathcal{N} \circ \mathcal{U}_2 \circ \mathcal{N} \circ \mathcal{U}_1\circ\mathcal{N}\,,
    \end{equation}
where $\{\UC_k\}^L_{k=1}$ denote unitary channels that describe collections of gates that act together in a layer, and $\NC=\bigotimes_{i=1}^n \NC_i$ is an instance of local Pauli channels, such that action of $\mathcal{N}_j$ on a local Pauli operator $\sigma\in\{X,Y,Z\}$  can be expressed as 
\begin{equation}\label{eq:appdx:noisemodel}
    \mathcal{N}_j(\sigma)=q^{(j)}_{\sigma}\sigma\,,
\end{equation}
where we assume $-1< q^{(j)}_X,q^{(j)}_Y,q^{(j)}_Z<1$ for all qubit labels $j$. We characterize the noise strength with a single parameter $q=\max_{j,\sigma}\big\{\big|q^{(j)}_\sigma \big|\big\}<1$.

\begin{lemma}\label{lem:tracedistance}
Consider $\mathcal{W}$ as defined in Eq.~\eqref{eq:appdx_noiseunitarylayers} acting on some input state $\rho_{in}$. Then we have    
\begin{equation}
    \left\| \mathcal{W}(\rho_{in}) - \frac{\id}{2^n} \right\|_1 \leq q^{{c}k} n^{1/2} \sqrt{2\ln 2} \,,
\end{equation}
{where $c = 1/(2\ln 2)$ is a constant}.
\end{lemma}
\begin{proof}
We have 
\begin{align}
    \left\| \mathcal{W}(\rho_{in}) - \frac{{\id}}{2^n} \right\|_1 &\leq \sqrt{2 \ln 2\, D\Big(\WC(\rho_{in})\Big\|\frac{\id}{2^n}\Big)} \\
    &\leq \sqrt{2 \ln 2\,q^{2{c}k} D\Big(\rho_{in}\Big\|\frac{\id}{2^n}\Big)} \\
    &\leq q^{{c}k} n^{1/2}\sqrt{2\ln 2}\,,
\end{align}
where $D(\cdot\|\cdot)$ denotes the relative entropy. The first inequality is Pinsker's \cite{ohya2004quantum}, and the second inequality comes from a direct application of Supplementary Lemma 6 in Ref.~\cite{wang2020noise} {(adapted from Corollary 5.6 of Ref.~\cite{hirche2020contraction})}.
\end{proof}

\subsection{Averages over unitary 2-designs}

\begin{lemma}\label{lem:expvalues}
Consider the cost function value $C_\sigma(U)=\Tr[U\sigma U^\dag O]$ where $U$ is a $d\times d$ unitary matrix and $\sigma\in S(\HC)$ is some quantum state. Consider expectation values over $U_i\in Y$ where $Y\subset\mathcal{U}(d)$ is a unitary 2-design and $\mathcal{U}(d)$ is the unitary group of degree $d$. Denote such expectation values as $\langle \cdot \rangle_U$.
Then, we have
\begin{align}
    &\langle C_\sigma \rangle_U = \frac{1}{d} \Tr[O]\,, \\
    &\langle C_\rho C_\sigma \rangle_U = \frac{\Tr[O^2]\left(d\Tr[\rho\sigma] -  1 \right) - \Tr[O]^2\left(\Tr[\rho\sigma]-d \right)}{d(d^2-1)} \,,
\end{align}
{for any two operators $\sigma,\rho \in B(\mathcal{H})$ which satisfy $\Tr[\rho] = \Tr[\sigma]=1$, $\mathrm{dim}(\mathcal{H})=d$}. This implies 
\begin{equation}
    \Var[C_\sigma] = \frac{\left(\Tr[O^2]- \frac{1}{d}\Tr[O]^2\right)  \left(\Tr[\sigma^2]-\frac{1}{d} \right)}{d^2-1} \,.
\end{equation}
\end{lemma}
\begin{proof}
We use the following standard expressions for integration with respect to the Haar measure over the unitary group of degree $d$:
\begin{align}
    \int_{\mathcal{U}(d)} d \mu(W) w_{i, j} w_{p, k}^{*} &= \frac{\delta_{i, p} \delta_{j, k}}{d}\,, \\
    \int_{\mathcal{U}(d)} d \mu(W) w_{i_{1}, j_{1}} w_{i_{2}, j_{2}} w_{i_{1}', j_{1}'}^{*} w_{i_{2}', j_{2}'}^{*} &=\frac{1}{d^{2}-1}\left(\delta_{i_{1}, i_{1}'} \delta_{i_{2}, i_{2}'} \delta_{j_{1}, j_{1}'} \delta_{j_{2}, j_{2}'}+\delta_{i_{1}, i_{2}'} \delta_{i_{2}, i_{1}'} \delta_{j_{1}, j_{2}'} \delta_{j_{2}, j_{1}'}\right) \\
    &\quad\, -\frac{1}{d\left(d^{2}-1\right)}\left(\delta_{i_{1}, i_{1}'} \delta_{i_{2}, i_{2}'} \delta_{j_{1}, j_{2}'} \delta_{j_{2}, j_{1}'}+\delta_{i_{1}, i_{2}'} \delta_{i_{2}, i_{1}'} \delta_{j_{1}, j_{1}'} \delta_{j_{2}, j_{2}'}\right)\,, \nonumber
\end{align}
where $w_{i,j}$ are the matrix elements of the unitary $W \in \mathcal{U}(d)$. Then, the expectation values over 2-designs can be evaluated as
\begin{equation}
    \langle C_\sigma \rangle_U = \frac{1}{d}  \Tr[O]\,,
\end{equation}
and
\begin{align}
    \langle C_\rho C_\sigma \rangle_U =& \frac{1}{d^2-1}\Tr[O]^2 + \frac{1}{d^2-1}\Tr[\rho\sigma]\Tr[O^2]\\ &\ - \frac{1}{d(d^2-1)}\Tr[\rho\sigma]\Tr[O]^2 - \frac{1}{d(d^2-1)}\Tr[O^2]\,, \nonumber
\end{align}
where we have used $\Tr[\rho]=\Tr[\sigma]=1$. The final statement comes by noting that 
\begin{align}
    \Var[C_\sigma] &= \langle C_\sigma^2 \rangle_U - \langle C_\sigma \rangle_U^2 \\ 
    &= \frac{\Tr[O^2]\left(d\Tr[\sigma^2] -  1 \right) - \Tr[O]^2\left(\Tr[\sigma^2]-d \right)}{d(d^2-1)} - \frac{1}{d^2} \Tr[O]^2 \\
    &= \frac{\Tr[O^2]\left(d\Tr[\sigma^2] - 1 \right)}{d(d^2-1)} - \frac{\Tr[O]^2\left(d\Tr[\sigma^2]-1 \right)}{d^2(d^2-1)}\,,
\end{align}
which can be factorized to give the desired result.
\end{proof}

\section{Exponential estimator concentration}\label{sec:appdx_asymptotics}

\setcounter{theorem}{0}
\setcounter{proposition}{0}
\setcounter{corollary}{0}

We present a proof of Theorem \ref{thm:fullmitigation}, and restate the result here for convenience.
\begin{theorem}\label{thm:appdx_fullmitigation}
Consider an error mitigation protocol {that} prepares the quantity
\begin{align}\label{eq:appdx_thm1_estimator} 
    E_{\sigma,X,M,k} = \Tr\left[X \left(\sigma^{\otimes M}\otimes \ket{0}\bra{0}^{\otimes k}\right)\right]\,,
\end{align}
for some quantum state $\sigma \in S(\mathcal{H})$, for $\ket{0}\bra{0} \in S(\mathcal{H}')$ and for some $X\in B(\mathcal{H}^{\otimes M}\otimes \mathcal{H}'^{\otimes k})$. We suppose $\sigma$ is prepared with a depth $L_{\sigma}$ circuit and experiences noise according to Eq.~\eqref{eq:noisystate}. Under these conditions, $E_{\sigma,X,M,k}$ exponentially concentrates on a state-independent fixed point in the depth of the circuit as
\begin{align}
    \left\vert E_{\sigma,X,M,k} - \Tr\left[X \left( \frac{\id^{\otimes M}}{2^{Mn}} \otimes \ket{0}\bra{0}^{\otimes k}\right)\right]\right\vert\, &\leq\, G_{\sigma,X,M}(n) \,, \label{eq:prop1_conc_copy} 
\end{align}
where $\id \in S(\mathcal{H})$ is the $n$-qubit identity operator and 
\begin{equation}\label{eq:appdx_thm1_G}
    G_{\sigma,X,M}(n) =  \sqrt{\ln 4}\,\big\|X\big\|_\infty M n^{1/2} q^{{c}(L_{\sigma}+1)}\,,
\end{equation}
with noise parameter $q\in[0,1)$ {and constant $c = 1/(2\ln 2)$}. 
\end{theorem}

\begin{proof}
Consider
\begin{align}
    \left| \Tr\Big[ \left(\sigma^{\otimes M}\otimes \ket{0}\bra{0}^{\otimes k}\right) X\Big]- \Tr \left[\left(\frac{\id}{2^n}\otimes\sigma^{\otimes {M-1}} \otimes \ket{0}\bra{0}^{\otimes k}\right)\, X \right]\right|  &= \left|\Tr \left[\left(\big(\sigma-\frac{\id}{2^n}\big)\otimes\sigma^{\otimes M-1}\otimes \ket{0}\bra{0}^{\otimes k}\right) X \right]\right|  \\
    &\leq \left\| \big(\sigma-\frac{\id}{2^n}\big)\otimes\sigma^{\otimes M-1} \otimes \ket{0}\bra{0}^{\otimes k} \right\|_1 \big\|X\big\|_\infty \\
    &=  \Big\| \sigma-\frac{\id}{2^n} \Big\|_1 \Tr[\sigma]^{M-1} \Tr[\ket{0}\bra{0}]^k \big\|X\big\|_\infty\\ \label{eq:ancillagoeshere}
    &\leq q^{{c}(L_{\sigma}+1)} n^{1/2} \sqrt{2\ln 2}\, \big\|X\big\|_\infty\,.
\end{align}

The first inequality is due to the matrix Hölder's inequality, and the second inequality follows from Lemma \ref{lem:tracedistance}. Similarly, we have $M-1$ further such equations, which we display with the original equation:

\begin{align}
    \Tr\Big[ \left(\sigma^{\otimes M}\otimes \ket{0}\bra{0}^{\otimes k}\right) X\Big] - \Tr \left[\left(\frac{\id}{2^n}\otimes\sigma^{\otimes {M-1}} \otimes \ket{0}\bra{0}^{\otimes k}\right)\, X \right]
    &\leq q^{{c}(L_{\sigma}+1)} n^{1/2} \sqrt{2\ln 2}\, \big\|X\big\|_\infty\,, \\
    \Tr \left[\left(\frac{\id}{2^n}\otimes\sigma^{\otimes {M-1}} \otimes \ket{0}\bra{0}^{\otimes k}\right)\, X \right] - \Tr \left[\left(\frac{\id}{2^n}\otimes \frac{\id}{2^n}\otimes \sigma^{\otimes{M-2}}\otimes \ket{0}\bra{0}^{\otimes k} \right)\, X \right]
    &\leq q^{{c}(L_{\sigma}+1)} n^{1/2} \sqrt{2\ln 2}\, \big\|X\big\|_\infty\,,\\
    ...\quad\quad& \nonumber\\ 
    \Tr \left[\left(\Big(\frac{\id}{2^n}\Big)^{\otimes M-1}\otimes\sigma\otimes \ket{0}\bra{0}^{\otimes k}\right)\, X \right] - \Tr \left[\left( \Big(\frac{\id}{2^n}\Big)^{\otimes M}\otimes \ket{0}\bra{0}^{\otimes k} \right)\, X \right]
    &\leq q^{{c}(L_{\sigma}+1)} n^{1/2} \sqrt{2\ln 2}\, \big\|X\big\|_\infty\,.
\end{align}
The summation of these equations gives
\begin{equation}
    \Tr\Big[ X \left(\sigma^{\otimes M}\otimes \ket{0}\bra{0}^{\otimes k}\right) \Big] - \Tr\left[X \left( \frac{{\id}^{\otimes M}}{2^{Mn}} \otimes \ket{0}\bra{0}^{\otimes k}\right)\right]
    \leq M q^{{c}(L_{\sigma}+1)} n^{1/2} \sqrt{2\ln 2}\, \big\|X\big\|_\infty\,,
\end{equation}
which gives the desired bound. 
\end{proof}

We now present a more detailed version of Corollary \ref{cor1} in the main text, which explains how one can spend an exponential number of resources in different ways in order to resolve concentrated cost values.

\begin{corollary}[Exponential estimator concentration]
Consider an error mitigation protocol that approximates the noise-free cost value $C(\thv)$ by estimating the quantity 
\begin{equation}\label{eq:appdx_cor1_estimator}
    C_m(\thv) = \sum_{(\sigma(\boldsymbol{\theta}),X,M,k)\in T} a_{X,M,k}E_{\sigma(\boldsymbol{\theta}),X,M,k}\,,
\end{equation}
where each $E_{\sigma,X,M,k}$ takes the form \eqref{eq:appdx_thm1_estimator}. We denote $M_{max}$ and $a_{max}$ as the maximum values of $M$ and $a_{X,M,k}$ respectively accessible from a set $T$ defined by the given protocol. Assuming $\|X\|_\infty \in \OC(\poly(n))$, there exists a fixed point $F$ independent of $\boldsymbol{\theta}$ such that
\begin{equation}\label{eq:appdx_cor1_estimator_conc}
    \big|C_m(\boldsymbol{\theta})-F\big|\in\OC(2^{-\beta n}a_{max}|T|M_{max})\,,
\end{equation}
for some constant $\beta\geq 1$ if the circuit depths satisfy 
\begin{equation}
    L_{\sigma(\boldsymbol{\theta})} \in \Omega(n)\,,
\end{equation}
for all $\sigma(\boldsymbol{\theta})$ in the construction \eqref{eq:appdx_cor1_estimator}. That is, if the depth of the circuits scale linearly or greater then one requires at least exponential resources to distinguish $C_m$ from its fixed point, for instance in one of the following ways:
\begin{itemize}
    \item $a_{max}|T|M_{max} \in \OC(\poly(n))$ and an exponentially large number of shots are used to distinguish two quantities with exponentially small separation
    \item $a_{max}|T| \in \Omega(2^{\beta' n})$ for some constant $\beta' \geq 1$ and an exponentially large number of shots are required to distinguish two quantities with exponentially large statistical uncertainty, as measurement outcomes are multiplied by $a_{max}|T|$.
    \item $M_{max} \in \Omega(2^{\beta'' n})$ for some constant $\beta'' \geq 1$ and an exponentially large number of copies of some quantum state $\sigma$ are required.
\end{itemize}
\end{corollary}

\begin{proof}
Explicitly applying the results of Theorem \eqref{thm:appdx_fullmitigation} to the construction of $C_m(\boldsymbol{\theta})$ in \eqref{eq:appdx_cor1_estimator} we have
\begin{align}
    \left|C_m(\boldsymbol{\theta})-\sum_{(\sigma(\boldsymbol{\theta}),X,M,k)\in T} a_{X,M,k}\,\Tr\left[X \left( \frac{\id^{\otimes M}}{2^{Mn}} \otimes \ket{0}\bra{0}^{\otimes k}\right)\right] \right| &\leq \sum_{(\sigma(\boldsymbol{\theta}),X,M,k)\in T} a_{X,M,k}\, G_{\sigma(\thv),X,M}(n)\\
    & \hspace{-5em} \in \OC\Bigg(\sum_{(\sigma(\boldsymbol{\theta}),X,M,k)\in T} a_{X,M,k} \big\|X\big\|_\infty M n^{1/2} q^{{c}(L_{\sigma(\thv)}+1)} \Bigg)\,,
\end{align}
where in the second line we have used \eqref{eq:appdx_thm1_G}. If $L_{\sigma(\boldsymbol{\theta})} \in \Omega(n)$ then $q^{{c}(L_{\sigma(\thv)}+1)}\in\OC(2^{-\beta(\thv)n})$ for some $\beta(\thv)\geq1$. Thus, we can write 
\begin{equation}
    \left|C_m(\boldsymbol{\theta})-\sum_{(\sigma(\boldsymbol{\theta}),X,M,k)\in T} a_{X,M,k}\,\Tr\left[X \left( \frac{\id^{\otimes M}}{2^{Mn}} \otimes \ket{0}\bra{0}^{\otimes k}\right)\right] \right| \in \OC(2^{-\beta n}a_{max}|T|M_{max})\,,
\end{equation}
as required, where we can denote $\beta = \min_{\thv}\beta(\thv)$ and the fixed point as $F$, noting that $F$ is indeed parameter independent. From here, we can inspect the three presented cases: 
\begin{itemize}
    \item If $a_{max}|T|M_{max} \in \OC(\poly(n))$ then $C_m$ has exponentially small separation from $F$.
    \item There exists choice $\beta' \geq 1$ such that if $a_{max}|T| \in \Omega(2^{\beta' n})$ such that $C_m$ is not exponentially concentrated on $F$, however, $C_m$ now has an exponentially large statistical uncertainty, as measurement outcomes are multiplied by coefficients of order $a_{max}|T|$.
    \item There exists choice of $\beta'' \geq 1$ such that $M_{max} \in \Omega(2^{\beta'' n})$ and $C_m$ is not exponentially concentrated on $F$.
\end{itemize}
\end{proof}

\section{Protocol-specific results}\label{appendix:protocol-specific}

\subsection{Zero-Noise Extrapolation}\label{sec:appdx_ZNE}

In this section we present our results for Zero-Noise Extrapolation. As discussed in Section \ref{sec:ZNE} of the main text, we will consider a Richardson extrapolation strategy based on Eq.~\eqref{eq:appdx_ZNE_taylor}, an exponential extrapolation strategy based on Eq.~\eqref{eq:appdx_ZNE_exp} and a NIBP extrapolation strategy based on Eq.~\eqref{eq:appdx_ZNE_Cm_NIBP}. As we deal with two types of noise parameters, throughout this section we will adopt the unifying notation
\begin{equation}\label{eq:appdx_ZNE_notation}
    \widetilde{C}(\thv,a) = \begin{cases}
    \widetilde{C}(\thv,a\varepsilon) &\text{for Richardson/exponential extrapolation} \\
    \widetilde{C}(\thv,q/a) \quad &\text{for NIBP extrapolation}\,,
    \end{cases}    
\end{equation}
for all $a\geq 1$. Thus, $\widetilde{C}(\thv,a)$ denotes the noisy cost value at the boosted noise level, and $\widetilde{C}(\thv,1)$ denotes the noisy cost value at the base noise level.

As stated in the main text, in order to estimate the sample cost of error mitigation we will make the key assumption that
\begin{equation}\label{eq:appdx_ZNE_assumption}
    \Var[\widetilde{C}(\thv,a)] \geq \Var[\widetilde{C}(\thv,1)]\,,
\end{equation}
for all $\thv$ and $a\geq 1$ that is, the statistical fluctuations in measurement outcomes at the boosted noise level are no smaller than that at the base noise level. 
Indeed, for noise models with a maximally mixed fixed point, we expect that high noise rates will tend to lead to larger variances. For example, in the simple scenario of a local Pauli measurement, the variance of measurement outcomes takes the form $\frac{(1-p_0)(p_0)}{N}$, where $p_0$ is the probability of obtaining a $"0"$ outcome and $N$ is the number of shots. This variance is maximized for $p_0=\frac{1}{2}$.

\subsubsection{Relative resolvability under global depolarizing noise}

\begin{proposition}[Relative resolvability of Zero-Noise Extrapolation with global depolarizing noise, 2 noise levels]
Consider a circuit with $L$ instances of global depolarizing noise of the form \eqref{eq:gdepol_noise}.
Consider a Richardson extrapolation strategy based on Eq.~\eqref{eq:appdx_ZNE_taylor}, an exponential extrapolation strategy based on Eq.~\eqref{eq:appdx_ZNE_exp} and a NIBP extrapolation strategy based on Eq.~\eqref{eq:appdx_ZNE_NIBP}. We presume access to an augmented noisy circuit where the error probability is exactly increased by factor $a_1>1$ as $p\rightarrow a_1p$. Then we have
\begin{align}
    &{\rchi_{depol}}\leq \frac{\left(c-\frac{(1-a_1p)^L}{(1-p)^L}\right)^2}{c^2+1}\,,
\end{align}
where $\rchi_{depol}$ is the relative resolvability (see Definition \ref{def:resolvability}) for global depolarizing noise, and where 
\begin{equation}
    c = \begin{cases}
    \quad a_1\quad &\text{for Richardson extrapolation,} \\
    \frac{a_1r(\varepsilon)^{t(\varepsilon)}}{r(a_1\varepsilon)^{t(a_1\varepsilon)}} \quad &\text{for exponential extrapolation,} \\
    \; a_1^{-(L+1)} \quad &\text{for NIBP extrapolation}\,.
    \end{cases}
\end{equation}
Thus, $\rchi_{depol}\leq 1$ for all of the above extrapolation strategies with access to 2 noise levels.
\end{proposition}

\begin{proof}
Upon inspecting Eqs.~\eqref{eq:appdx_ZNE_TaylorCm}, \eqref{eq:appdx_ZNE_expCm} and \eqref{eq:appdx_ZNE_Cm_NIBP}, one can verify that the Richardson, exponential and NIBP extrapolation strategies all take the form
\begin{equation}\label{eq:ZNE_generalCm}
    C_m(\thv) = \frac{A\cdot\widetilde{C}(\thv,1)-B\cdot\widetilde{C}(\thv,a)}{D} + E\,,
\end{equation}
where $A,B\geq 0$ (note that for NIBP extrapolation $E$ contains the state-independent cost value that represents the fixed point of the noise) and where we have adopted the notation defined in \eqref{eq:appdx_ZNE_notation}. We note that under $L$ instances of global depolarizing noise (of the form \eqref{eq:gdepol_noise}) with error probability $p$, noisy cost differences are given by 
\begin{equation}
    \Delta \widetilde{C}(a) = {(1-ap)^L} \Delta C\,,
\end{equation}
for any pair of cost function points, where $\Delta C$ is the corresponding noise-free cost difference. 

The error-mitigated cost function difference $\Delta C_m = {C}_m(\boldsymbol{\theta}_1) - {C}_m(\boldsymbol{\theta}_2)$ between two arbitrary points is given by
\begin{align}
    \Delta{C}_m &= \frac{A\cdot\Delta\widetilde{C}(1)-B\cdot\Delta\widetilde{C}(a)}{D}\label{eq:appdx_ZNE_dCm} \\ 
    &= \frac{A\cdot(1-p)^L\Delta{C}-B\cdot(1-ap)^L\Delta{C}}{D}\,.
\end{align}
{Inspecting} \eqref{eq:appdx_ZNE_dCm}, we see that the error mitigation cost can be bounded simply as
\begin{align}
    \gamma &= \frac{A^2+B^2\frac{\Var[\widetilde{C}(\thv,a)]}{\Var[\widetilde{C}(\thv,1)]}}{D^2}\\
    &\geq \frac{A^2+B^2}{D^2}\,,
\end{align}
where the inequality comes from our core assumption \eqref{eq:appdx_ZNE_assumption}.
Inserting $\gamma$, $\Delta{C}_m$ and $\Delta\widetilde{C}(1)$ into Definition \ref{def:resolvability}, we have
\begin{align}
    \rchi_{depol} = \frac{1}{\gamma}\left(\frac{\Delta{C}_m}{\Delta\widetilde{C}(1)}\right)^2 &\leq \frac{\left(A(1-p)^L -B(1-ap)^L\right)^2}{(A^2+B^2)(1-p)^{2L}}\\
    &= \frac{\left(c-\frac{(1-ap)^L}{(1-p)^L}\right)^2}{c^2+1}\,,
\end{align}
where we have denoted $c=A/B$. By inspecting the specific values of $A$ and $B$ for the Richardson, exponential and NIBP extrapolation strategies respectively, we obtain the results for each strategy.
\end{proof}

\subsubsection{Average relative resolvability}

\begin{proposition}[Average relative resolvability of Zero-Noise Extrapolation, 2 noise levels]\label{prop:appdx_ZNE2}
Consider a Richardson extrapolation strategy based on Eq.~\eqref{eq:appdx_ZNE_taylor}, an exponential extrapolation strategy based on Eq.~\eqref{eq:appdx_ZNE_exp} and a NIBP extrapolation strategy based on Eq.~\eqref{eq:appdx_ZNE_NIBP}. We presume perfect access to an augmented noisy circuit where the noise rate is increased by factor $a_1>1$. We denote $\boldsymbol{\theta}_{\varepsilon*}$ as the parameter corresponding to the global cost minimum at base noise parameter $\varepsilon$. 
Further denote $ \frac{\langle\Delta\widetilde{C}(\thv_{i,\varepsilon*}, a_1\varepsilon)\rangle_{i}}{\langle\Delta\widetilde{C}(\thv_{i,\varepsilon*},\varepsilon)\rangle_{i}} = z$. 
Any such noise model has an average relative resolvability
\begin{equation}
    \overline{\rchi} \leq  \frac{(z-c)^2}{c^2+1}\,,
\end{equation}
where 
\begin{equation}
    c = \begin{cases}
    \quad a_1\quad &\text{for Richardson extrapolation,} \\
    \frac{a_1r(\varepsilon)^{t(\varepsilon)}}{r(a_1\varepsilon)^{t(a_1\varepsilon)}} \quad &\text{for exponential extrapolation,} \\
    \; a_1^{-(L+1)} \quad &\text{for NIBP extrapolation}\,.
    \end{cases}
\end{equation}
Thus, under the assumption that $z\leq1$ and $\langle\Delta \widetilde{C}(\thv_{i,\varepsilon*},a_1\varepsilon)\rangle_i\geq0$, $\overline{\rchi}\leq 1$ for all of the above extrapolation strategies with access to 2 noise levels.
\end{proposition}

\begin{proof}
As in the previous proof, we can inspect Eqs.~\eqref{eq:appdx_ZNE_TaylorCm}, \eqref{eq:appdx_ZNE_expCm} and \eqref{eq:appdx_ZNE_Cm_NIBP}, and see that the Richardson, exponential and NIBP extrapolation strategies all take the form
\begin{equation}
    C_m(\thv) = \frac{A\cdot\widetilde{C}(\thv,1)-B\cdot\widetilde{C}(\thv,a)}{D} + E
\end{equation}
where $A,B, D\geq 0$ (note that $E$ contains the state-independent cost value that represents the fixed point of the noise) and we have adopted the notation of \eqref{eq:appdx_ZNE_notation}. The average mitigated cost differences (averaged over accessible parameters $\{\thv_i\}_i$) can be written
\begin{equation}
    \langle\Delta C_m(\thv_{i,{\varepsilon*}})\rangle_i = \frac{A\cdot\langle\Delta \widetilde{C}(\thv_{i,{\varepsilon*}},1)\rangle_i- B\cdot\langle\Delta \widetilde{C}(\thv_{i,{\varepsilon*}},a)\rangle_i}{D} \,.
\end{equation}
Thus, we have
\begin{align}
    \frac{\langle\Delta C_m(\thv_{i,{\varepsilon*}})\rangle_i}{\langle\Delta \widetilde{C}(\thv_{i,{\varepsilon*}})\rangle_i}  &=   \frac{A  - B\frac{\langle\Delta \widetilde{C}(\thv_{i,{\varepsilon*}},a)\rangle_i}{\langle\Delta \widetilde{C}(\thv_{i,{\varepsilon*}},1)\rangle_i}  }{D} \\
    &=\frac{A  - Bz }{D} \,.
\end{align}
{Finally}, by noting once again that the error mitigation cost is simply bounded as $\gamma \geq \frac{A^2+B^2}{D^2}$ due to \eqref{eq:appdx_ZNE_assumption}, we have 
\begin{equation}
    \rchi \leq \frac{(A-Bz)^2}{A^2 + B^2} \,,
\end{equation}
where we can obtain the desired form by defining $c=A/B$. Finally, the specific values of $c$ for each extrapolation strategy can be read off by inspecting Eqs.~\eqref{eq:appdx_ZNE_TaylorCm}, \eqref{eq:appdx_ZNE_expCm} and \eqref{eq:appdx_ZNE_Cm_NIBP}.
\end{proof}

We now introduce a modification of the basis-averaged relative resolvability in Definition \ref{def:av_resolvabilityII} that we will use to prove an additional result for Zero-Noise Extrapolation. {Here instead of averaging over the basis of an output state, we average over the basis of the measurement observable. This can be thought of as a more natural quantity to consider for comparing Zero-Noise Extrapolation to non-mitigated optimization as the protocol calls for processing of multiple noisy states.}

\begin{definition}[{Basis-averaged relative resolvability II}]\label{def:chi_av_3}
Consider a spectrum $\boldsymbol{\lambda}\in \mathbb{R}_{\geq0}^{2^n}$ with unit $\ell_1$-norm, which corresponds to a noisy reference state. Then define the unitarily-averaged relative resolvability as
\begin{equation}
    \widehat{\overline{\rchi}}_{\boldsymbol{\lambda}} = \frac{1}{\gamma} \frac{\big\langle\big(\widehat{C}_m(\rho,U_i,O_{\boldsymbol{\lambda}}) - \Tr[O_{\boldsymbol{\lambda}}]/2^n \big)^2\big\rangle_{U_i}}{\big\langle\big(\widetilde{C}(\rho,U_i,O_{\boldsymbol{\lambda}}) - \Tr[O_{\boldsymbol{\lambda}}]/2^n \big)^2\big\rangle_{U_i} }\,, 
\end{equation}
where $\langle\cdot\rangle_{U_i}$ denotes an average over $U_i$ drawn from a unitary 2-design, and where we denote 
\begin{align}
    \widetilde{C}(\rho_{\boldsymbol{\lambda}},U_i,O_{\boldsymbol{\lambda}}) &= \Tr[U_i\rho_{\boldsymbol{\lambda}}U^\dag_i O_{\boldsymbol{\lambda}}] \\
    \widehat{C}_m(\rho_{\boldsymbol{\lambda}},U_i,O_{\boldsymbol{\lambda}}) &= \Tr[U_i\mathcal{M}(\rho_{\boldsymbol{\lambda}})U^\dag_i O_{\boldsymbol{\lambda}}]
\end{align}
where $\mathcal{M}: S(\mathcal{H})\mapsto B(\mathcal{H})$ is the map that describes the action of the error mitigation protocol.
\end{definition}

For this averaged relative resolvability we present a result for Zero-Noise Extrpolation.

\begin{supproposition}[Basis-averaged relative resolvability II with Zero-Noise Extrapolation]
Consider a Richardson extrapolation strategy based on Eq.~\eqref{eq:appdx_ZNE_taylor}, an exponential extrapolation strategy based on Eq.~\eqref{eq:appdx_ZNE_exp} and a NIBP extrapolation strategy based on Eq.~\eqref{eq:appdx_ZNE_NIBP}. We presume perfect access to an augmented noisy circuit where the noise rate is increased by factor $a>1$. Denote the output state at the base and augmented noise levels as $\rho(1)$ and $\rho(a)$ respectively. Then we have
\begin{align}
    \widehat{\overline{\rchi}}_{\boldsymbol{\lambda}} \leq \frac{c^2+\frac{ P(a)-1/2^n}{P(1)-1/2^n}}{c^2+1}\,,
\end{align}
where $\widehat{\overline{\rchi}}_{\boldsymbol{\lambda}}$ is the averaged relative resolvability defined in Definition \ref{def:chi_av_3}, $P(1)$ is the purity of $\rho(1)$, $P(a)$ is the purity of $\rho(a)$, and 
\begin{equation}
    c = \begin{cases}
    \quad a\quad &\text{for Richardson extrapolation} \\
    \frac{ar(\varepsilon)^{t(\varepsilon)}}{r(a\varepsilon)^{t(a\varepsilon)}} \quad &\text{for exponential extrapolation} \\
    \; a^{-(L+1)} \quad &\text{for NIBP extrapolation}\,.
    \end{cases}
\end{equation}
Thus, $\widehat{\overline{\rchi}}_{\boldsymbol{\lambda}}\leq1$ when $P(a)\leq P(1)$.
\end{supproposition}

\begin{proof}
We denote reference states $\widetilde{\rho}(\varepsilon)$ and $\widetilde{\rho}(a\varepsilon)$ as states with purity $P(\varepsilon)$ and $P(a\varepsilon)$ respectively. Moreover, denote the noisy cost function values $\widetilde{C}(U_i,\varepsilon) = \Tr[U_i\widetilde{\rho}(\varepsilon)U_i^\dag O]$ and $\widetilde{C}(U_i,a\varepsilon) = \Tr[U_i\widetilde{\rho}(a\varepsilon)U_i^\dag O]$ and further denote $C_m(U_i)$ as the corresponding error mitigated estimator. We start again by noting that in all three Zero-Noise Extrapolation strategies the estimator takes the form
\begin{equation}\label{eq:ZNE_generalCmIII}
    C_m(U_i) = \frac{A\cdot\widetilde{C}(U_i,1)-B\cdot\widetilde{C}(U_i,a)}{D} + E
\end{equation}
where $A,B\geq 0$ (see Eqs.~\eqref{eq:appdx_ZNE_TaylorCm}, \eqref{eq:appdx_ZNE_expCm} and \eqref{eq:appdx_ZNE_Cm_NIBP}) and we have adopted the notation of \eqref{eq:appdx_ZNE_notation}.
We first evaluate the relevant expectation values which correspond to integrals over the Haar distribution over the unitary group of degree $2^n$. We now proceed to derive the result for Richardson/exponential extrapolation, however, we note that the proof follows in a similar way for NIBP extrapolation with the simple substitution $a\varepsilon \mapsto q/a$. Utilizing Lemma \ref{lem:expvalues}, we have 
\begin{align}
    &\langle \widetilde{C}(U_i,\varepsilon) \rangle_{U_i} = \frac{1}{2^n}\Tr[\widetilde{\rho}(\varepsilon)]\Tr[O] = \frac{1}{2^n}\Tr[O]\,, \\[.2em]
    &\big\langle (\Delta\widetilde{C}(U_i,\varepsilon))^2\big\rangle_{U_i} = \big\langle\big(\widetilde{C}(U_i,\varepsilon) - \langle \widetilde{C}(U_j,\varepsilon) \rangle_{U_j} \big)^2\big\rangle_{U_i} = \frac{\left(\Tr[O^2]- \frac{1}{2^n}\Tr[O]^2\right)  \left(\Tr[\widetilde{\rho}^2(\varepsilon)]-\frac{1}{2^n}\Tr[\widetilde{\rho}(\varepsilon)]^2 \right)}{2^{2n}-1} \\
    &\hspace{22.13em} = \frac{\Tr[O^2]- \frac{1}{2^n}\Tr[O]^2 }{2^{2n}-1} \left(P(\varepsilon)-\frac{1}{2^n} \right)\,, \label{eq:ZNE_noisyvar2} \\[1em]
    &\big\langle (\Delta\widetilde{C}(U_i,\varepsilon))(\Delta\widetilde{C}(U_i,a\varepsilon))\big\rangle_{U_i} = \frac{\left(\Tr[O^2]- \frac{1}{2^n}\Tr[O]^2\right)  \left(\Tr[\widetilde{\rho}(\varepsilon)\widetilde{\rho}(a\varepsilon)]-\frac{1}{2^n}\Tr[\widetilde{\rho}(\varepsilon)]\Tr[\widetilde{\rho}(a\varepsilon)] \right)}{2^{2n}-1}\\
    &\hspace{12.6em}\geq 0 \,, \label{eq:ZNE_noisycov2}
\end{align}
where the inequality comes by observing that $\Tr[\widetilde{\rho}(\varepsilon)]=\Tr[\widetilde{\rho}(a\varepsilon)]=1$ and further applying Cauchy-Schwarz to $\Tr[\widetilde{\rho}(\varepsilon)\widetilde{\rho}(a\varepsilon)]$ and noting that the purity of an {$n$-qubit} state is lower bounded by $1/2^n$. Inspecting Eq.~\eqref{eq:ZNE_generalCmIII}  we have
\begin{align}
    &\langle C_{m}(U_i) \rangle_{U_i} = \frac{1}{2^n}\frac{A-B}{D}\Tr[O] + E \,,\\[0.2em]
    &\big\langle\big(C_{m}(U_i) - \langle C_{m}(U_j) \rangle_{U_j} \big)^2\big\rangle_{U_i} = \left\langle \left( \frac{A\cdot\widetilde{C}(U_i,\varepsilon)-B\cdot\widetilde{C}(U_i,a\varepsilon)}{D}+E - \left( \frac{1}{2^n}\frac{A-B}{D}\Tr[O]+E\right) \right)^2 \right\rangle_{U_i} \\
    &\hspace{8.5em} =\left\langle \left( \frac{ A \cdot\Delta\widetilde{C}(U_i,\varepsilon) -B\cdot\Delta\widetilde{C}(U_i,a\varepsilon)}{D} \right)^2 \right\rangle_{U_i} \\
    &\hspace{8.5em} = \frac{A^2\big\langle \big(\Delta\widetilde{C}(U_i,\varepsilon)\big)^2\big\rangle_{U_i} + B^2\big\langle \big(\Delta\widetilde{C}(U_i,a\varepsilon)\big)^2\big\rangle_{U_i} - 2AB \big\langle \Delta\widetilde{C}(U_i,\varepsilon)\Delta\widetilde{C}(U_i,a\varepsilon)\big\rangle_{U_i}}{D^2} \\
    &\hspace{8.5em} \leq \frac{\Tr[O^2]- \frac{1}{2^n}\Tr[O]^2 }{D^2(2^{2n}-1)}\left( A^2\Big(P(\varepsilon)-\frac{1}{2^n}\Big) + B^2\Big(P(a\varepsilon)-\frac{1}{2^n} \Big) \right)\,. \label{ZNE_estvar2}
\end{align}
The inequality comes by substituting in the expressions for $\big\langle \big(\Delta\widetilde{C}(U_i,\varepsilon)\big)^2\big\rangle_{U_i}$ and $\big\langle \big(\Delta\widetilde{C}(U_i,a\varepsilon)\big)^2\big\rangle_{U_i}$ obtained in \eqref{eq:ZNE_noisyvar2}, and dropping the third term in the numerator, where we have used Eq.~\eqref{eq:ZNE_noisycov2}. 
Finally, we note that Eq.~\eqref{eq:ZNE_generalCm} gives $\gamma^{-1} = \frac{D^2}{A^2 + B^2}$. Substituting the obtained expressions for $\gamma^{-1}$, \eqref{ZNE_estvar2} and \eqref{eq:ZNE_noisyvar2} into Definition \ref{def:chi_av_3} we obtain 
\begin{equation}
    \dbloverline{\chi}_{\boldsymbol{\lambda}} \leq \frac{A^2 + B^2\frac{P(a\varepsilon)-1/2^n}{P(\varepsilon)-1/2^n}}{A^2 + B^2}\,,
\end{equation}
where we can define $c=A/B$ to obtain the desired result. Further, the explicit form of $c$ for Richardson, exponential and NIBP extrapolation can be respectively found by inspecting Eqs.~\eqref{eq:appdx_ZNE_TaylorCm}, \eqref{eq:appdx_ZNE_expCm} and \eqref{eq:appdx_ZNE_Cm_NIBP}.
\end{proof}

\subsubsection{Richardson extrapolation with 3 noise levels}\label{sec:appdx_ZNE3}

In this section we focus on Richardson extrapolation (see Appendix \ref{sec:appdx_review_ZNE} for review) and investigate the change in resolvability under an extrapolation strategy that utilizes 3 distinct noise levels. 

\begin{supproposition}[Relative resolvability of Richardson extrapolation with global depolarizing noise, 3 noise levels]
Consider $L$ instances of global depolarizing noise of the form \eqref{eq:gdepol_noise} acting through a circuit.
Consider a Richardson extrapolation strategy based on Eq.~\eqref{eq:appdx_ZNE_taylor}, an exponential extrapolation strategy based on Eq.~\eqref{eq:appdx_ZNE_exp} and a NIBP extrapolation strategy based on Eq.~\eqref{eq:appdx_ZNE_NIBP} in the appendix. We presume access to two augmented noisy circuits where the error probability is perfectly increased by factors $a_2>a_1>1$ as $p\rightarrow a_1p$ and $p\rightarrow a_2p$ respectively. Then for all three extrapolation strategies and any such choices of $a_2$ and $a_1$, we have
\begin{align}
    &{\rchi_{depol}}\leq 1\,,
\end{align}
where $\rchi_{depol}$ is the relative resolvability (see Definition \ref{def:resolvability}) for global depolarizing noise.
\end{supproposition}

\begin{proof}
We start by noting that under $L$ instances of global depolarizing noise with error probability $p$ (of the form \eqref{eq:gdepol_noise}), noisy cost differences are given by 
\begin{equation}\label{eq:appdx_gdepol2}
    \Delta \widetilde{C}(a) = {(1-ap)^L}\Delta C\,,
\end{equation}
for any noise augmentation factor $a$ and any pair of cost function points, where $\Delta C$ is the corresponding noise-free cost difference. 

The error-mitigated cost function difference $\Delta C_m (\thv_{1,2}) = {C}_m(\boldsymbol{\theta}_1) - {C}_m(\boldsymbol{\theta}_2)$ between two arbitrary points constructed under Richardson extrapolation with 3 noise levels is given by
\begin{align}
    \Delta{C}_m &= \frac{a_1a_2(a_2-a_1) \Delta\widetilde{C}(p) - a_2(a_2-1)\Delta\widetilde{C}( a_1p) + a_1(a_1-1)\Delta\widetilde{C}( a_2p) }{(a_1-1)(a_2-1)(a_2-a_1)} \\ 
    &= \frac{a_1a_2(a_2-a_1) {(1-p)^L}\Delta C - a_2(a_2-1){(1-a_1p)^L}\Delta C + a_1(a_1-1){(1-a_2p)^L}\Delta C }{(a_1-1)(a_2-1)(a_2-a_1)}\,,
\end{align}
where in order to obtain the first equality we have used \eqref{eq:appdx_ZNE_TaylorCm3}. The second equality comes by substituting in \eqref{eq:appdx_gdepol2}. Inspecting \eqref{eq:appdx_ZNE_dCm}, we see that the error mitigation cost can be bounded simply as
\begin{align}
    \gamma &= \frac{a_1^2a_2^2(a_2-a_1)^2+a^2_2(a_2-1)^2 \frac{\Var[\widetilde{C}(\thv,a_1p)]}{\Var[\widetilde{C}(\thv,p)]} + a^2_1(a_1-1)^2\frac{\Var[\widetilde{C}(\thv,a_2p)]}{\Var[\widetilde{C}(\thv,p)]}}{(a_1-1)^2(a_2-1)^2(a_2-a_1)^2}\\
    &\geq \frac{a_1^2a_2^2(a_2-a_1)^2+a^2_2(a_2-1)^2 + a^2_1(a_1-1)^2 }{(a_1-1)^2(a_2-1)^2(a_2-a_1)^2}
\end{align}
for any $\thv$, where the inequality comes from our core assumption \eqref{eq:appdx_ZNE_assumption}.
Inserting our expressions for $\gamma$, $\Delta{C}_m$ and $\Delta\widetilde{C}(1)$ into Definition \ref{def:resolvability}, we have
\begin{align}
    \rchi_{depol} = \frac{1}{\gamma}\left(\frac{\Delta{C}_m}{\Delta\widetilde{C}(1)}\right)^2 &\leq \frac{\left(a_1a_2(a_2-a_1) {(1-p)^L} - a_2(a_2-1){(1-a_2p)^L} + a_1(a_1-1){(1-a_2p)^L}\right)^2}{\left(a_1^2a_2^2(a_2-a_1)^2+a^2_2(a_2-1)^2 + a^2_1(a_1-1)^2\right)(1-p)^{2L}}\\
    &= \frac{\left(a_1a_2(a_2-a_1)  - a_2(a_2-1)\frac{(1-a_1p)^L}{(1-p)^L} + a_1(a_1-1)\frac{(1-a_2p)^L}{(1-p)^L}\right)^2}{a_1^2a_2^2(a_2-a_1)^2+a^2_2(a_2-1)^2 + a^2_1(a_1-1)^2}\,.
\end{align}
The desired result can be observed by noting that $a_2(a_2-1)>a_1(a_1-1)$ and that $\frac{(1-a_1p)^L}{(1-p)^L}>\frac{(1-a_2p)^L}{(1-p)^L}$.
\end{proof}

\begin{supproposition}[Average resolvability of Richardson extrapolation, 3 noise levels]
Consider a Richardson extrapolation strategy based on Eq.~\eqref{eq:appdx_ZNE_taylor}, an exponential extrapolation strategy based on Eq.~\eqref{eq:appdx_ZNE_exp} and a NIBP extrapolation strategy based on Eq.~\eqref{eq:appdx_ZNE_NIBP} in the appendix. We presume perfect access to two augmented noisy circuits where the noise rate is increased by factors $a_2>a_1>1$. We denote $\boldsymbol{\theta}_{\varepsilon*}$ as the parameter corresponding to the global cost minimum at base noise parameter $\varepsilon$. Further denote $ \frac{\langle\Delta\widetilde{C}(\thv_{i,{\varepsilon*}},a_1\varepsilon)\rangle_{i}}{\langle\Delta\widetilde{C}(\thv_{i,{\varepsilon*}},\varepsilon)\rangle_{i}} = z_1$ and $\frac{\langle\Delta\widetilde{C}(\thv_{i,{\varepsilon*}},a_2\varepsilon)\rangle_{i}}{\langle\Delta\widetilde{C}(\thv_{i,{\varepsilon*}},\varepsilon)\rangle_{i}} = z_2$. Any such noise model has an average relative resolvability
\begin{equation}
    \overline{\rchi} \leq  \frac{\left(a_1a_2(a_2-a_1)  - a_2(a_2-1)z_1 + a_1(a_1-1)z_2\right)^2}{a_1^2a_2^2(a_2-a_1)^2+a^2_2(a_2-1)^2 + a^2_1(a_1-1)^2}\,,
\end{equation}
where $\overline{\rchi}$ is the averaged relative resolvability (see Definition \ref{def:resolvability}).
Thus, under the assumption that $z_2\leq z_1\leq1$ (on average the cost concentrates with increasing noise level) and $\langle\Delta \widetilde{C}_{i,{\varepsilon*}}(a_1\varepsilon)\rangle_i, \langle\Delta \widetilde{C}_{i,{\varepsilon*}}(a_2\varepsilon)\rangle_i \geq0$ (boosting the noise level does not shift the cost value of the global minimum above the average cost value), then $\overline{\rchi}\leq 1$.
\end{supproposition}

\begin{proof}
The averaged error-mitigated cost function difference $\langle\Delta C_m(\thv_{i,{\varepsilon*}})\rangle_i = \langle{C}_m(\boldsymbol{\theta}_i) - {C}_m(\boldsymbol{\theta}_{\varepsilon*})\rangle_i$ between two arbitrary points constructed under Richardson extrapolation with 3 noise levels is given by
\begin{align}
    \langle\Delta C_m(\thv_{i,{\varepsilon*}})\rangle_i &= \left\langle\left(\frac{a_1a_2(a_2-a_1) \Delta\widetilde{C}(\thv_{i,{\varepsilon*}},p) - a_2(a_2-1)\Delta\widetilde{C}(\thv_{i,{\varepsilon*}}, a_1p) + a_1(a_1-1)\Delta\widetilde{C}(\thv_{i,{\varepsilon*}}, a_2p) }{(a_1-1)(a_2-1)(a_2-a_1)}\right)_{i,{\varepsilon*}}\right\rangle_i \\ 
    &= \frac{a_1a_2(a_2-a_1) \langle\Delta\widetilde{C}(\thv_{i,{\varepsilon*}},p)\rangle_i - a_2(a_2-1)\langle\Delta\widetilde{C}((\thv_{i,{\varepsilon*}},a_1p)\rangle_i + a_1(a_1-1)\langle\Delta\widetilde{C}(\thv_{i,{\varepsilon*}},a_2p)\rangle_i }{(a_1-1)(a_2-1)(a_2-a_1)}\,.
\end{align}
As in the previous proof, we can inspect \eqref{eq:appdx_ZNE_dCm} and we see that the error mitigation cost can be bounded simply as
\begin{align}
    \gamma &= \frac{a_1^2a_2^2(a_2-a_1)^2+a^2_2(a_2-1)^2 \frac{\Var[\widetilde{C}(\thv,a_1p)]}{\Var[\widetilde{C}(\thv,p)]} + a^2_1(a_1-1)^2\frac{\Var[\widetilde{C}(\thv,a_2p)]}{\Var[\widetilde{C}(\thv,p)]}}{(a_1-1)^2(a_2-1)^2(a_2-a_1)^2}\\
    &\geq \frac{a_1^2a_2^2(a_2-a_1)^2+a^2_2(a_2-1)^2 + a^2_1(a_1-1)^2 }{(a_1-1)^2(a_2-1)^2(a_2-a_1)^2}
\end{align}
for any $\thv$, where the inequality comes from our core assumption \eqref{eq:appdx_ZNE_assumption}.
Inserting our expressions for $\gamma$ and $\Delta{C}_m$ into Definition \ref{def:resolvability}, we have
\begin{align}
    \overline{\rchi} = \frac{1}{\gamma}\left(\frac{\Delta{C}_m}{\Delta\widetilde{C}(1)}\right)^2 &\leq \frac{\left(a_1a_2(a_2-a_1)  - a_2(a_2-1)\frac{\langle\Delta\widetilde{C}(\thv_{i,{\varepsilon*}},a_1\varepsilon)\rangle_{i}}{\langle\Delta\widetilde{C}(\thv_{i,{\varepsilon*}},\varepsilon)\rangle_{i}} + a_1(a_1-1)\frac{\langle\Delta\widetilde{C}(\thv_{i,{\varepsilon*}},a_2\varepsilon)\rangle_{i}}{\langle\Delta\widetilde{C}(\thv_{i,{\varepsilon*}},\varepsilon)\rangle_{i}}\right)^2}{a_1^2a_2^2(a_2-a_1)^2+a^2_2(a_2-1)^2 + a^2_1(a_1-1)^2}\,,
\end{align}
and the desired result comes by denoting $ \frac{\langle\Delta\widetilde{C}(\thv_{i,{\varepsilon*}},a_1\varepsilon)\rangle_{i}}{\langle\Delta\widetilde{C}(\thv_{i,{\varepsilon*}},\varepsilon)\rangle_{i}} = z_1$ and $\frac{\langle\Delta\widetilde{C}(\thv_{i,{\varepsilon*}},a_2\varepsilon)\rangle_{i}}{\langle\Delta\widetilde{C}(\thv_{i,{\varepsilon*}},\varepsilon)\rangle_{i}} = z_2$.
\end{proof}

As with the results of Proposition \ref{prop:appdx_ZNE2} we see that $\overline{\rchi}$ decreases with increasing cost concentration.

\subsection{Virtual Distillation}\label{sec:appdx_VD}

\subsubsection{Bounds on error mitigation cost}

We recall the two error mitigation protocols of Ref.~\cite{koczor2020exponential}, denoted "A" and "B" respectively, to prepare
\begin{equation}\label{eq:VD_A_copy}
    C_m^{(A)}(\boldsymbol{\theta}_i)=\Tr[\widetilde{\rho}_i^M O]/\Tr[\widetilde{\rho}_i^M]\,,
\end{equation} and
\begin{equation}\label{eq:VD_B_copy}
    C_m^{(B)}(\boldsymbol{\theta}_i) = \Tr[\widetilde{\rho}_i^M O]/\lambda_i^M\,,
\end{equation}
where $\lambda_i$ is the dominant eigenvalue of $\widetilde{\rho}_i \equiv \widetilde{\rho}(\boldsymbol{\theta}_i)$. The protocols considered explicitly construct these quantities as
\begin{align}
    \Tr[\widetilde{\rho}_iO] &= 2\mathrm{prob}_{1,i} -1 \,,\\
    \Tr[\widetilde{\rho}_i^MO] &= 2\mathrm{prob}_{M,i} -1 \,,\\
    \Tr[\widetilde{\rho}_i^M] &= 2\mathrm{prob}'_{M,i} -1 \,,
\end{align}
where $\mathrm{prob}_{1,i},\ \mathrm{prob}_{M,i}$ and $\mathrm{prob}'_{M,i}$ are expectation values of a Pauli-$Z$ measurement on a qubit ancillary subsystem. In order to obtain our results we will hereon make the core assumption
\begin{align}
    \Var[\mathrm{prob}_{M,i}]&\geq \Var[\mathrm{prob}_{1,i}] \quad \forall i,\, M\geq 2\,,
\end{align}
that is, 
the statistical uncertainty of the measurement outcomes of the circuit that prepares $\widetilde{\rho}^M_i$ are at best equal to that of $\widetilde{\rho}_i$. In the case of large $M$ we expect $\Var[\mathrm{prob}_{M,i}]$ to be large, as for any (non-pure) $\widetilde{\rho}$, the quantity $\Tr[\widetilde{\rho}^MO]$ is close to zero for large $M$. This corresponds to $\mathrm{prob}_{M,i}=\frac{1}{2}$, which maximizes the variance for a binomial distribution.

\begin{lemma}[Bounds on error mitigation cost of virtual distillation]\label{lem:VD_EM_cost}
Denote the error mitigation cost (see Definition \ref{def:em_cost}) corresponding to \eqref{eq:VD_A_copy} and \eqref{eq:VD_B_copy} as $\gamma^{(A)}$ and $\gamma^{(B)}$ respectively. We have
\begin{align} \label{VD_gammas}
    \gamma^{(A)}\geq\frac{1}{(\Tr[\widetilde{\rho}^M])^2}\,,\quad \gamma^{(B)}{\geq}\frac{1}{\lambda^{2M}}\,.
\end{align}
\end{lemma}

\begin{proof}
For $\gamma^{(A)}$ and $\gamma^{(B)}$ we need to compute the variances of the estimators of $C_m^{(A)}$, $C_m^{(B)}$ respectively and likewise $\widetilde{C}=\Tr[\widetilde{\rho} O]$. We have
\begin{align}
    \Var[\widetilde{C}] =  \Var[\Tr[\widetilde{\rho} O]] &= \Var[2\textrm{prob}_1 -1]\,,\\
    &=4\Var[\textrm{prob}_1]\,,
\end{align}
\begin{align}
    \Var[C_m^{(B)}] = \Var\Big[\frac{\Tr[\widetilde{\rho}^M O]}{\lambda^M}\Big] &= \frac{1}{\lambda^{2M}}\Var[2\textrm{prob}_M -1]\,,\\
    &=\frac{4}{\lambda^{2M}}\Var[\textrm{prob}_M]\,,
\end{align}
\begin{align}
    \Var[C_m^{(A)}] = \Var\Big[\frac{\Tr[\widetilde{\rho}^M O]}{\Tr[\widetilde{\rho}^M]}\Big] &= 4\Var[\textrm{prob}_{M}]\, \left(\mathbb{E}\left[\frac{1}{2\textrm{prob}'_{M}-1}\right]\right)^2+ 4\Tr[\widetilde{\rho}^MO]^{2}\,\Var\left[\frac{1}{2\textrm{prob}'_{M}-1}\right] \label{eq:VD_varEA} \\  
        &\quad\ + 4\Var[\textrm{prob}_{M}]\,\Var\left[\frac{1}{2\textrm{prob}'_{M}-1}\right]  \nonumber\\
    &\geq 4\Var[\textrm{prob}_{M}]\, \left(\mathbb{E}\left[\frac{1}{2\textrm{prob}'_{M}-1}\right]\right)^2 \\
    &\geq 4\Var[\textrm{prob}_{M}]\, \frac{1}{\left(\mathbb{E}\left[2\textrm{prob}'_{M}-1\right]\right)^2} \\
    &= 4\Var[\textrm{prob}_{M}]\, \frac{1}{\left(\Tr[\widetilde{\rho}^M]\right)^2}\,.
\end{align}
Equation \eqref{eq:VD_varEA} comes from the standard formula for the variance of the product of two independent random variables. To obtain the first inequality we simply drop the second and third terms, which are positive. The second inequality is an application of Jensen's inequality. Recalling the definition of error mitigation cost (Definition \ref{def:em_cost}), the above three equations enable us to write
\begin{align}
    \gamma^{(A)}\geq\frac{1}{(\Tr[\widetilde{\rho}^M])^2}\,,\quad \gamma^{(B)}{\geq}\frac{1}{\lambda^{2M}}\,,
\end{align}
where we have used our core assumption that $\Var[\textrm{prob}_1]\leq\Var[\textrm{prob}_M]$.
\end{proof}

\subsubsection{Relative resolvability for global depolarizing noise}\label{sec:appdx_VD_depol}

Here we present a proof of Proposition \ref{prop:VD_depol}, in which we upper bound the relative resolvability for Virtual Distillation, for any two cost function points under global depolarizing noise. 

\begin{proposition}[Relative resolvability of Virtual Distillation with global depolarizing noise]
Consider $l$ instances of global depolarizing noise $\mathcal{D}$ of the form 
\begin{equation}
    \rho \xrightarrow{\mathcal{D}} \widetilde{\rho} = q^l \rho + (1-q^l) \frac{\id}{2^n}
\end{equation}
acting on some pure state $\rho$ with some noise parameter $q\in [0,1)$. We consider the two error mitigation protocols of Ref.~\emph{\cite{koczor2020exponential}} (denoted "A" and "B") to respectively prepare \eqref{eq:VD_A_copy} and \eqref{eq:VD_B_copy}. The relative resolvability of any pair of arbitrary cost function points satisfies
\begin{align}
    \rchi^{(A)} \leq \rchi^{(B)} = \Gamma(n,M,q^l)
\end{align}
for all $n\geq1$, $M\geq2$, $q^l \in [0,1]$, and where 
\begin{equation}
    \Gamma(n,M,q^l)\leq 1 \,,
\end{equation}
is a monotonically decreasing function in $M$ (with asymptotically exponential decay) in the quadrant $n\geq 1$, $M\geq 2$. The bound is saturated as $\Gamma(1,2,p)= 1$ for all $p$.
\end{proposition}

\begin{proof}
In this proof we consider arbitrary cost function differences, that is, given two arbitrarily chosen points in parameter space $\thv_1$ and $\thv_2$, we consider
\begin{equation}
    \Delta C = C(\thv_1) - C(\thv_2)\,,
\end{equation}
and the respective differences for the noisy cost $\widetilde{C}(\thv)$ and the mitigated costs $C^{(A)}_m(\thv)$, $C^{(B)}_m(\thv)$. In order to evaluate $\rchi_A$ and $\rchi_B$ we need to first evaluate the following quantities: 
\begin{equation}
    \Delta \widetilde{C}\,,\ \Delta C_m^{(A)} \,,\ \Delta C_m^{(B)} \,,\  \gamma^{(A),(B)}
\end{equation}
that is, the noisy cost function difference between two points, the difference between the virtual distillation estimators for the same points for both protocols, and the error mitigation rate for both protocols. The noisy cost function difference under global depolarizing noise is simply related the noiseless difference as
\begin{equation}
    \Delta \widetilde{C} = q^{l} \Delta C\,. \label{eq:VD_deltaC}
\end{equation}
To evaluate the other quantities we note that 

\begin{align}
    \widetilde{\rho} &= \Big[q^l + \frac{1}{2^n}(1-q^l)\Big] \rho + \Big[\frac{1}{2^n}(1-q^l) \Big] (\id - \rho)\,, \\
    \widetilde{\rho}^M &= \left[ \Big[q^l + \frac{1}{2^n}(1-q^l)\Big]^M - \Big[\frac{1}{2^{n}}(1-q^l)\Big]^M \right]\rho + \frac{2^n}{2^{nM}}(1-q^l)^M \frac{\id}{2^n}\,, \\
    \Tr[\widetilde{\rho}^M] &= \Big[q^l + \frac{1}{2^n}(1-q^l)\Big]^M + \frac{2^n-1}{2^{nM}}(1-q^l)^M \,,\quad \label{eq:stateindependent} \\
    \Tr[\widetilde{\rho}^M O] &= \left[ \Big[q^l + \frac{1}{2^n}(1-q^l)\Big]^M - \Big[\frac{1}{2^{n}}(1-q^l)\Big]^M \right]\Tr[\rho O] + \frac{1}{2^{nM}}(1-q^l)^M \Tr[O]\,.
\end{align}
In particular, we highlight that the expression for $\Tr[\widetilde{\rho}^M]$ {is independent of the noise-free output state $\rho$}. As the dominant noisy eigenvalue $\lambda$ and $\Tr[\widetilde{\rho}^M]$ are state independent we have 
\begin{align}
    \Delta C_m^{(A)} &= \frac{1}{\Tr[\widetilde{\rho}^M]} \left[ \Big[q^l + \frac{1}{2^n}(1-q^l)\Big]^M - \Big[\frac{1}{2^{n}}(1-q^l)\Big]^M \right] \Delta C\,, \label{eq:VD_deltaEA} \\
    \Delta C_m^{(B)} &= \frac{1}{\lambda^M} \left[ \Big[q^l + \frac{1}{2^n}(1-q^l)\Big]^M - \Big[\frac{1}{2^{n}}(1-q^l)\Big]^M \right] \Delta C \,, \label{eq:VD_deltaEB}
\end{align}
where the choice of $\widetilde{\rho}$ is arbitrary. 

Now, using Definition \ref{def:resolvability} and combining \eqref{eq:VD_deltaEA}, \eqref{eq:VD_deltaEB}, \eqref{eq:VD_deltaC} along with the result of Lemma \ref{lem:VD_EM_cost}, we have
\begin{equation}
    \chi^{(A)} \leq \chi^{(B)} = \Gamma(n,M,q^l) \,, \label{eq:VD_SM}
\end{equation}
where we define the function
\begin{equation}\label{eq:appdx_Gamma}
    \Gamma(n,M,q^l) = \frac{1}{q^{2l}} \left[ \Big[q^l + \frac{1}{2^n}(1-q^l)\Big]^M - \Big[\frac{1}{2^{n}}(1-q^l)\Big]^M \right]^2\,.
\end{equation}
First, note that for $M=2$
\begin{align}
    \Gamma(q^l,n,2) &= \frac{1}{q^{2l}} \left[ \Big[q^l + \frac{1}{2^n}(1-q^l)\Big]^2 - \Big[\frac{1}{2^{n}}(1-q^l)\Big]^2 \right]^2 \\
    &= \left( q^l + \frac{2}{2^n} (1-q^l) \right)^2\,.
\end{align}
For $n=1$, $\Gamma(n,2,q^l)=1$. For all $n>1$, $\Gamma(n,2,q^l)<1$ as $(1-q^l)>0$. Thus,
\begin{equation}\label{eq:gamma(M=2)}
    \Gamma(n,2,q^l)\leq1\quad \forall\ n\geq1 \,.
\end{equation}
We complete the proof by showing that $\Gamma(n,M,q^l)$ monotonically decreases with $M$ for all $n\geq1$, $M\geq2$. This can be seen by inspecting the partial derivative (making the decomposition $\Gamma = (\Gamma^{1/2})^2$ due to the square in \eqref{eq:appdx_Gamma})
\begin{equation}
    \frac{\partial \Gamma}{\partial M} = 2 \Gamma^{1/2} \frac{\partial \Gamma^{1/2}}{\partial M}\,.
\end{equation}
We investigate when this quantity is negative. As $\Gamma^{1/2}$ is always positive, negativity is determined by the sign of $\frac{\partial \Gamma^{1/2}}{\partial M}$. Denoting $\delta = \frac{1}{2^{n}}(1-q^l)$, we have 
\begin{align}
    \frac{\partial \Gamma^{1/2}}{\partial M} &= \frac{1}{q^l} \left[ (q^l + \delta)^M \ln (q^l + \delta) - \delta^M \ln \delta \right] \\
    &= \frac{1}{q^l} \left[  \delta^M\left(\ln (q^l+\delta) - \ln \delta \right) + \ln(q^l + \delta)\left( (q^l + \delta)^M - \delta^M \right) \right] \\
    &\leq \frac{1}{q^l}\left[ q^l\delta^{M-1} + (q^l + \delta -1)\left( (q^l+\delta)^M - \delta^M \right) \right] \\
    &=\frac{1}{q^l}\left[ q^l\delta^{M-1} - (2^n-1)\delta\left( (q^l+\delta)^M - \delta^M \right) \right] \\
    &\leq \frac{1}{q^l}\left[ q^l\delta^{M-1} - (2^n-1)\delta\left( Mq^l\delta^{M-1} + \frac{1}{2}M(M-1)q^{2l}\delta^{M-2} \right) \right] \\
    &= \frac{1}{q^l}\left[ q^l\delta^{M-1} - (2^n-1)\delta\left( \frac{1}{2^n}M(1-q^l)q^l\delta^{M-2} + \frac{1}{2}M(M-1)q^{2l}\delta^{M-2} \right) \right] \\
    &=  \delta^{M-1}\left[ 1 - \left(1-\frac{1}{2^n}\right)M  - \frac{1}{2}(2^n-1)M\left(M-1-\frac{2}{2^n}\right)q^{l} \right] \\ 
    &\leq \delta^{M-1}\left[ 1 - \frac{1}{2}M  - \frac{1}{2}M\left(M-2\right)q^{l} \right]\quad \forall\ n\geq1 \,,
\end{align}
where in order to obtain the first inequality we use the inequalities $\ln(q^l+\delta) - \ln \delta \leq q^l/\delta$ and $\ln(q^l+\delta) \leq q^l+\delta-1$. The second inequality comes from observing that the expansion of $\left( (q^l+\delta)^M - \delta^M \right)$ is a sum of positive terms, and considering only two such terms. The above implies that 
\begin{equation}
    \frac{\partial \Gamma}{\partial M} \leq 0 \quad \forall\ n\geq1, M\geq2\,,
\end{equation}
that is, $\Gamma$ is monotonically decreasing with $M$ in the quadrant $n\geq1$ $M\geq2$. Combined with \eqref{eq:gamma(M=2)}, we have the proof as required.
\end{proof}

\subsubsection{Average relative resolvability}\label{sec:appdx_VD_av}

Here we present a proof of Proposition \ref{prop:VD_av}, in which we upper bound the 2-design-averaged relative resolvability for Virtual Distillation.

\begin{proposition}[Average relative resolvability of Virtual Distillation]
Consider an error mitigation protocol that prepares estimator $C_m(\boldsymbol{\theta}_i) = \Tr[\widetilde{\rho}_i^M O]/\Tr[\widetilde{\rho}_i^M]$ from some noisy parameterized quantum state $\widetilde{\rho}_i \equiv \widetilde{\rho}(\boldsymbol{\theta}_i)$. Consider the average relative resolvability for noisy states of some spectrum $\boldsymbol{\lambda}$ with purity $P_{\boldsymbol{\lambda}}$ as defined in Definition \ref{def:av_resolvability}. We have
\begin{equation}\label{eq:appdx_vd_avbound}
    \dbloverline{\rchi}_{\boldsymbol{\lambda}}\leq G(n,M,P)\,\leq\,1\,,
\end{equation}
where $G(n,M,P)$ is a monotonically decreasing function in $M$ (with asymptotically exponential decay) for all $n\geq 1$, $M\geq 2$.
Within this region the bound is saturated as $G(1,2,P)=1$ for all $P$ and $G(n,M,1)=1$ for all $n\geq 1, M\geq 2$. Explicitly, we have for $n=1$
\begin{equation}
    {G(n=1,M,P)}= \frac{1}{2^{2M}} \frac{\left[(1+\sqrt{2P-1})^M-(1-\sqrt{2P-1})^M\right]^2}{2P-1}\,.
\end{equation}
For $n\geq2$ and $M=2$ 
\begin{align}
    G(n\geq2,M=2,P) = \min\left(
    \frac{4}{2^{2n}} + \frac{4}{2^{n/2}}g_2\sqrt{P-\frac{1}{2^n}} + 2^n g^2_2\left(P-\frac{1}{2}\right),\ \frac{P^2}{P-\frac{1}{2^n}}\left(1-\frac{1}{2^n}\right)\right)\,,
\end{align}
where we denote $g_k=\left(\frac{2^n-1}{2^n}\right)^k+\left( \frac{1}{2^n} \right)^k$. Further, for $n\geq2$ and $M\geq3$ we have 
\begin{align}
    G(n\geq2,M\geq3,P) = \min\left(
    \frac{2^{n}}{4} \frac{\left[ \left(\sqrt{2\left(P-\frac{1}{2^n}\right)}+\frac{1}{2^n}\right)^M - \left(\frac{1}{2^n}\right)^M \right]^2}{P-\frac{1}{2^n}},\ \frac{P^M}{P-\frac{1}{2^n}}\left(1-\frac{1}{2^n}\right) \right)\,.
\end{align}
\end{proposition}

\begin{proof}
From Definition \ref{def:av_resolvabilityII} we have 
\begin{equation}\label{eq:avshotmitigationIIcopy}
    \dbloverline{\rchi}_{\boldsymbol{\lambda}} = \frac{1}{\gamma(\boldsymbol{\lambda})} \frac{\big\langle\big(C_m(\rho_{\boldsymbol{\lambda}},U_i) - \Tr[O]/2^n \big)^2\big\rangle_{U_i}}{\big\langle\big(\widetilde{C}(\rho_{\boldsymbol{\lambda}},U_i) - \Tr[O]/2^n \big)^2\big\rangle_{U_i} }\,.
\end{equation}
Let us first evaluate the required averages over unitary 2-designs. The relevant first moments for virtual distillation are given by
\begin{align}
    \langle \Tr[U\rho_{\boldsymbol{\lambda}} U^\dag O] \rangle_U &= \Tr[O]/2^n\,, \\
    \langle \Tr[U\rho_{\boldsymbol{\lambda}}^M U^\dag O] \rangle_U &= \Tr[\rho_{\boldsymbol{\lambda}}^M]\Tr[O]/2^n\,,
\end{align}
where we have used Lemma \ref{lem:expvalues}. Thus we can see that the numerator and denominator of \eqref{eq:avshotmitigationIIcopy} correspond to variances which we now evaluate. Again, utilizing Lemma \ref{lem:expvalues}, the second moments are given by
\begin{align}
    \big\langle\big(\widetilde{C}(U_i) - \langle \widetilde{C}(U_j) \rangle_{U_j} \big)^2\big\rangle_{U_i} &= \langle (\Tr[U\rho_{\boldsymbol{\lambda}} U^\dag O])^2 \rangle_U - (\Tr[O]/2^n)^2 \\
    &= \frac{\left(\Tr[O^2]- \frac{1}{2^n}\Tr[O]^2\right)  \left(\Tr[\rho_{\boldsymbol{\lambda}}^2]-\frac{1}{2^n}\Tr[\rho_{\boldsymbol{\lambda}}]^2 \right)}{2^{2n}-1} \\[1em]
    \big\langle\big(C_m(U_i) - \langle C_m(U_j) \rangle_{U_j} \big)^2\big\rangle_{U_i} &= \bigg\langle \left(\frac{\Tr[U\rho_{\boldsymbol{\lambda}}^M U^\dag O]}{\Tr[\rho_{\boldsymbol{\lambda}}^M]}\right)^2 \bigg\rangle_U - (\Tr[O]/2^n)^2 \\
    &=\frac{\left(\Tr[O^2]- \frac{1}{2^n}\Tr[O]^2\right)  \left(\frac{\Tr[\rho_{\boldsymbol{\lambda}}^{2M}]}{\Tr[\rho_{\boldsymbol{\lambda}}^M]^2}-\frac{1}{2^n} \right)}{2^{2n}-1}\,,
\end{align}
where in the final equality we have used the fact that $\Tr\left[\frac{\rho_{\boldsymbol{\lambda}}^{M}}{\Tr[\rho_{\boldsymbol{\lambda}}^M]}\right]=1 ${.}
Using the definition of the basis-averaged relative resolvability (Definition \ref{def:av_resolvabilityII}), we can arrive at a bound written explicitly in terms of $\rho_{\boldsymbol{\lambda}}$ as
\begin{equation} \label{eq:appdx_VD_avbound1}
    \dbloverline{\chi}_{\boldsymbol{\lambda}} = \frac{1}{\gamma} \frac{\big\langle\big(C_m(U_i) - \langle C_m(U_j) \rangle_{U_j} \big)^2\big\rangle_{U_i}}{\big\langle\big(\widetilde{C}(U_i) - \langle \widetilde{C}(U_j) \rangle_{U_j} \big)^2\big\rangle_{U_i}} \leq  \frac{\Tr[\rho_{\boldsymbol{\lambda}}^{2M}]-\frac{1}{2^n}\Tr[\rho_{\boldsymbol{\lambda}}^M]^2}{\Tr[\rho_{\boldsymbol{\lambda}}^2]-\frac{1}{2^n}\Tr[\rho_{\boldsymbol{\lambda}}]^2}\,,
\end{equation}
where we have used the fact that the error mitigation cost $\gamma \geq 1/(\Tr[\rho_{\boldsymbol{\lambda}}^M])^2$.

The goal is to now investigate whether or not $f(M) = \Tr[\rho_{\boldsymbol{\lambda}}^{2M}]-\frac{1}{2^n}\Tr[\rho_{\boldsymbol{\lambda}}^M]^2$ is monotonically decreasing for $M\in \mathbb{N}_+$. This quantity has two interpretations. First, it can be seen to be a Hilbert Schmidt distance between $\rho_{\boldsymbol{\lambda}}^M$ and $\Tr[\rho_{\boldsymbol{\lambda}}^M]\frac{\id}{2^n}$. Second, by considering the eigenvalue decomposition of $\rho$, it can be seen to be proportional to the population variance of the distribution $\{\lambda^M_i \}$, where $\lambda_i$ are the eigenvalues of $\rho_{\boldsymbol{\lambda}}$, that is,
\begin{equation}
    f(M) = 2^n \Var[\boldsymbol{\lambda}^{(M)}] = \sum_i \lambda_i^{2M} - \frac{1}{2^n}\left(\sum_i \lambda_i^M\right)^2\,,
\end{equation}
where here $\Var[(.)]$ denotes the population variance of the contained vector. Thus, we can rewrite Eq.~\eqref{eq:appdx_VD_avbound1} as
\begin{equation}\label{eq:VD_chi_av2}
    \dbloverline{\chi}_{\boldsymbol{\lambda}}(M)\leq \frac{f(M)}{f(1)} = \frac{\Var[\boldsymbol{\lambda}^{(M)}]}{\Var[\boldsymbol{\lambda}^{(1)}]}\,.
\end{equation}

 Let us first treat the qubit setting of $n=1$. Consider eigenvalue decomposition $\rho_{\boldsymbol{\lambda}} = \lambda\ket{\psi}\bra{\psi} + (1-\lambda)\ket{\psi_\perp}\bra{\psi_\perp}$, where we have defined $\lambda_1 = 1-\lambda$, $\lambda_2 = \lambda$ and without loss of generality we fix $1-\lambda \geq \lambda$. We define $G(1,M,P)=f(M)/f(1)$ and will determine $f(M)$ exactly for single-qubit states. For generic $M$ we have 
\begin{align}
    f(M) &= \lambda^{2M} + (1-\lambda)^{2M} - \frac{1}{2}\left(\lambda^M + (1-\lambda)^M \right)^2\\
    &= \frac{1}{2}((1-\lambda)^M-\lambda^{M})^2 \\
    &= \frac{1}{2^{2M+1}} \left[(1+\sqrt{2P-1})^M-(1-\sqrt{2P-1})^M\right]^2\,,
\end{align}
where in the final equality we have used the fact that for single-qubit states $\lambda = \frac{1}{2}(1-\sqrt{2P-1})$. Further, using $f(1)=P-\frac{1}{2}$ we have the bound as required.

Now let us consider the setting of $n\geq2$. We will construct two bounds, for the respective high purity and low purity limits. We start with the bound for high purity states. 
We can write the right hand side of Eq.~\eqref{eq:VD_chi_av2} explicitly as
\begin{align}
    \frac{\Var[\boldsymbol{\lambda}^{(M)}]}{\Var[\boldsymbol{\lambda}^{(1)}]} &= \frac{\frac{1}{2^n}\sum_i\lambda_i^{2M}-(\frac{1}{2^n}\sum_i\lambda_i^M)^2}{\frac{1}{2^n}\sum_i\lambda_i^2-\frac{1}{2^{2n}}} \\
    &= \frac{\frac{1}{2^n}\sum_i\lambda_i^{2M}-\frac{1}{2^{2n}}\sum_{i}\lambda_i^{2M}-\frac{1}{2^{2n}}\sum_{i\neq j}\lambda_i^M\lambda_j^M}{\frac{1}{2^n}\sum_i\lambda_i^2-\frac{1}{2^{2n}}} \\
    &\leq \frac{(2^n-1)(\sum_i\lambda_i^{2M})}{2^n\sum_i\lambda_i^{2}-1} \\
    &\leq \frac{(2^n-1)(\sum_i\lambda_i^{2})^M}{2^n\sum_i\lambda_i^{2}-1} \\
    &= \frac{P^M}{P-\frac{1}{2^n}}\left(1-\frac{1}{2^n}\right) \,,
\end{align}
where in order to obtain the first inequality we have dropped the cross terms $\frac{1}{2^{2n}}\sum_{i\neq j}\lambda_i^M\lambda_j^M$, and in the second inequality we have introduced new cross terms. The final equality comes by substituting in the definition of the purity $P$. We note this first bound is upper-bounded by $1$ for all $P\geq\frac{1}{2^n-1}$. Thus, we seek a tighter bound for $P\leq\frac{1}{2^n-1}$. 

We can now construct our second bound for strongly mixed states {(those states with purity close to $1/2^n$)}. We will consider bounds on $\Var[X^M]$ where a random variable $X$ when it is known that it takes values close to its mean $\mu$. We consider the decomposition 
\begin{align}
    X^M &= \left((X-\mu)-\mu\right)^M \\
    &= \mu^M + \sum^M_{k=1}Y_k  
\end{align}
where we have defined the random variables $Y_k=\binom{M}{k} \mu^{M-k}(X-\mu)^k$. Further, we can write
\begin{align}
    \Var[X^M] &= \Var\bigg[\sum^M_{k=1}Y_k \bigg] \\
    &= \mathbb{E}\left[ \bigg(\sum_k Y_k - \mathbb{E}\Big[ \sum_k Y_k \Big] \bigg) \bigg(\sum_j Y_j - \mathbb{E}\Big[ \sum_j Y_j \Big] \bigg) \right] \\
    &= \sum_{k,j}\mathbb{E}\left[ \Big( Y_k - \mathbb{E}[  Y_k ] \Big) \Big( Y_j - \mathbb{E}[ Y_j ] \Big) \right] \\
    &=\sum_{k,j} \Cov \big[ Y_k, Y_j \big] \\
    &\leq \sum_{k,j} \sqrt{\Var\big[Y_k\big]\Var\big[Y_j\big]}\,,\label{eq:VD_highlvl-varbound}
\end{align}
where the inequality is due to Cauchy-Schwarz. We now take $X$ to be the random variable which takes values $\{\lambda_i\}_i$ with uniform probability and mean $\mu=\frac{1}{2^n}$. We will bound $\Var\big[Y_k\big]$ under the assumption that $\{\lambda_i\}_i$ are close in value to the maximally mixed value $\frac{1}{2^n}$. 

First, we note that each $Y_k$ is a function of $(X-\mu)^k$, and so we must investigate the shifted spectrum which we denote $\hat{\boldsymbol{\lambda}}$ where $\hat{\lambda}_i = \lambda_i - \frac{1}{2^n}$ for all $i$. Using Popoviciu's inequality, we have the bound 
\begin{equation}\label{eq:VD_Popoviciu}
    \Var\left[(X-\mu)^k\right] \leq \frac{1}{4}\left( \hat{\lambda}^k_{max}-\hat{\lambda}^k_{min} \right)^2\,.
\end{equation}
Now suppose that we have the constraint
\begin{equation}\label{eq:VD_lam-constraint}
    \lambda_{max}-\lambda_{min}=2b
\end{equation}
for some $b\geq0$. For any $k$, we have 
\begin{equation}\label{eq:VD_lambda-hat}
    \hat{\lambda}^k_{max}-\hat{\lambda}^k_{min} \leq \big|\hat{\lambda}_{max}\big|^k+\big|\hat{\lambda}_{min}\big|^k\,.
\end{equation}
Let us now bound the quantity on the right by considering its maximum value over all spectra with constraint \eqref{eq:VD_lam-constraint}. The quantity on the right is maximized by the choice of vector $\big(\big|\hat{\lambda}_{max}\big|,\big|\hat{\lambda}_{min}\big|\big)$ that majorizes all others, given some fixed value of $\big|\hat{\lambda}_{max}\big|+\big|\hat{\lambda}_{min}\big|$. Indeed, $\big|\hat{\lambda}_{max}\big|+\big|\hat{\lambda}_{min}\big|=b$ is fixed by our constraint \eqref{eq:VD_lam-constraint} ($\hat{\lambda}_{min}$ must be negative in order to preserve trace). Thus the quantity on the right hand side of \eqref{eq:VD_lambda-hat} can be bounded by maximizing $\hat{\lambda}_{max}$ and minimizing $\big|\hat{\lambda}_{min}\big|$. This is achieved by setting all other $\hat{\lambda}_i$ equal to $\hat{\lambda}_{min}$. We then have pair of constraints 
\begin{align}
    &\hat{\lambda}_{max} + (2^n-1)\hat{\lambda}_{min} = 0 \,,\\
    &\hat{\lambda}_{max} - \hat{\lambda}_{min} = 2b\,,
\end{align}
where the first constraint comes from preservation of trace, and the second is our original constraint. This is a linear system of equations with solution
\begin{equation}
    \hat{\lambda}^*_{max} = 2b\frac{2^n-1}{2^n}\,,\quad \hat{\lambda}^*_{min} = -2b\frac{1}{2^n}
\end{equation}
substituting these values into \eqref{eq:VD_lambda-hat} we have the bound 
\begin{align}
    \hat{\lambda}^k_{max}-\hat{\lambda}^k_{min} &\leq (2b)^k\left( \left( \frac{2^n-1}{2^n}\right)^k + \left( \frac{1}{2^n}\right)^k \right)\label{eq:VD_itm-bound1}\\
    &\leq (2b)^k\label{eq:VD_itm-bound2}\,.
\end{align}
We will find it necessary to use the tighter bound \eqref{eq:VD_itm-bound1} in the case of $M=2$, but the looser bound \eqref{eq:VD_itm-bound2} will enable us to write a bound with a more compact form for $M\geq3$.

We now relate $b$ to the purity. We can write a general spectrum that satisfies the constraint in \eqref{eq:VD_lam-constraint} as $\boldsymbol{\lambda}_{b,c,\boldsymbol{a}}= (\frac{1}{2^n}+b+c,\,\frac{1}{2^n}-b+c,\,\frac{1}{2^n}-a_1,\,...,\,\frac{1}{2^n}-a_{d-2})$, for some $c$ and set $\{a_i\}_i$ that satisfy $\sum^{d-2}_i a_i = 2c$ (in order to preserve trace). The purity that corresponds to this spectrum is given by
\begin{align}
    P(\boldsymbol{\lambda}_{b,c,\boldsymbol{a}}) &= \left(\frac{1}{2^n}+b+c \right)^2 + \left(\frac{1}{2^n}-b+c \right)^2 + \sum_{i=1}^{d-2}\left(\frac{1}{2^n}-a_i \right)^2 \\
    &= \frac{1}{2^n} + 2b^2 + c^2 + \sum_i a^2_i + \frac{2}{2^n}\bigg[ 2c - \sum^{d-2}_{i=1}a_i \bigg]\\
    &\geq \frac{1}{2^n} + 2b^2\,.
\end{align}
Moreover, this purity bound is achievable by the spectrum $\boldsymbol{\lambda}_{b,0,\boldsymbol{0}}= (\frac{1}{2^n}+b,\frac{1}{2^n}-b,\frac{1}{2^n},...,\frac{1}{2^n})$ if we have $b \leq \frac{1}{2^n}$. We conclude that for any spectrum $\boldsymbol{\lambda}_b$ that satisfies the constraint \eqref{eq:VD_lam-constraint}, we have 
\begin{equation}\label{eq:VD_b-bound}
    b\leq \sqrt{\frac{1}{2}\left(P(\boldsymbol{\lambda}_b)-\frac{1}{2^n}\right)}\,.
\end{equation}
And we now have all the tools to bound $\Var\big[Y_k\big]$ for all $k$ and subsequently $\Var[X^M]$

By combining the bounds \eqref{eq:VD_Popoviciu} and \eqref{eq:VD_itm-bound1} we have
\begin{equation}\label{eq:VD_varXk-bound}
    \Var\left[(X-\mu)^k\right] \leq \frac{1}{4} (2b)^{2k} g_k^2
\end{equation}
where we have denoted $g_k= \left( \frac{2^n-1}{2^n}\right)^k + \left( \frac{1}{2^n}\right)^k \leq 1$. This allows us to bound $\Var[Y_k]$ by writing
\begin{align}\label{eq:VD_var-gk}
    \sqrt{\Var[Y_k]} &= \binom{M}{k} \mu^{M-k} \sqrt{\Var\left[(X-\mu)^k\right]}\\
    &\leq \frac{1}{2}\binom{M}{k} \mu^{M-k}(2b)^{k} g_k\,.
\end{align}
We first pursue a bound for general $M\in\mathbb{N}$ and replace each $g_k$ with $1$. We observe that the quantities $\{\binom{M}{k} \mu^{M-k}(2b)^{k}\}^M_{k=1}$ are simply the terms in the expansion of $(2b-\mu)^M - \mu^M$, that is,
\begin{equation}
    \sum_k \sqrt{\Var[Y_k]} \leq \frac{1}{2}\left((2b-\mu)^M - \mu^M \right)\,.
\end{equation}
Returning to \eqref{eq:VD_highlvl-varbound}, we have
\begin{align}
    \Var[X^M]&\leq\frac{1}{4}\left((2b-\mu)^M - \mu^M \right)^2\\
    &\leq \frac{1}{4}\left[\left(2\sqrt{\frac{1}{2}\left(P-\frac{1}{2^n}\right)}-\mu\right)^M - \mu^M \right]^2
\end{align}
where in order to obtain the second inequality we have used \eqref{eq:VD_b-bound} to substitute $b$ with its bound in terms of the purity. We further note that $\Var[X] = \frac{1}{2^n}(P-\frac{1}{2^n})$, and dividing the two quantities we obtain
\begin{equation}
     \dbloverline{\chi}\leq\frac{2^{n}}{4} \frac{\left[ \left(\sqrt{2\left(P-\frac{1}{2^n}\right)}+\frac{1}{2^n}\right)^M - \left(\frac{1}{2^n}\right)^M \right]^2}{P-\frac{1}{2^n}}
\end{equation}
as required. To summarize, combining the two bounds for high purity and low purity, so far we have 
\begin{align}\label{eq:VD_prelim-bound}
    G'(n\geq2,M\geq2,P) = \min\left( 
    \frac{2^{n}}{4} \frac{\left[ \left(\sqrt{2\left(P-\frac{1}{2^n}\right)}+\frac{1}{2^n}\right)^M - \left(\frac{1}{2^n}\right)^M \right]^2}{P-\frac{1}{2^n}},\ \frac{P^M}{P-\frac{1}{2^n}}\left(1-\frac{1}{2^n}\right) \right)\,.
\end{align}

Now we discuss the magnitude of our bound obtained thus far, as well as its monotonicity with respect to $M$. In particular, we will show that its value can exceed 1 for $M=2$, and so we will pursue a tighter bound for $M=2$. We can evaluate $G'(n\geq2,M\geq2,P)$ explicitly at $P=\frac{1}{2^n-1}$ as
\begin{align}\label{eq:VD_prelim-bound2}
    G'\Big(n\geq2,M\geq2,P=\frac{1}{2^n-1}\Big) = \min\left(
    \frac{2^{2n}(2^n-1)}{4} \left[ \left(\sqrt{\frac{2}{2^n(2^n-1)}}+\frac{1}{2^n}\right)^M - \left(\frac{1}{2^n}\right)^M \right]^2,\ \frac{(2^n-1)^2}{(2^n-1)^M} \right)\,.
\end{align}
Firstly, by inspection this is a decreasing function in $n$ for all $M\geq 2$, so in order to bound its magnitude we can consider $n=2$. At $M=2$ we have 
\begin{align}
    G'\Big(2,2,\frac{1}{2^2-1}\Big) = \min\left(1,\ 
    \frac{5+2\sqrt{6}}{6} \right)= 1\,,
\end{align}
where we note $\frac{5+2\sqrt{6}}{6}\geq 1$. As the first function in the minimization of \eqref{eq:VD_prelim-bound} has negative gradient for $P<\frac{1}{2^n-1}$ for $n\geq2, M=2$, this implies that {there} exists a set of values $P=\frac{1}{2^n-1}-\delta$, where $\delta>0$ is small, such that the first function has value exceeding 1. The second function also has value exceeding $1$ in such a region as it is continuous. Thus, there exist values of $P$ for which the bound $G'>1$ at $M=2$. Moving on to $M=3$, we can numerically verify that $ G'\Big(2,3,\frac{1}{2^n-1}\Big)\leq1$ with both functions in the minimization having value below $1$. As functions of the form $f(x) = a^x - b^x$ where $b\leq a \leq 1$ only have one stationary point which is a maximum, this implies that $G'\Big(2,M,\frac{1}{2^n-1}\Big)$ is decreasing for all $M\geq 3$ and thus $G'\Big(2,M\geq3,\frac{1}{2^n-1}\Big)\leq 1$. 

We will replace $G'\Big(n\geq2,2,P\Big)$ with a tighter bound that is less than 1 for all $n\geq2$. We return to \eqref{eq:VD_var-gk} and now explicitly consider the $g_k$ terms. Substituting this into \eqref{eq:VD_highlvl-varbound} for $M=2$ we have 
\begin{align}
    \Var[X^2] &\leq \Var[Y_1] + \Var[Y_2] + 2\sqrt{\Var[Y_1]\Var[Y_2]} \\
    &= (2\mu)^2\Var[X-\mu] + \Var\left[\left(X-\mu \right)^2\right] + {4\mu}\sqrt{\Var\left[\left(X-\mu \right)^2\right]\Var[X-\mu]}\\
    &\leq (2\mu)^2\Var[X] + \frac{1}{4} (2b)^{4} g_2^2 + {4\mu}\sqrt{\frac{1}{4} (2b)^{4} g_2^2 \Var[X]} \\
    &\leq \frac{4}{2^{2n}}\Var[X] + \left(P-\frac{1}{2^n} \right)^2g_2^2 + \frac{4}{2^n}g_2\left({P-\frac{1}{2^n}}\right)  \sqrt{\Var[X]}
\end{align}
where in the first equality we use the definition of $Y_k$ for $M=2$, in the first inequality we use \eqref{eq:VD_varXk-bound} along with the fact that $g_1 =1$, and in the final inequality we use \eqref{eq:VD_b-bound}. Finally, dividing by $\Var[X]$ we have 
\begin{align}
    \dbloverline{\rchi}(M=2) &\leq \frac{4}{2^{2n}} + \frac{1}{\Var[X]}\left(P-\frac{1}{2^n} \right)^2g_2^2 + \frac{4}{2^n}g_2\sqrt{P-\frac{1}{2^n}}  \sqrt{\frac{1}{\Var[X]}} \\
    &= \frac{4}{2^{2n}} + 2^n g^2_2\left(P-\frac{1}{2^n}\right) + \frac{4}{2^{n/2}}g_2\sqrt{P-\frac{1}{2^n}}\,,
\end{align}
where we have used $\Var[X] = \frac{1}{2^n}(P-\frac{1}{2^n})$.
\end{proof}

We note in the following remark that outside of the low purity regime, as purity decreases, our bounds monotonically decrease. 

\begin{remark}
The bounds in Proposition \ref{prop:VD_av} are monotonically increasing with purity $P$ for all $P\geq \frac{1}{2^n} \max (\frac{2^n}{2^n-1}, \frac{M}{M-1})$. 
\end{remark}
\begin{proof}
For $n\geq 2$, this can be seen by inspecting the partial derivative of the high purity bound, which is given by
\begin{equation}
    \frac{\partial G(n\geq 2,M,P\geq \frac{1}{2^n-1})}{\partial P} =  \frac{(M+1)P^M-\frac{1}{2^n}MP^{M-1}}{(P-\frac{1}{2^n})^2}\left(1-\frac{1}{2^n}\right)
\end{equation}
which is positive for all $P>\frac{M}{M-1}\frac{1}{2^n}$. Similarly, for $n=1$ the bound can be shown to be monotonically increasing in $P$ for all physically allowable values of $P$. We note that the bound for $n=1$ satisfies
\begin{equation}\label{eq:appdx_VD_derivativeP}
    \sqrt{2^{2M}G(P,n=1,M)}=  \frac{(1+x)^M-(1-x)^M}{x}
\end{equation}
where we have denoted $x=\sqrt{2P-1}$. The derivative of the numerator with respect to $x$ takes the value
\begin{align}
    \frac{d\left( (1+x)^M-(1-x)^M\right)}{dx} &= M\left((1+x)^{M-1}+(1-x)^{M-1}\right)\\
    &\geq M(1+(M-1)x + 1-(M-1)x)\\
    &=2M
\end{align}
where in order to obtain the inequality we have used the standard inequality $(1+x)^n \geq 1+nx,\; \forall x\geq -1, n>1$. Thus, as we only consider $M\geq 2$, the numerator of \eqref{eq:appdx_VD_derivativeP} increases at a faster rate than the denominator. Moreover, the second derivative of the numerator is positive, and both the numerator and denominator of \eqref{eq:appdx_VD_derivativeP} take value $0$ at $x=0$ ($P=1/2^n$). Thus, \eqref{eq:appdx_VD_derivativeP} is an increasing function in the purity $P$.
\end{proof}

We plot the bounds obtained in Proposition \ref{prop:VD_av} on the 2-design-averaged resolvability in Fig.~\ref{fig:vdbound}. First, in the left figure we plot the intermediate bound $\frac{\Var[\boldsymbol{\lambda}^{(M)}]}{\Var[\boldsymbol{\lambda}^{(1)}]}$ in \eqref{eq:VD_chi_av2} for states with 100 randomly generated spectra for increasing number of qubits $n$ and number of state copies $M$. {We see that all values lie below 1. Moreover, this plot visualizes the exponential scaling with $M$ for fixed $n$ and we observe that broadly, the bound is decreasing with increasing number of qubits $n$ for fixed $M$}. Further, as expected, the bound is always less than or equal to 1. Second, in order to demonstrate the behaviour of our final upper bound \eqref{eq:appdx_vd_avbound} we plot increasing number of state copies $M$ ranging from $2$ to $4$ for $n=2$. For each $M$, we randomly generate 10000 states and plot $\frac{\Var[\boldsymbol{\lambda}^{(M)}]}{\Var[\boldsymbol{\lambda}^{(1)}]}$ against the purity of the state as separate points. The final upper bound is then plotted as a line for each value of $M$.

\begin{figure}
    \centering
    {{\includegraphics[width=0.45\columnwidth]{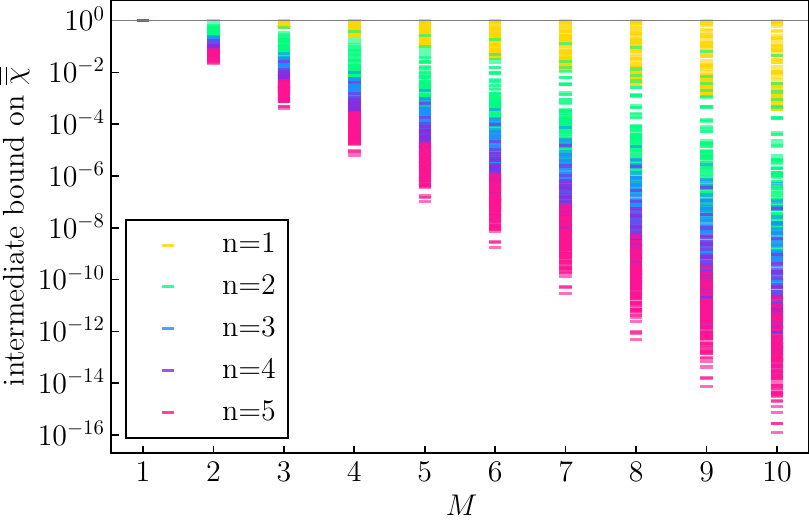} }}
    \qquad
    {{\includegraphics[width=0.45\columnwidth]{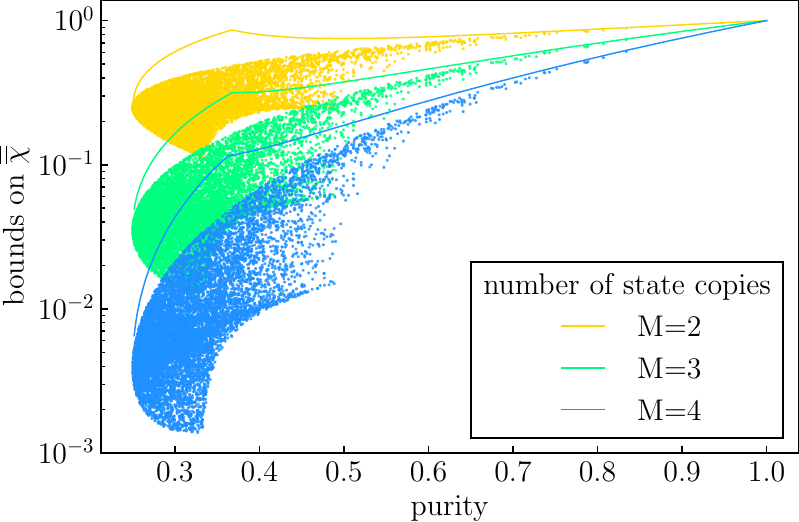} }}
    \caption{\textbf{Bounds on $\dbloverline{\rchi}$ for VD.} (Left): We plot the intermediate upper bound $\frac{\Var[\boldsymbol{\lambda}^{(M)}]}{\Var[\boldsymbol{\lambda}^{(1)}]}$ for randomly generated states with increasing number of qubits $n$ and number of state copies $M$. (Right): We plot as points the intermediate upper bound $\frac{\Var[\boldsymbol{\lambda}^{(M)}]}{\Var[\boldsymbol{\lambda}^{(1)}]}$ against purity for randomly generated states at different values of $M$ at $n=2$. We also plot the final bound  \eqref{eq:appdx_vd_avbound}, which is a function of purity, as a line.}
    \label{fig:vdbound}
\end{figure}

\subsection{Probabilistic Error Cancellation}\label{sec:appdx_QP}

\subsubsection{Error mitigation of multiple gates}\label{sec:appdx_QP_gammatot}
We first consider the error mitigation cost of mitigating multiple noise channels. Suppose we have two noisy gates which we represent as the channel
\begin{equation}
    \mathcal{N}'\circ\mathcal{U}'\circ\mathcal{N}\circ\mathcal{U}
\end{equation}
where $\{\mathcal{U}',\mathcal{U}\}$ are channels that represent the ideal gates and $\{\mathcal{N}',\mathcal{N}\}$ are noise channels. Note that this framework also includes as a special case the scenario where two gates act in parallel on different subsystems. Given {a} set of basis gates $\{\mathcal{B}_\alpha\}_\alpha$, we can construct a quasiprobability distribution for the ideal channel as
\begin{equation}\label{eq:appdx_doubleQPdistribution}
    \mathcal{U}'\circ\mathcal{U} = \sum_{\alpha,\beta} k_\alpha k_\beta\, \mathcal{B}_\alpha\circ\mathcal{N}'\circ\mathcal{U}'\circ\mathcal{B}_\beta\circ\mathcal{N}\circ\mathcal{U}\,.
\end{equation}
where we have used \eqref{eq:appdx_QuasiNoiselessgate}.
From \eqref{eq:appdx_doubleQPdistribution} we see that the error mitigation cost is 
\begin{equation}
    \gamma_{tot} = \sum_{\alpha,\beta} k^2_\alpha k^2_\beta = \gamma \gamma'
\end{equation}
where $\gamma, \gamma'$ are the individual error mitigation costs for each gate. We can see the above reasoning can be extended inductively to show that the error mitgation cost of a collection of gates with the probabilistic error cancellation is equal to the product of the individual error mitigation costs.

\subsubsection{Global depolarizing noise}

\begin{proposition}[Relative resolvability of Probabilistic Error Cancellation for global depolarizing noise]
Consider a quasi-probability method that corrects global depolarizing noise of the form \eqref{eq:gdepol_noise}. For any pair of states corresponding to points on the cost function landscape, the optimal quasiprobability scheme gives
\begin{equation}
    \rchi_{depol} = \frac{2^{2n}}{2^{2n}-p(2-p)} \geq 1 \,,
\end{equation}
for all $n\,\geq1,\,  p\in[0,1]$, which is achieved with access to noisy Pauli gates.           
\end{proposition}

\begin{proof}
Ref.~\cite{takagi2020optimal} gives the optimal quasi-probability decomposition for the inverse noise channel as 
\begin{equation}
    \mathcal{D}^{-1} =  \left(1 + \frac{(2^{2n}-1)p}{2^{2n}(1-p)} \right)\mathcal{I} - \sum_{i=1}^{2^{2n}-1} \frac{p}{2^{2n}(1-p)}\mathcal{P}_i\,,
\end{equation}
where $\mathcal{I}$ is the identity channel and $\mathcal{P}_i$ is the Pauli channel corresponding to the $i$th Pauli tensor product. This has corresponding error mitigation cost 
\begin{equation}\label{eq:QP_gamma-global}
    \gamma = 
    \frac{2^{2n}-2p+p^2}{2^{2n}(1-p)^2}.
\end{equation}
Assuming perfect correction we have $\Delta \widetilde{C} = (1-p)\Delta C$ which implies
\begin{equation}
     \rchi_{depol} = \frac{2^{2n}}{2^{2n}-2p+p^2}\,,
\end{equation}
which is greater than or equal to 1 as $-2p + p^2 \leq0$ for all $0\leq p\leq1$.
\end{proof}

\subsubsection{Local depolarizing noise}

Here we consider a model of cost concentration due to a single instance of local depolarizing noise in a circuit. We presume that the concentration follows a similar form of scaling to global depolarizing noise and a tensor product of local depolarizing noise (see Eq.~\eqref{eq:costconcentration}). We show that, under this assumption, the relative resolvability has regimes of being greater than 1 or less than 1, depending on the strength of the cost concentration.

\begin{supproposition}[Relative resolvability of Probabilistic Error Cancellation with one instance of local depolarizing noise]\label{prop:QP_local-depol}
Consider a single instance of local depolarizing noise occurring with error probability $p$ acting at an arbitrary point in the parameterized circuit. Suppose that due to this noise channel we have
\begin{equation}\label{eq:QP_local-assumption}
    \langle\Delta \widetilde{C}(\boldsymbol{\theta}_{i,*})\rangle_i \geq (1-b_\alpha p) \langle\Delta C(\boldsymbol{\theta}_{i,*})\rangle_i
\end{equation}
for all $p\in [0,1]$ where $\langle\cdot\rangle_i$ denotes an average over all avaliable parameters and $b_\alpha$ where is some positive constant. Then the optimal quasiprobability scheme gives:
\begin{itemize}[leftmargin=*]
    \item for $b_\alpha\leq \frac{3}{4}$,
\begin{equation}
    \overline{\rchi}\leq 1\,,\quad \forall p \in [0,1],
\end{equation}
    \item for $\frac{3}{4}<b_\alpha\leq1$,
\begin{align}
    &\overline{\rchi}\leq 1 + \frac{1}{4}p(2-p)+\mathcal{O}(p^2)\,,\quad \forall p \in [0,1]\,, \\
    &\overline{\rchi}>1\,,\quad \forall p \in \Big(0,\,1-\sqrt[3]{3(b^{-1}-1)}\Big]\,,
\end{align}
    \item for $b_\alpha>1$,
\begin{equation}
    \overline{\rchi}> 1+\frac{p(2-p)}{4-p(2-p)}\,, \quad \forall p \in \big(0,\,1/b_\alpha\big]\,.
\end{equation}
\end{itemize}
\end{supproposition}

\begin{proof}
From Eq.~\eqref{eq:QP_gamma-global}, we can write the optimal error mitigation cost for one instance of local depolarizing noise acting on one qubit as
\begin{equation}\label{eq:QP_gamma-local}
    \gamma = 
    \frac{4-2p+p^2}{4(1-p)^2}\,.
\end{equation}
Now, due to our assumption \eqref{eq:QP_local-assumption} and assuming perfect implementation of the basis of noisy gates (leading to perfect correction of the noise) we have 
\begin{equation}
     \overline{\rchi} \leq \frac{4(1-p)^2}{(4-2p+p^2)(1-bp)^2}\,,
\end{equation}
and we denote the quantity on the right hand side as $h(p)$. Note that for any value of $b$, $h(p=0)=1$ and $h(p=1)=0$. The partial derivative can be found to satisfy
\begin{equation}\label{eq:QP_derivative}
    \frac{\partial h}{\partial p} \propto \left(1-p\right)\left(1-bp\right)\left(-p^3 + 3p^2 - 3p + \frac{1}{b} -4\left(\frac{1}{b}-1\right)\right)\,,
\end{equation}
where the proportionality factor we omit is positive for all $b\geq0$ and $p\in[0,1]$. The third bracket is a cubic form with discriminant 
\begin{equation}
    \Delta = \frac{108}{b^2}\left(-8b^2 + 11b- 4\right)\,,
\end{equation}
which is strictly negative for all $b$. Thus, the cubic form only has one real root and, inspecting its behaviour for large $p$, we can conclude it has negative gradient for all $p$. 
The cubic form has root at $p=0$ when $b=3/4$. More generally, the root can be found to take the form
\begin{equation}
    p'=1+\sqrt[3]{3(1-b^{-1})}\,.
\end{equation}
By evaluating the second derivative of $h(p)$, this root can be seen to correspond to a local maximum of $h(p)$. We now find the maximum value of $h(p)$ over the interval $p\in [0,1]$ for different regimes of cost concentration strength $b$.

First, we inspect the regime where $b\leq3/4$. In this case $p'\leq0$ and thus $\frac{\partial h}{\partial p}\leq 0$ for $p\in[0,1]$. Thus the maximum value of $h(p)$ on the interval $p\in[0,1]$ is $h(0)=1$. We can then conclude that $\overline{\rchi}\leq1$ with bound saturated in the limit of zero error probability.

Now consider the regime $3/4<b\leq1$. In this case $0< p' \leq1$ and $\frac{\partial h}{\partial p}> 0$ for small values of $p$. Specifically, it is clear that $\overline{\rchi} >1$ for $0<p\leq 1+\sqrt[3]{3(1-b^{-1})}$. The upper limit on $p$ can be raised, however, the exact interval is obtained by solving a quartic equation which we omit here as it is not very insightful. Moreover, the upper limit is tight in the limit $b\rightarrow 1$ and we obtain the result that when $b=1$, $\overline{\rchi} >0$ for all $p \in (0,1)$.    

Finally, consider the regime $b\geq1$. Now $\frac{\partial h}{\partial p}$ has a different root $p''=1/b$ due to the second bracket in \eqref{eq:QP_derivative}. Again, this can be shown to correspond to a maximum of $\overline{\rchi}$ and we can write $\overline{\rchi}>1$ for $0<p\leq1/b$.
\end{proof}

\begin{proposition}[Scaling of Probabilistic Error Cancellation with local depolarizing noise]
Consider tensor-product local depolarizing noise with local depolarizing probability $p$ acting in $L$ instances through a depth $L$ circuit as in Eq.~\eqref{eq:noisystate}. Suppose that the effect of this noise is to cause cost concentration 
\begin{equation}\label{eq:appdx_QP_localscaling-assumption}
    \langle\Delta \widetilde{C}(\boldsymbol{\theta}_{i,*})\rangle_i = Aq^L \langle\Delta C(\boldsymbol{\theta}_{i,*})\rangle_i\,,
\end{equation}
for some constant $A$ and noise parameter $q\in [0,1)$. The optimal quasiprobability method to mitigate the depolarizing noise in the circuit yields 
\begin{equation}
    \overline{\rchi} = \frac{1}{A^2q^{2L}}\left(Q(p) \right)^{nL} \,,
\end{equation}
where $0\leq Q(p)\leq 1$ for all $p$. Thus, the average relative resolvability has unfavourable scaling with system size.
\end{proposition}

\begin{proof}
As shown in Section \ref{sec:appdx_QP_gammatot}, error mitigation cost of multiple gates with probabilistic error cancellation is the product of the individual error mitigation costs. Thus, for the collection of gates considered, we have total error mitigation cost
\begin{equation}
    \gamma_{tot} = \left(\frac{4(1-p)^2}{4-2p+p^2} \right)^{nL}\,,
\end{equation}
where we have used Eq.~\eqref{eq:QP_gamma-local}. We suppose that mitigation perfectly corrects the error, such that $\Delta C_m(\thv_{i,*}) = \Delta C(\thv_{i,*})$. Combining this with our assumption \eqref{eq:appdx_QP_localscaling-assumption} we obtain the desired result, where we denote
\begin{equation}
    Q(p)= \frac{4-2p+p^2}{4(1-p)^2} = 1 - \frac{3p(2-p)}{4-p(2-p)}\,,
\end{equation}
which clearly satisfies $0\leq Q(p) \leq 1$.
\end{proof}

\subsection{Linear Ansatz Methods}

\subsubsection{Global depolarizing noise is exactly correctable}

Consider the linear ansatz
\begin{equation}
    C_m(\boldsymbol{a}) = a_1\widetilde{C} + a_2\,,
\end{equation}
where we denote $\boldsymbol{a} = (a_1,a_2)$. As shown in \cite{czarnik2020error}, this ansatz is particularly suited to global depolarizing noise and the ansatz can correct the noise exactly. Namely, the $n$-qubit noise channel
\begin{equation}
    \rho \xrightarrow[]{\mathcal{D}} (1-p)^L\rho + (1-(1-p)^L)\frac{\id}{2^n}
\end{equation}
can be exactly corrected by using
\begin{equation}
    a_1 = \frac{1}{(1-p)^L}\,,\quad a_2=-\frac{(1-(1-p)^L)}{(1-p)^L}\Tr[O]/2^n\,.
\end{equation}
As correction is exact, $\Delta C_m =\Delta C$. It can also be seen that $\Delta \widetilde{C} = (1-p)^L\Delta C$ and the error mitigation cost is $\gamma = 1/(1-p)^{2L}$. This gives $\rchi =1$ for any pair of cost function points. 

Note in this discussion we have neglected the shot burden of training. Whilst this may be significant and difficult to quantify for other noise channels, in the case of global depolarizing noise this is minimal as only two training data points are required and the ansatz is universal for any state.

\subsubsection{Relative resolvability between two points with same ansatz applied}

\begin{proposition}[Linear ansatz methods]
Consider any error mitigation strategy that mitigates noisy cost function value $\widetilde{C}(\boldsymbol{\theta})$ by constructing an estimator $C_m(\boldsymbol{\theta})$ of the form
\eqref{eq:CDR_ansatz}. For any two noisy cost function points to which the same ansatz is applied, we have
\begin{equation}
    \rchi = 1\,,
\end{equation}
for any noise process.
\end{proposition}

\begin{proof}
By applying the same ansatz of the form \eqref{eq:CDR_ansatz} to two noisy cost function points corresponding to parameter sets $\thv_1, \thv_2$, one can write
\begin{align}
    C_m(\thv_1,\boldsymbol{a}) = a_1\widetilde{C}(\thv_1) + a_2\,, \\
    C_m(\thv_2,\boldsymbol{a}) = a_1\widetilde{C}(\thv_2) + a_2\,.
\end{align}
This gives $\gamma = a_1$ and $\Delta C_m = a_1\Delta\widetilde{C} $. Thus, substituting these quantities into Definition \ref{def:resolvability} one obtains $\rchi=1$ as required.
\end{proof}

\section{Numerical simulations - implementation details  }\label{appendix:numerics}

\begin{figure}[t]
\centering
    \includegraphics[width=0.8\columnwidth]{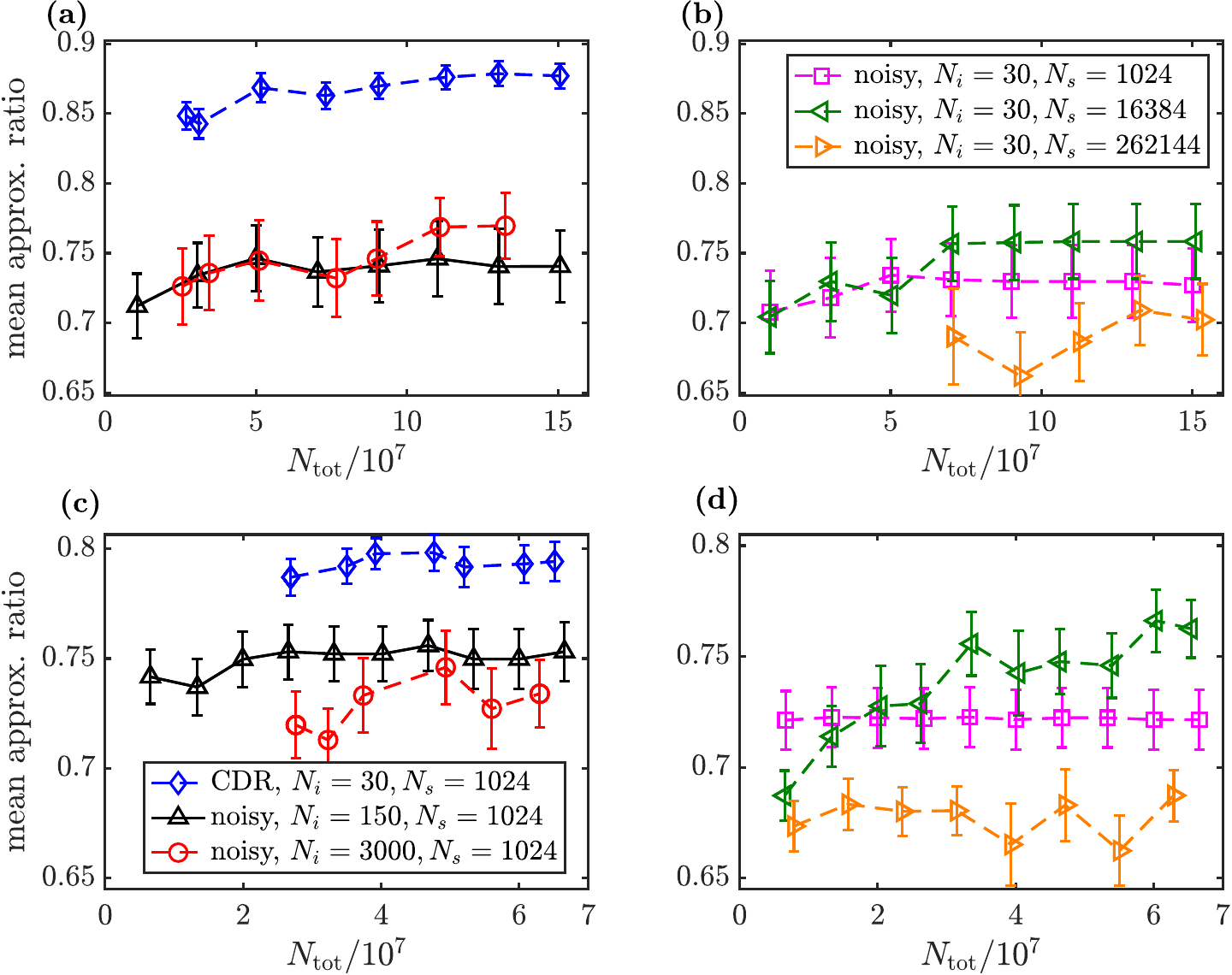}
    \caption{\textbf{Benchmarking various implementations of the noisy optimization for {5-qubit and 8-qubit} MaxCut QAOA {with $\boldsymbol{p=4}$}.} In \textbf{(a,b)}  we plot the approximation ratio averaged over $36$   $n=5$  Max-Cut graphs  chosen randomly from  the Erd\"os-R\'enyi ensemble  as a function of $N_{\rm tot}$. The error bars are computed as in Fig.~\ref{fig:CDR-numerics}. {In \textbf{(c,d)} we show in a similar manner the results for  $30$   $n=8$ random Max-Cut Erd\"os-R\'enyi graphs.}   We compare the results for various numbers $N_i$ of optimization instances initialized randomly and various numbers of shots $N_s$ per cost function evaluation. As a reference, we show the results of CDR optimization for $p=4$. {For the 5-qubit (8-qubit) case we have  {$N_{\rm tot}=10^7$ to $1.5\times10^8$  ( $N_{\rm tot}=1\times10^7$ to $7\times10^7$)} as in Fig.~\ref{fig:CDR-numerics} (Fig.~\ref{fig:CDR-numerics-n8}).}    Additionally, as in Fig.~\ref{fig:CDR-numerics} we use the approximation ratio computed with the exact energy to benchmark the optimization, and in the case of $N_i>1$ we choose as the result of optimization the best instance determined according to the optimized cost function. {The error bars are computed as in Fig.~\ref{fig:CDR-numerics}.} We consider various values of $N_s = 1024,\, 16384, \,262144$ for $N_i=30$ and various values of $N_i=30, \, 150, \, 3000$ for $N_s=1024$. For the 5-qubit case, we find that $N_i=3000, N_s=1024$ yields the best results although differences in quality between most of the noisy optimization implementations are relatively small in comparison to the CDR mitigated optimization.  {In the case of $n=8$ the best noisy results are obtained for $N_i=30, N_s=16384$, but again different choices of $N_i$ and $N_s$ lead to similar quality of the solutions.}       }
    \label{fig:num_details}
\end{figure}

We perform our optimizations using the MATLAB implementation of the Nelder-Mead algorithm~\cite{lagarias1998convergence}. For each MaxCut graph, we perform optimization independently for $N_i$ random choices of an initial simplex. We evaluate the cost function by performing perfect sampling of the simulated state with $N_s$ shots. After each iteration of the Nelder-Mead algorithm, we compute the total cost of optimization per graph $N_{\rm tot}$ by summing the shot budget spent for all $N_i$ instances of the optimization.  {To  analyze the convergence of results with $N_{\rm tot}$, as shown in Figs.~\ref{fig:CDR-numerics}, \ref{fig:num_details}, we take the optimization results after $N_{\rm tot}$ shots to be the best of $N_i$ instances  according to the optimized  cost function.} The optimization is terminated  for $n=5$ {($n=8$)} when  $N_{\rm tot}$ exceeds $1.5\times10^8$ {($7\times 10^7$)}.

\subsection{CDR-mitigated optimization }\label{appendix:CDR_numerics}

We perform CDR-mitigated optimization with $N_i=30$ and $N_s=1024$. We use training circuits constructed with a  non-Clifford gates projection algorithm of~\cite{lowe2020unified}. To construct the  training circuits we decompose $e^{i \gamma_j H_{\rm MaxCut}}$ to native gates of an IBM quantum computer using a decomposition from~\cite{abhijith2020quantum}.  In order to account for linear connectivity of the simulated device we use SWAP gates to implement $ e^{-i \gamma_j Z_k Z_l} $ for non  nearest-neighbors terms.    The training circuits contain $100$ near-Clifford circuits with at most $30$ non-Clifford gates. In the case of circuits with  fewer than $60$ non-Clifford gates, we construct training circuits with half of the non-Clifford gates replaced by Clifford gates. We evaluate the cost function for the training circuits using perfect sampling and $N_s=1024$ shots. We perform CDR mitigation for each 2-body  term of $H_{\rm MaxCut}$
independently. In general, in order to maximize the quality of the mitigation one should construct the training circuits independently for each new set of QAOA angles. Here, for the sake of shot efficiency, for each new set of parameters we compute the training set from scratch only if the 1-norm distance of its QAOA angles  $(\gamma_1,\beta_1, \gamma_2,\beta_2, \dots \gamma_p,\beta_p)$ from the closest point of a simplex is larger than $0.01$.  Otherwise, we use the CDR linear ansatz  for the closest point of the simplex.

For the noisy (unmitigated) optimization we benchmark various combinations of  
$N_i$ and $N_s$ values for {$n=5$, $p=4$}. In particular we consider increasing $N_s$ for $N_i=30$ and increasing $N_i$ for $N_s=1024$.   We gather the results in Fig.~\ref{fig:num_details}. {We  find that  for $n=5$ ($n=8$)  while using  $N_{\rm tot}$ ranging from {$10^7$ to $1.5\times10^8$ (from $10^7$ to $7\times10^7$)}, as considered in  Fig.~\ref{fig:CDR-numerics}~(\ref{fig:CDR-numerics-n8}), the best results are obtained for {$N_i=3000$, $N_s=1024$ ($N_i=30$, $N_s=16384$)}.  We use these values for the noisy optimization presented in Figs.~\ref{fig:CDR-numerics}~and~\ref{fig:CDR-numerics-n8}, respectively.} 

\begin{figure}[t]
\centering
    \includegraphics[width=0.9\columnwidth]{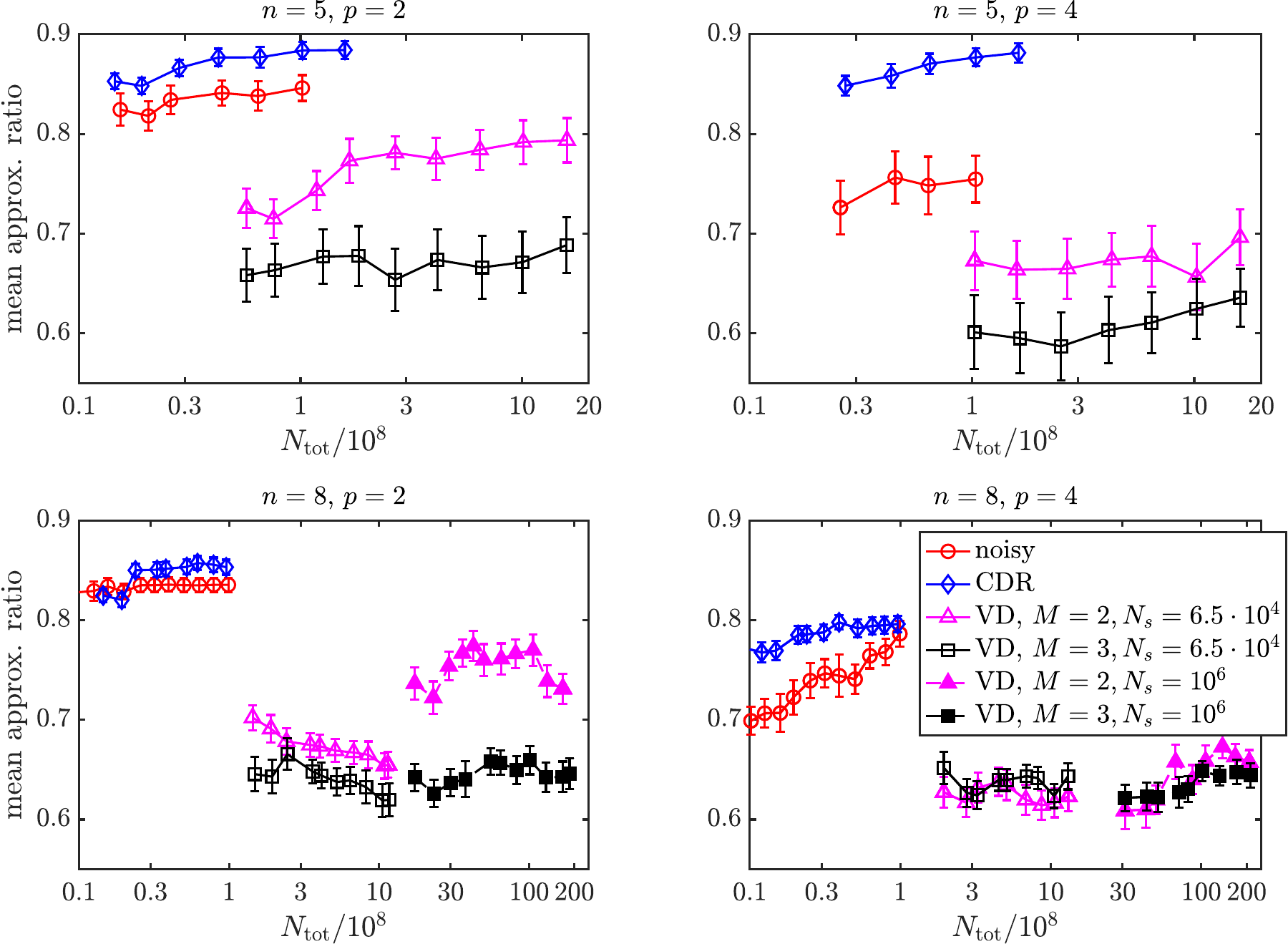}
    \caption{{\textbf{Virtual Distillation mitigated optimization for $\boldsymbol{p=2,4}$, 5-qubit  and 8-qubit MaxCut QAOA.} We plot the approximation ratio averaged over instances of Max-Cut graphs {randomly chosen from the} Erd\"os-R\'enyi ensemble  as a function of total shot number $N_{\rm tot}$. For $n=5$ ($n=8$) we choose $36$ ($30$) graphs. The results were obtained with $N_i=30$ initializations and $\widetilde{N}_s=65536$ shots per $\Tr[\widetilde{\rho}^M Z_i Z_j]$ and $\Tr[\widetilde{\rho}^M]$ estimation. For larger $n=8$ graphs we {also show} results obtained with $\widetilde{N}_s=10^6$, $N_i=30$. For reference we also present our results of CDR-mitigated and noisy optimization from {Figs.~\ref{fig:CDR-numerics}~and~\ref{fig:CDR-numerics-n8}. The error bars are computed as described in the caption of Fig.~\ref{fig:CDR-numerics}}. We observe that for this setting  the optimization with Virtual Distillation does not outperform the noisy or CDR-mitigated optimization.  }       }
    \label{fig:num_VD}
\end{figure}

\subsection{Optimization with Virtual Distillation  }\label{appendix:VD_numerics}

In this Appendix, we compare 5-qubit {and 8-qubit} MaxCut QAOA optimization of the VD-mitigated cost function with optimization of the noisy and CDR-mitigated cost function for $p=2,4$. We perform the comparison using the same randomly chosen graphs from the Erd\"os-R\'enyi ensemble as in Figs.~\ref{fig:CDR-numerics},~\ref{fig:CDR-numerics-n8}.  We perform VD mitigation for each expectation value of a 2-site term of $H_{\rm MaxCut}$ according to~\eqref{eq:VD_A}. Therefore, a key assumption is that we neglect derangement noise, which would affect realistic VD implementation on hardware~\cite{koczor2020exponential}. We use the Nelder-Mead algorithm  as described {above}.
We have $N_i=30$ as for CDR simulations from  Figs.~\ref{fig:CDR-numerics},~\ref{fig:CDR-numerics-n8} and assign {$\widetilde{N}_s=65536,10^6$} shots in order to estimate  $\Tr[\widetilde{\rho}^M Z_iZ_j]$ for each 2-site term of  $H_{\rm MaxCut}$ and  $\Tr[\widetilde{\rho}^M]$. Consequently, the total shot cost of the mitigated cost function estimation is $(n_e+1)\times \widetilde{N}_s$, where $n_e$ is the number of Max-Cut graph edges. We consider $M=2,\,3$ state copies  as  the shot cost of VD mitigation for given precision grows with increasing $M$~\cite{czarnik2021qubit} and $M=2,\,3$ was shown to be sufficient for typical applications~\cite{koczor2021dominant}. We find that {for this setup} $M=2$ gives better results than $M=3$ similar to   our analytical results.  
Here we allow for $N_{\rm tot}$  {up to $2\times10^9$ for $n=5$ and up to $2\times10^{10}$ for $n=8$, i.e.~up to $1$ and $2$ order of magnitudes more shots than considered  for CDR-mitigated and noisy optimization in Figs.~\ref{fig:CDR-numerics},~\ref{fig:CDR-numerics-n8}, respectively.    

We gather the results in Figs.~\ref{fig:num_VD} comparing them with noisy and CDR mitigated optimization from Figs.~\ref{fig:CDR-numerics},~\ref{fig:CDR-numerics-n8}.  We find that even with smaller $N_{\rm tot}$ the  noisy and CDR-mitigated optimization outperforms the VD-mitigated optimization. This example shows that even for  circuits outside of the NIBP regime there is no guarantee that using an error-mitigated cost function  leads to better performance than noisy cost function optimization. We note that  this result does not prohibit  VD-mitigated optimization advantage   for different choices of $N_{\rm tot},\, M$
or the shot number per cost function evaluation. 

\begin{figure}[t]
\centering
    \includegraphics[width=0.8\columnwidth]{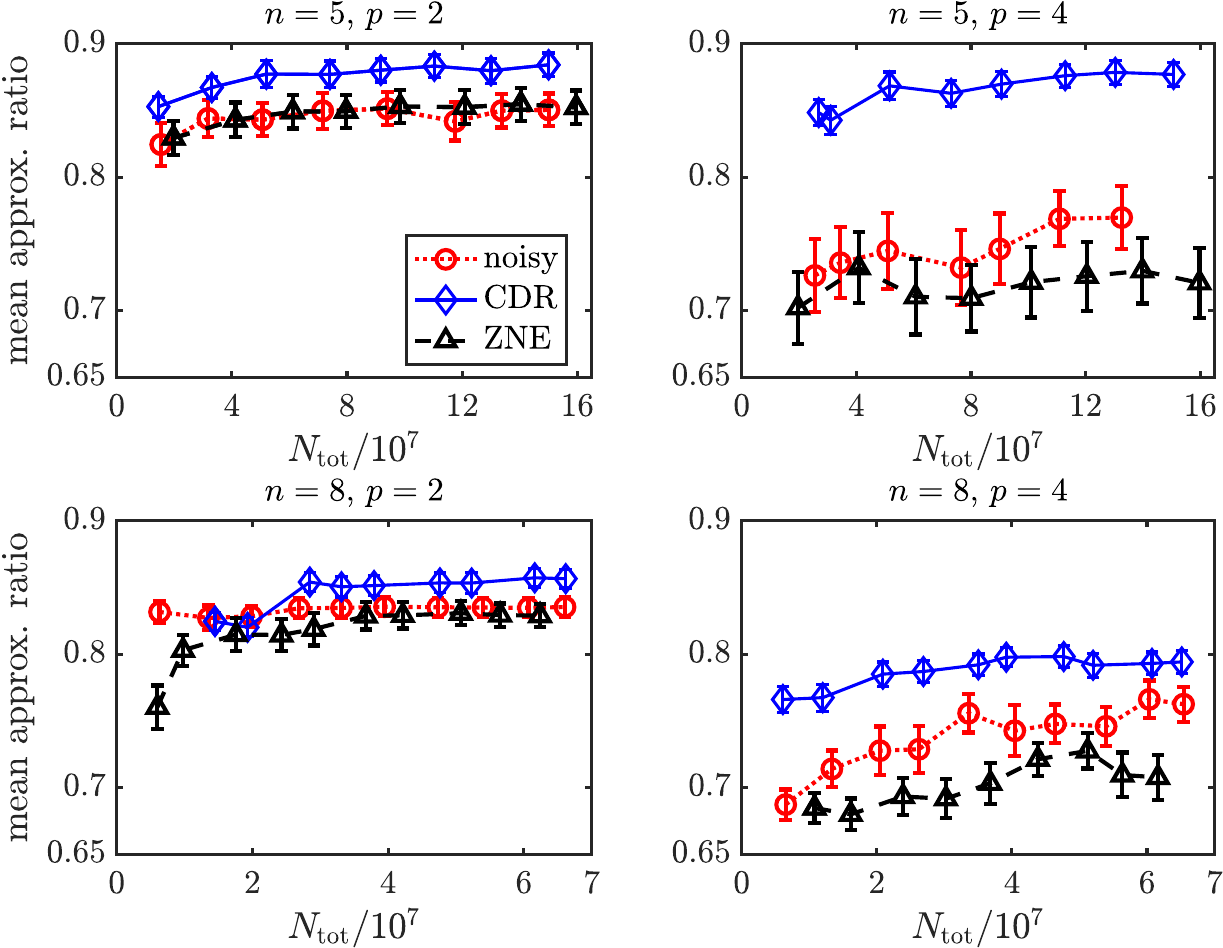}
    \caption{ {\textbf{Zero noise extrapolation mitigated optimization for $\boldsymbol{p=2,4}$, 5-qubit  and 8-qubit MaxCut QAOA.} We plot the approximation ratio averaged over instances of Max-Cut graphs {randomly chosen from the} Erd\"os-R\'enyi ensemble as a function of total shot number $N_{\rm tot}$. For $n=5$ ($n=8$), we choose $36$ ($30$) graphs. The results were obtained with $N_i=30$ initializations and $30000$ shots per ZNE-mitigated cost function evaluation. ZNE is performed by CNOT identity insertions with noise levels amplified by factors $a_0=1, a_1=3, a_2=5$, and with a linear extrapolation. The error bars are computed as in Fig.~\ref{fig:CDR-numerics}. For reference, we also present the results of CDR-mitigated and noisy optimization from Figs.~\ref{fig:CDR-numerics}, \ref{fig:CDR-numerics-n8}.  }       }
    \label{fig:num_ZNE}
\end{figure}

\subsection{Optimization with Zero-Noise Extrapolation (ZNE) }\label{appendix:ZNE_numerics}

{
Here, we analyze the optimization of ZNE-mitigated cost function for 5-qubit and 8-qubit MaxCut QAOA optimization. As in Figs.~\ref{fig:CDR-numerics},~\ref{fig:CDR-numerics-n8},~\ref{fig:num_VD} for CDR and VD,  we investigate number of rounds satisfying $p=2,4$. More precisely, we have used the same Erd\"os-R\'enyi  graphs as the ones in the benchmark simulations described above. We use the Nelder-Mead algorithm with $N_i=30$ initializations, the same as for the  CDR- and VD-mitigated optimization. For each noise level used to perform an extrapolation to the zero-noise limit, we evaluate the cost function with $N_s=10000$ shots.   Consequently, the shot cost of ZNE-mitigated cost function evaluation is $n_l N_s$ where $n_l$ is the number of noise levels. Here we investigate $N_{\rm tot}$ similar to $N_{\rm tot}$ for the noisy and CDR-mitigated optimization, i.e., $N_{\rm tot} = (2-16) \times 10^7$ for $n=5$, and $N_{\rm tot} = (1-7)\times 10^7$ for $n=8$. We have found that for considered here $N_{\rm tot}$ and $n=5$ values, our choice of $N_s$ leads to a better quality of the ZNE-mitigated optimization than  $N_s=1000$ and $N_s=100000$.
}

{
ZNE is performed by CNOT identity insertions, and we consider noise levels amplified by factors $a_0=1, a_1=3, a_2=5$, i.e., a CNOT gate in the original circuit is replaced by $a_i$ CNOTs~\cite{giurgica2020digital}. We use linear extrapolation to obtain the ZNE-mitigated expectation values for each term in the Hamiltonian. Such an extrapolation in the presence of more than two noise levels has been proposed to improve the robustness of ZNE results~\cite{giurgica2020digital} for realistic noise whose strength is challenging to scale accurately and has been applied in real-hardware ZNE implementations~\cite{kim2021scalable}. Furthermore, we find that for $n=5$ and $N_{\rm tot} = (1-10) \times 10^7$, using all three values of $a_l$ for linear extrpolation outperforms ZNE-mitigated optimization with  $a_0=1,\, a_1=3$. We also find that this choice outperforms an approach with  $a_0=1,\, a_1=3, \,a_2=5$ using quadratic extrapolation for the considered problem parameters.
We note that a detailed characterization of the effects of choice of the noise levels on performance of ZNE-mitigated optimization is beyond the scope of this work. However, we explore a range a hyperparameters in order to quickly gauge the power of a relatively simple ZNE approach in comparison to CDR and VD approaches analyzed above.    
}

{
We gather the results in  Fig.~\ref{fig:num_ZNE} plotting the approximation ratio averaged over graph instances versus $N_{\rm tot}$, the same as in Figs.~\ref{fig:CDR-numerics},~\ref{fig:CDR-numerics-n8},~\ref{fig:num_VD}, and comparing it to the noisy and CDR-mitigated results. For $p=2$, the ZNE-mitigated optimization gives results similar to the noisy one. For $p=4$, the ZNE-mitigated approximation ratios are slightly worse than the ones obtained by the noisy optimization. 
}

\end{document}